\documentclass[11pt]{article}
\usepackage{amsfonts,amsmath,amsthm,amssymb,mathtools,bbm}
\usepackage[margin=1in]{geometry}
\usepackage{enumitem}
\usepackage{xcolor}
\usepackage[normalem]{ulem}
\usepackage{tikz}
\usepackage{pgfplots}\pgfplotsset{compat=1.16}
\usepackage[colorlinks=true,allcolors=black]{hyperref}
\usepackage[capitalize,nameinlink]{cleveref}
\usepackage{titletoc}
\usepackage{comment}

\newtheorem{theorem}{Theorem}
\newtheorem{proposition}{Proposition}
\newtheorem{lemma}{Lemma}
\newtheorem{corollary}{Corollary}
\newtheorem{remark}{Remark}
\newtheorem{definition}{Definition}
\newtheorem{assumption}{Assumption}
\newtheorem{example}{Example}
\crefname{assumption}{Assumption}{Assumptions}
\Crefname{assumption}{Assumption}{Assumptions}
\crefname{algocf}{Algorithm}{Algorithms}
\Crefname{algocf}{Algorithm}{Algorithms}
\crefformat{enumi}{#2#1#3}
\Crefformat{enumi}{#2#1#3}

\usepackage[ruled]{algorithm2e}
\SetKwIF{If}{ElseIf}{Else}{if}{}{else if}{else}{end if}%
\SetKwFor{While}{while}{}{end while}%
\SetKwRepeat{Do}{do}{while}
\SetKw{KwGoTo}{go to}

\newcommand{\sep}{\,|\,}
\newcommand{\supp}{\operatorname{supp}}

\newcommand{\aff}{\operatorname{aff}}
\newcommand{\Msg}{\mathcal{M}}
\newcommand{\W}{\mathcal{W}}
\newcommand{\wT}{R}
\newcommand{\cups}{\cup \ldots \cup}
\newcommand{\prob}{\operatorname{Pr}}

\definecolor{suggestgreen}{RGB}{0,128,33}

\newcommand{\vtent}{v^{\mathrm{tent}}}
\definecolor{suggestpurple}{RGB}{128,0,160}
\definecolor{suggestred}{RGB}{180,0,0}
\definecolor{suggestorange}{RGB}{204,102,0}

\newcommand{\DEG}{\textup{PD}}

\usepackage[
            backend = biber,
            style = philosophy-classic,
            maxcitenames = 3,
            scauthorsbib = true,
            dashed = true,
            doi = false,
            isbn = false,
            url = false,
            eprint = false,
            backref = true
            ]{biblatex}
\renewbibmacro{in:}{}
\DeclareDelimFormat[cbx@textcite]{nameyeardelim}{\addspace}
\DeclareDelimFormat{nameyeardelim}{\addcomma\space}
\renewbibmacro*{pageref}{%
  \iflistundef{pageref}
    {}%
    {\setunit{\addperiod\space}%
     \printtext[backrefparens]{%
       \ifnumgreater{\value{pageref}}{1}
         {pp\adddot\ppspace}%
         {p\adddot\ppspace}%
       \printlist[pageref][-\value{listtotal}]{pageref}}\nopunct}}%

\let\cite\textcite

\begin{document}

\title{Coalition-Proof Disclosure\thanks{We thank Nageeb Ali, Aislinn Bohren, Steve Callander, Jeff Ely, Johannes H\"orner, Navin Kartik, Elliot Lipnowski, Jacopo Perego, and the audience members at the University of Pennsylvania, Penn State, QMUL-City Theory Workshop, Berlin Verifiability in Mechanism and Information Design Workshop, Gerzensee and MPSA. All errors are our own.}}
\author{Germ\'an Gieczewski\thanks{Princeton University. E-mail: \href{mailto:germang@princeton.edu}{germang@princeton.edu}.} \and Maria Titova\thanks{Vanderbilt University. E-mail: \href{mailto:motitova@gmail.com}{motitova@gmail.com}.}}
\date{October 13, 2026}
\maketitle

\begin{abstract}
We study generalized verifiable disclosure by a sender with type-independent preferences. We first characterize all perfect Bayesian equilibria as partitions of the sender's types into coalitions ordered by payoff. An equilibrium survives our coalition-proofness criterion if no group of types can profitably coordinate an announced deviation that only they are able and willing to join. Coalition-proof equilibria are greedy partitions: at each step, a coalition of the remaining types attains the highest payoff consistent with equilibrium. In our general setting, coalition-proof equilibria may fail to exist. We establish existence under restrictions on preferences and evidence: quasiconcavity with message completeness, betweenness, cheap talk with message completeness, and payoff degradation. For a disclosure game with maximally rich evidence, we show that the sender's ex-ante payoff in a coalition-proof equilibrium is a piecewise linear tent over the belief simplex, analogous to the concave closure in Bayesian persuasion and the quasiconcave closure in cheap talk.
\end{abstract}

\section{Introduction}\label{sec:intro}

Games of verifiable disclosure model communication by an informed sender who chooses which message to send from a set that depends on his private information. The canonical models \parencite{grossman1981,milgrom1981} deliver the classic unraveling result: skepticism toward silence forces full revelation. However, richer disclosure games in which only partial certification is possible or evidence is coarse have also been studied extensively in the literature \parencite{glazerrubinstein,sher2011,hagenbachkoessler,hart2017,BenPorathEtAl2019}. In these latter models, partial revelation in rich patterns is possible, but it typically comes with severe multiplicity of equilibria driven by how the receiver interprets messages nobody sends. Most of the literature responds by selecting receiver-optimal equilibria that often coincide with the receiver's commitment solution \parencite{glazerrubinstein,hart2017,BenPorathEtAl2019}.

This paper takes a different approach. We study disclosure games under one substantive assumption: the sender's preferences are type-independent. The setup is otherwise general: the sender's payoff depends on the receiver's posterior belief through a correspondence, and an arbitrary message mapping assigns each type a set of available messages. Rather than selecting equilibria on the receiver's behalf, we ask which equilibria survive credible collective deviations by the informed party. A perfect Bayesian equilibrium (PBE) fails our test of \emph{coalition proofness} if some group of sender types can announce a joint deviation (possibly mixed, possibly spread over several messages, possibly by reusing messages already sent on path) such that, when the receiver correctly anticipates exactly which types join the announcement, every participant gains. In the spirit of neologism proofness for cheap-talk games, the selection requires that credible coalitional deviations be correctly interpreted \parencite{Farrell1993}.

We begin by characterizing all PBE. We show that all PBE strategies, coalition-proof or not, are \emph{partition strategies}: they partition the type space into \emph{coalitions} ordered by equilibrium payoff from high to low. Each coalition bundles a set of types, evidence that only they can produce among the remaining types, and a common payoff. Equilibrium analysis thus reduces to a recursive partition problem implemented by an algorithm, which at each step removes one coalition from the game until no types are left. Coalition-proof PBE are exactly the outputs of the \emph{greedy} version of this algorithm in which, generally speaking, the coalition that attains the largest available payoff is removed at each step. 

In this general setting coalition-proof PBE may fail to exist for the same reason that neologism-proof equilibria may fail to exist in cheap-talk games. We thus give sufficient conditions for existence. If the message mapping is \emph{complete}, in the sense that the available evidence is closed under unions \parencite{bertomeucianciaruso}, quasiconcavity of the upper envelope suffices. \emph{Betweenness} \parencite{hart2017} of a single-valued payoff guarantees existence without any restrictions on the message mapping. Under either of these two conditions, if the game is generic, in the sense that distinct pooling groups earn distinct payoffs, the coalition-proof PBE partition is essentially unique. With cheap talk also available, completeness suffices without quasiconcavity. Finally, \emph{payoff degradation} guarantees existence whenever any coalition can be rebuilt with the same members and any lower payoff down to a common floor. This property holds under \emph{free disposal} of the sender's payoff, with no restriction on messages, or under cheap talk and \emph{revelation aversion}, a property indicating that the sender's worst outcome is having his type revealed.

Finally, we study the model with a maximally rich message space built on the stochastic message-mapping game of \textcite{titova}, where any combination of fractions of types can pool while excluding all others. There, a coalition-proof PBE takes a simple recursive form: the top coalition assembles the belief that maximizes the sender's payoff on the current support, and the recursion proceeds down the type space. The ex-ante payoff of this construction is a piecewise linear function of the prior that we call the \emph{tent}. The tent is the disclosure analogue of the concave closure of Bayesian persuasion \parencite{kamenicagentzkow} and the quasiconcave envelope of cheap talk \parencite{lipnowskiravid}: it lies below the former and generally coincides with neither.

\subsection*{Two Motivating Examples}

The first example shows that receiver-optimal equilibria can rest on implausibly asymmetric interpretations of deviations wanted equally by all sender types.

\begin{example}[Implausible revelation in receiver-optimal equilibria]\label{ex:intro-incumbent}
 A centrist incumbent (the sender) learns the state of the world $\theta\in\{-1,1\}$ drawn from a uniform prior; he can reveal it to the voter (the receiver) or not, choosing a message $m\in\{\theta,\varnothing\}$. The voter sees $m$ and elects a left-wing challenger ($a=-1$), a right-wing challenger ($a=1$), or the incumbent ($a=S$). The incumbent is office-motivated and obtains $\mathbbm{1}_{a=S}$. The voter receives $1$ if the elected challenger matches the state ($a=\theta$), $0$ if not, and $0.9$ from the safe option of reelecting the incumbent.
\end{example}

This game has two kinds of PBE. In the first, both types send $m=\varnothing$ and the incumbent is reelected. In the second, the incumbent is reelected with probability zero: the empty message is either off path and interpreted as coming disproportionately from one state, or it is sent disproportionately by one type on path. In any receiver-optimal equilibrium, the voter learns $\theta$ with probability one, and the voter's belief after $\varnothing$ places a probability of at most $0.1$ or at least $0.9$ on $\theta=1$: the voter extracts full revelation by interpreting $\varnothing$ asymmetrically, whether it is off path or on path, even though the two types have identical incentives to send it. For the incumbent, revealing the state is weakly dominated by silence: anything he says can be used against him. He could announce, ``It is in my interest to reveal nothing, and I benefit from this announcement regardless of the state; so hearing nothing, keep your prior.'' Both types strictly benefit from making this announcement, and under the correct interpretation both gain.\footnote{A refinement such as the intuitive criterion or D1 does not help here: the refinements give conditions under which to rule out some types as deviators, but in the example just given, the types are equally plausible deviators.} Coalition proofness eliminates every equilibrium except pooling on silence.

One may conjecture that limiting the receiver's ability to interpret deviations adversarially simply shifts surplus to the sender such that selecting ex-ante sender-optimal equilibria will achieve the same outcome. The second example shows it will not: coalition proofness is an interim concept, and profitable interim deviations can destroy ex-ante optimal pooling.

\begin{example}[Implausible ex-ante sender-optimal equilibria]\label{ex:intro-nuisance}
 Consider a \textcite{grossman1981} game with a nuisance dimension. Let $\Theta=\{(L,A),(L,B),(H,A),(H,B)\}$ with a uniform prior. Type $\theta$ can send any message $m\in2^\Theta$ with $\theta\in m$. The receiver's best response to belief $\mu$ is $a^*(\mu)=\sqrt{\prob_\mu(\theta(1)=H)}-8\bigl[\prob_\mu(\theta(2)=A)-0.5\bigr]^2$, and the sender's payoff is $a$.
\end{example}

Revealing $\theta(1)=H$ is good for the sender; revealing anything about the nuisance dimension $\theta(2)$ hurts him. One PBE outcome is full revelation supported by adversarial off-path beliefs. The ex-ante sender-optimal PBE has all types pool on $m=\Theta$ with payoff $\sqrt{0.5}$, supported by the same kind of asymmetric threat: off path, the receiver attributes the message $\{(H,A),(H,B)\}$ disproportionately to a single type. But the high types can announce exactly that message, evidence only they can produce, and both gain under the correct interpretation, meaning that the pooling equilibrium is not coalition-proof. In the coalition-proof prediction, the high types separate along the payoff-relevant dimension and pool along the nuisance dimension, mirroring the unique equilibrium outcome of the game with the nuisance dimension removed. When the sender wants to reveal some dimensions of the state and conceal others, coalition proofness delivers precisely that outcome.

\subsection*{Related Literature}

Our selection concept belongs to the family of belief-based refinements built on credible deviations: announcement proofness \parencite{matthews1991}, neologism proofness \parencite{Farrell1993}, the closely related equilibria of \textcite{GrossmanPerry1986}, and undefeated equilibrium \parencite{undefeated}. \textcite{callander} apply announcement proofness to select the sender-optimal equilibrium of cheap talk in complex environments. Neologism proofness and the criterion of \textcite{GrossmanPerry1986} allow deviations only to a single off-path message, announcement proofness lets the deviating types spread over several new messages, and undefeated equilibrium tests an equilibrium against alternative equilibria that use its off-path messages. In each case, the deviation uses only messages that are off the equilibrium path; our blocking coalitions may mix over several messages and use messages already on the equilibrium path, which is what makes the greedy characterization exact and existence obtainable under weaker conditions.

\textcite{bertomeucianciaruso} adapt neologism proofness to disclosure games. A Grossman--Perry--Farrell equilibrium (GPFE) is a pure-strategy PBE that admits no single-message self-signaling deviation. We allow mixed strategies, including the mixing that arises in truth-leaning equilibria. Their priority algorithm is analogous to our \cref{alg:greedy-partition}, but its termination is necessary but not sufficient for a GPFE to exist, and their existence result requires their convexity condition (CV) together with completeness (CP) or nested reports (NR). Under betweenness, a coalition-proof PBE exists without further restrictions on evidence (\cref{thm:two}). In their Example~7, the priority algorithm does not terminate and no GPFE exists, though a coalition-proof PBE does. Their model also allows message-dependent payoffs through disclosure costs or lemons discounts, which our belief-based payoff cannot represent.

\textcite{koesslerskreta} study information design by a privately informed designer and propose \emph{interim optimality} as a selection criterion. Their single-receiver setting with type-independent preferences is comparable to our rich disclosure benchmark (\cref{sec:rich-games}), as both permit arbitrary splittings of the prior. The timing of deviations differs, however: their designer evaluates a mechanism before its random message is realized, whereas in our rich game a sender observes his payoff-irrelevant label and knows whether he is included in a coalition. Their mechanisms let types commit to mixing over messages that yield different payoffs; sender mixing in our equilibria requires indifference.

A second strand, closer to our approach, characterizes equilibria of disclosure games directly. \textcite{rappoport} characterizes receiver-optimal equilibria. \textcite{wenhao} characterizes all PBE under the assumptions that each piece of evidence has cheap-talk copies and any two pieces of evidence can be jointly presented. Our \cref{prop:pbe-characterization} characterizes all PBE for an arbitrary message mapping.

Our paper also contributes to the literature on disclosure with structured evidence. 
\textcite{hart2017} study evidence games, which are disclosure games with a particular evidence structure and sender payoff. They propose an equilibrium refinement of truth leaning and show that it selects the receiver's commitment outcome. Their key payoff condition, betweenness, is one of our existence conditions. 
We show in \cref{sec:comp-hkp} that in evidence games every truth-leaning equilibrium outcome is also a coalition-proof PBE outcome. Thus, coalition proofness provides an alternative sender-side foundation for focusing on receiver-optimal outcomes in these games.\footnote{Likewise, \textcite{gao2026} shows that a version of announcement proofness selects the truth-leaning outcome in a continuous-type version of \citeauthor{hart2017}'s evidence-game framework with selective data disclosure and posterior-mean sender payoffs.} 
\textcite{glazerrubinstein}, \textcite{sher2011,sher2014}, and \textcite{BenPorathEtAl2019} develop the commitment side of persuasion with evidence, and \textcite{hagenbachkoessler} study disclosure among many players. \textcite{callanderjpe} select the sender-optimal equilibrium in a verifiable-disclosure model of expert advice, and \textcite{perego} test selective disclosure experimentally.

Finally, our benchmarks connect disclosure to the values of neighboring communication models. With commitment, the sender's value is the concave closure of \textcite{kamenicagentzkow}; with pure cheap talk, it is the quasiconcave envelope of \textcite{lipnowskiravid}, which coalition proofness selects as the unique prediction in the cheap-talk corner of our model. The tent of the rich-message benchmark \parencite{titova} is the disclosure counterpart of these objects.

The rest of the paper proceeds as follows. \Cref{sec:model} lays out the setting. \cref{sec:coalitions} develops coalitions and partitions. \Cref{sec:pbe} characterizes PBE. \Cref{sec:cppbe} introduces coalition proofness and characterizes coalition-proof PBE as greedy partitions. \Cref{sec:existence} presents the existence results and shows that truth-leaning equilibria in evidence games are coalition-proof. \Cref{sec:rich-games} solves the rich-message benchmark and constructs the tent. \Cref{sec:conclusion} concludes. Proofs are provided in the appendix.

\section{Model}\label{sec:model}

There are two players, a sender and a receiver. First, the sender learns his type $\theta \in \Theta$ drawn from a common prior $\mu^0 \in \Delta \Theta$.\footnote{For a finite set $S$, $\Delta S$ denotes the set of probability distributions over $S$.} Then, he sends a message $m \in \Msg$ from a non-empty set $M(\theta)$, where $M: \Theta \to 2^{\Msg} \smallsetminus \{ \varnothing \}$ is the message mapping. Finally, the receiver updates her belief and acts, and payoffs are realized. We employ the belief-based approach and abstract away from the receiver's action. The sender's payoff primitive is a correspondence $V: \Delta \Theta \rightrightarrows \mathbb{R}$; its envelopes are $\overline{v}(\mu) := \max V(\mu)$ and $\underline{v}(\mu) := \min V(\mu)$. Throughout, we refer to the sender's payoff simply as ``the payoff.'' In all sections except for \cref{sec:rich-games}, we maintain the following assumption:
\begin{assumption}\label{assumption:game}
\begin{enumerate}[label=(A\arabic*), ref=(A\arabic*)]
    \item \label{A1:finite-sets} $\Theta$ and $\Msg$ are finite.
    \item \label{A2:full-support} $\mu^0$ has full support.
    \item \label{A3:message-coverage} $\Msg = \bigcup_{\theta \in \Theta} M(\theta)$.
    \item \label{A4:V-uhc} $V$ is upper hemicontinuous with non-empty compact interval values.
\end{enumerate}
\end{assumption}
Here, \cref{A3:message-coverage} imposes that every message is feasible for at least one type. \cref{A4:V-uhc} follows from Berge's maximum theorem under mild conditions.\footnote{For example, suppose the receiver's action space $A$ is a non-empty compact Hausdorff space, the sender's payoff $u_S: A \to \mathbb{R}$ is continuous, and for each $\theta \in \Theta$, the receiver's payoff $u_R(\cdot, \theta) : A \to \mathbb{R}$ is continuous, and let $V(\mu)$ be the convex hull of the sender's payoffs $u_S(a)$ over the receiver's best responses $a$ to $\mu$. Then, among other properties, $V$ is non-empty, compact-valued, convex-valued, and upper hemicontinuous (see \cref{lem:best-response} in the online appendix).} Under \cref{A4:V-uhc}, the upper envelope $\overline{v}$ is upper semicontinuous, the lower envelope $\underline{v}$ is lower semicontinuous, and $V(\mu) = [\underline{v}(\mu), \overline{v}(\mu)]$.

Our main object of study is a disclosure game defined as follows:
\begin{definition}\label{dfn:game}
    A \emph{disclosure game} $G$ is a tuple $(\Theta,\Msg,M,\mu^0,V)$ satisfying conditions \cref{A1:finite-sets}--\cref{A4:V-uhc}.
\end{definition}

\paragraph{Strategies.} A sender strategy in game $G$ is a map $\sigma: \Theta \to \Delta \Msg$ such that $\supp \sigma(\cdot \sep \theta) \subseteq M(\theta)$. The \emph{evidence} of $\sigma$ is the set of messages it uses, 
\[
    X(\sigma) := \bigcup_{\theta \in\Theta} \supp \sigma(\cdot \sep \theta).
\]
Message $m$ is \emph{on path} under $\sigma$ if
\[
    p_\sigma(m) := \sum_{\theta \in \Theta} \mu^0(\theta) \sigma(m \sep \theta) > 0;
\]
since $\mu^0$ has full support by \cref{A2:full-support}, the set of on-path messages is exactly $X(\sigma)$. For $m \in X(\sigma)$, the induced belief $\mu_\sigma(\cdot \sep m) \in \Delta \Theta$ is
\[
    \mu_\sigma (\theta \sep m) := \frac{\mu^0(\theta) \sigma(m \sep \theta)}{p_\sigma(m)}.
\]

\paragraph{Restricted games.} Let $\wT$ be a non-empty subset of $\Theta$. A restricted game on $\wT$ is
\[
    G \vert_{\wT} := (\wT, \Msg_{\wT}, M_{\wT}, \mu^0_{\wT}, V_{\wT}),
\]
where $\Msg_{\wT} := \bigcup_{\theta \in \wT} M(\theta) \subseteq \Msg$ is the set of all messages available to types in $\wT$; $M_{\wT}(\theta) := M(\theta)$, treated as a subset of $\Msg_{\wT}$; $\mu^0_{\wT} \in \Delta \wT$ is the conditional prior defined as
\[
    \mu^0_{\wT}(\theta) := \frac{\mu^0(\theta)}{ \mu^0(\wT) }, \text{ where } \mu^0(\wT) := \sum_{\theta \in \wT} \mu^0(\theta);
\]
and $V_{\wT}$ is the sender's payoff correspondence restricted to beliefs supported on $\wT$.\footnote{For finite sets $S \subseteq S'$, let $\iota_{S,S'}: \Delta S \to \Delta S'$ denote the zero-extension map from $S$ to $S'$: for each $\nu \in \Delta S$ and $x \in S'$, $\iota_{S,S'}(\nu)(x)=\nu(x)$ if $x \in S$ and $\iota_{S,S'}(\nu)(x)=0$ otherwise. Thus, $V_{\wT} = V \circ \iota_{\wT,\Theta}$.}

Any restricted game $G \vert_{\wT}$ is itself a disclosure game, $G \vert_{\Theta} = G$, and restriction is transitive: for non-empty $\wT' \subseteq \wT \subseteq \Theta$, we have $(G \vert_{\wT})\vert_{\wT '} = G \vert_{\wT '}$ (see \cref{lem:restriction} in the online appendix). All the notions and arguments below are described for an arbitrary disclosure game $G$ and therefore apply to every restricted game. When we consider a restricted game $G \vert_{\wT}$, we add the subscript $\wT$.

\paragraph{Useful notation.} Write $n:=|\Theta|$ for the number of types and $\delta_\theta\in\Delta\Theta$ for the point mass on $\theta$. Given a non-empty set of types $\wT \subseteq \Theta$ and messages $X \subseteq \Msg$, the set of types in $\wT$ that can send at least one message in $X$ is the preimage of $X$ relative to $\wT$,
\[
    M^{-1}_{\wT}(X) := \{ \theta \in \wT \sep M(\theta) \cap X \neq \varnothing \},
\]
which we abbreviate as $M^{-1}(X)$ when there is no possibility of confusion. We write $
    P(m) := M^{-1}(\{ m \})$
for the set of types that can send message $m$. Note that $P(m)$ is non-empty for each $m \in \Msg$ by \cref{A3:message-coverage}.
Since $\mu^0$ has full support, $\mu^0_{\wT}$ is supported exactly on $\wT$. We identify distributions on $\wT$, including posterior beliefs induced in $G\vert_{\wT}$, with their zero extensions to $\Delta\Theta$ whenever they are used in the original game; similarly, we identify distributions on $\Msg_{\wT}$ with their zero extensions to $\Delta\Msg$.

\section{Coalitions and Partitions}\label{sec:coalitions}

Fix a disclosure game $G = (\Theta,\Msg,M,\mu^0,V)$ and a non-empty set of types $C \subseteq \Theta$. A \emph{coalition strategy on $C$} is a sender strategy $\sigma$ of the restricted game $G \vert_C$. The posterior beliefs induced by $\sigma$ are the beliefs $\mu_\sigma(\cdot \sep m)$ already defined for strategies, which are computed in the restricted game $G\vert_C$ and zero-extended to $\Delta\Theta$ when used in the original game. Since game restriction is transitive, these beliefs depend only on $C$ and $\sigma$, not on which larger (restricted) game contains $C$.

\begin{definition}\label{dfn:coalition}
    A \emph{coalition} of a disclosure game $G$ is a triple $(C,\sigma,w)$ such that
    \begin{enumerate}[label=(C\arabic*), ref=(C\arabic*)]
    \item \label{C1} $C \subseteq \Theta$ is non-empty;
    \item \label{C2} $\sigma$ is a coalition strategy on $C$;
    \item \label{C3} the evidence of $\sigma$ is exclusive to $C$:
    \[
        M^{-1}(X(\sigma)) \subseteq C;
    \]
    \item \label{C4} every message in $X(\sigma)$ yields the same sender payoff:
    \[
        w \in V( \mu_\sigma (\cdot \sep m) ) \text{ for every } m \in X(\sigma).
    \]
\end{enumerate}
\end{definition}
A coalition consists of the participating types, their strategy, and their common payoff. The coalition payoff is constant across all messages in the coalition's evidence and hence across messages sent by each coalition type. Note that the reverse inclusion $C \subseteq M^{-1}(X(\sigma))$ holds automatically since each $\theta \in C$ uses messages from $M(\theta) \cap X(\sigma)$; consequently, condition \cref{C3} becomes $M^{-1}(X(\sigma)) = C$ and states that the types participating in the coalition are exactly the types that can produce the coalition's evidence. The beliefs in \cref{C4} are well-defined because every $m \in X(\sigma)$ is on path.

\begin{remark}\label{remark:coalitions-exist}
    Fix a disclosure game $G$ and a message $m \in \Msg$, and let $\sigma^m$ be a coalition strategy on $P(m)$ under which every type that has access to message $m$ sends it with probability one. Then for every $w \in V( \mu^0_{P(m)})$, the triple $(P(m),\sigma^m,w)$ is a coalition of $G$. In particular, every disclosure game admits a coalition.
\end{remark}

\begin{definition}
    A \emph{partition} of $G$ is a finite sequence $        \Pi = \{ (C_t,\sigma_t,w_t) \}_{t=1}^T$ 
    such that $C_1, \ldots, C_T$ are pairwise disjoint, non-empty, and $\Theta = \bigcup_{t=1}^T C_t$, and for each $t$, the triple $(C_t,\sigma_t,w_t)$ is a coalition in the residual game $G_t := G\vert_{R_t}$, where $R_t := \bigcup_{s=t}^T C_s$. 
\end{definition}

We refer to $t$ as the \emph{step} and $C_t$ as the \emph{cell}, and we write $X_t := X(\sigma_t)$ for the evidence used at step $t$. For a partition $\Pi$, we call $R_t$ the \emph{residual set} at step $t$. More generally, for any object $Y_R$ attached to the restricted game $G\vert_R$, we abbreviate $Y_{R_t}$ as $Y_t$; thus $G_t=G\vert_{R_t}$, $\Msg_t=\Msg_{R_t}$, and so on. 

All partitions arise from the recursion formalized by the following algorithm:
\begin{algorithm}[ht!]
    \caption{Partition Algorithm}\label{alg:partition}
    Let $t := 1$ and $R_1 := \Theta$\;
    \While{$R_t \neq \varnothing$}{
        select a coalition $(C_t,\sigma_t,w_t)$ of $G_t$\;
        let $R_{t+1} := R_t \smallsetminus C_t$ and $t := t+1$\;
        \medskip
    }
\end{algorithm}

A \emph{run} of the algorithm is one of its possible executions. It has one selection per step. The algorithm \emph{terminates} when the type space is exhausted; its \emph{output} is the selected sequence $\{ (C_t,\sigma_t,w_t) \}_{t=1}^T$. At every step of every run, a selection exists by \cref{remark:coalitions-exist}, each step removes a non-empty cell (hence at least one type), and each output satisfies the definition of a partition. 
\begin{remark}\label{remark:partitions-exist}
    Every disclosure game admits a partition, and every partition has at most $n$ steps.
\end{remark}

By transitivity of game restriction, $G_t\vert_{C_t} = G\vert_{C_t}$, so $\sigma_t$ is simply a coalition strategy on $C_t$ with beliefs $\mu_{\sigma_t}$ that do not depend on the larger game in which the coalition is viewed. Similarly, the common-payoff condition \cref{C4} of the step $t$ coalition is the same whether viewed in $G_t\vert_{C_t}$, $G\vert_{C_t}$, or $G$. The only condition that uses the residual set $R_t$ is exclusivity \cref{C3}: no type in $R_t \smallsetminus C_t$ can produce evidence in $X_t$. However, types in $C_1 \cups C_{t-1}$ may be able to produce such evidence.

\begin{definition}
    Given a partition $\Pi = \{ (C_t,\sigma_t,w_t) \}_{t=1}^T$ of $G$, its associated sender strategy $\sigma^{\Pi}: \Theta \to \Delta \Msg$ is obtained by letting every $\theta \in C_t$ use $\sigma_t(\cdot \sep \theta)$ zero-extended to $\Delta \Msg$. A sender strategy is a \emph{partition strategy} if it is associated with some partition. 
\end{definition}
Note that $\sigma^\Pi$ is indeed a sender strategy of $G$: if $\theta\in C_t$, then
\[
    \supp\sigma^\Pi(\cdot\sep\theta)
    =
    \supp\sigma_t(\cdot\sep\theta)
    \subseteq M(\theta).
\]
Note, too, that distinct steps use disjoint evidence. If, on the contrary, $m\in X_s\cap X_t$ with $s<t$, then some type in $C_t\subseteq R_s$ can send a message in $X_s$, contradicting the exclusivity of step $s$ coalition in $G_s$. Thus, whenever
$m\in X_t$, only types in $C_t$ send $m$ under $\sigma^\Pi$; in particular, $m$ is on path.

\section{Perfect Bayesian Equilibria}\label{sec:pbe}

To define and characterize PBE of disclosure games, we introduce the concepts of feasible beliefs, skeptical payoffs, and individually rational partitions.

\begin{definition}\label{dfn:set-feasible-beliefs}
    For message $m \in \Msg$, the \emph{set of feasible beliefs} is
    \[
        \mathcal{F}(m) := \{ \mu_\sigma (\cdot \sep m) \sep \sigma \text{ is a sender strategy of } G \text{ with } m \in X(\sigma) \}.
    \]
\end{definition}
Given message $m$, feasible beliefs are those that arise, via Bayes' rule, under some sender strategy that uses that message as evidence. In other words, these are the beliefs that are consistent with disclosure: if message $m$ cannot be sent by some types, then all feasible beliefs have zero mass on those types; if some type can only send message $m$, then, among other things, all feasible beliefs after that message have positive mass on that type. Each $\mathcal{F}(m)$ is a non-empty compact convex polytope parametrized explicitly by the mixing weights of the types that can send $m$ (see \cref{lem:feasible-polytope} in the online appendix).

\begin{definition}\label{dfn:skeptical-payoff}
    For a type $\theta \in \Theta$, the \emph{skeptical payoff} is
    \[
        \underline{u}(\theta) := \max_{m \in M(\theta)} \min_{\mu \in \mathcal{F}(m)} \underline{v}(\mu).
    \]
\end{definition}
 The skeptical payoff is the payoff type $\theta$ obtains from his best message when, at every message, the receiver holds the least favorable feasible belief and breaks ties against the sender.\footnote{The maximum and minimum are attained because $\underline{v}$ is lower semicontinuous and $\mathcal{F}(m)$ is non-empty and compact (see \cref{lem:feasible-polytope} in the online appendix).} Throughout, the skeptical payoff $\underline{u}(\theta)$ is the one computed in $G$ itself. We say that a partition is individually rational if every type receives at least the skeptical payoff.

\begin{definition}\label{dfn:IR}
    A partition $\Pi = \{ (C_t,\sigma_t,w_t) \}_{t=1}^T$ of $G$ is \emph{individually rational (IR)} if
    \begin{equation}\label{eqn:IR}
        w_t \geq \underline{u}(\theta) \quad \text{ for every } t \text{ and every } \theta \in C_t. \tag{IR}
    \end{equation}
\end{definition}

Next we define PBE strategies and equilibrium.

\begin{definition}\label{dfn:PBE}
    A sender strategy $\sigma: \Theta \to \Delta \Msg$ is a \emph{PBE strategy} if there exist a belief system $\mu: \Msg \to \Delta \Theta$ and a sender payoff function $r: \Msg \to \mathbb{R}$ such that
    \begin{itemize}
        \item[1.] beliefs are feasible:
        \item[] \quad \quad \quad\quad \quad \quad \quad \quad \quad $\mu(\cdot \sep m) \in \mathcal{F}(m)$ for every $m \in \Msg$;
        \item[2.] beliefs are Bayesian on path:
            \item[] \quad \quad\quad \quad \quad\quad \quad \quad $\mu(\cdot \sep m) = \mu_\sigma(\cdot \sep m)$   for every $m \in X(\sigma)$;
        \item[3.] the realized payoff is compatible with the payoff correspondence:
\item[] \, \quad \quad\quad \quad \quad\quad \quad \quad $r(m) \in V(\mu(\cdot \sep m))$  for every $m \in \Msg$;
        \item[4.] the sender is sequentially rational:
\item[] \quad \quad \quad\quad \quad \quad$\supp \sigma(\cdot \sep \theta) \subseteq \arg\max_{m \in M(\theta)} r(m) $  for every $\theta \in \Theta$.
    \end{itemize}


    A triple $(\sigma,\mu,r)$ satisfying these four conditions is a \emph{PBE}.
\end{definition}
The first condition requires that even the off-path beliefs are feasible, that is, consistent with disclosure/rationalizable by some sender strategy. The function $r$ incorporates the receiver's tie-breaking and mixing: after each message $m$, the receiver selects the sender's payoff from the set $V(\mu(\cdot\sep m))$.
Every disclosure game admits a PBE (see \cref{lem:pbe-existence} in the online appendix).

If $(\sigma,\mu,r)$ is a PBE, then sequential rationality makes $r$ constant on $\supp \sigma(\cdot \sep \theta)$, equal to type $\theta$'s equilibrium payoff in that PBE, $u(\theta) := \max_{m \in M(\theta)} r(m)$.

The PBE characterization below rests on two simple observations. 
First, the partition strategy preserves the beliefs generated by its cells: if $m\in X_t$, then under $\sigma^\Pi$ only types in $C_t$ send $m$, so the Bayes posterior after $m$ is the posterior generated by $\sigma_t$ in $G\vert_{C_t}$, viewed as a belief in $\Delta\Theta$. Second, skeptical payoffs are the relevant off-path benchmark.
In any PBE, type $\theta$ receives at least $\underline u(\theta)$; conversely, after any unused message, the receiver can select a feasible belief and a compatible payoff that give no type more than its skeptical payoff. 
\begin{proposition}\label{prop:pbe-characterization}
    A sender strategy $\sigma$ is a PBE strategy if and only if $\sigma$ is associated with an IR partition $\Pi = \{ (C_t,\sigma_t,w_t) \}_{t=1}^T$ with $w_1 \geq \ldots \geq w_T$. 
    Furthermore, every PBE strategy is associated with an IR partition with strictly decreasing payoffs $w_1 > \ldots > w_T$.
\end{proposition}

Intuitively, in any PBE, we can group sender types by payoff in a strictly decreasing order and construct a coalition for each of those payoffs. Conversely, any IR partition with non-increasing payoffs can be supported as a PBE: on-path messages receive the cell payoff, while off-path messages are assigned skeptical feasible beliefs and payoffs. A type cannot profitably deviate off path because its cell payoff is at least its skeptical payoff; it cannot profitably deviate to earlier evidence because that evidence is unavailable in the relevant residual game; and it does not wish to deviate to later evidence because later cells receive weakly lower payoffs.

Call an IR partition with non-increasing payoffs a \emph{PBE partition}. \cref{alg:pbe-partition} returns all PBE partitions: it is simply the partition algorithm that imposes PBE restrictions (i.e., IR and non-increasing payoffs) at each step.

\begin{algorithm}[ht!]
    \caption{PBE Partition Algorithm}\label{alg:pbe-partition}
    Let $t := 1$, $R_1 := \Theta$, and $w_0 := \infty$\;
    \While{$R_t \neq \varnothing$}{
        select a coalition $(C_t,\sigma_t,w_t)$ of $G_t$ with
        $w_t \in \left[\max_{\theta\in R_t}\underline{u}(\theta),\,w_{t-1}\right]$\;
        \textbf{halt} with no output if there is none\;
        let $R_{t+1} := R_t \smallsetminus C_t$ and $t := t+1$\;
        \medskip
    }
\end{algorithm}

Note that \cref{alg:pbe-partition} uses a stronger (but ultimately equivalent) version of IR: at step $t$, it imposes $w_t\geq\max_{\theta\in R_t}\underline{u}(\theta)$, rather than $w_t\geq\max_{\theta\in C_t}\underline{u}(\theta)$.
For partitions with non-increasing payoffs, this stronger condition does not change the set of resulting partitions, as shown in \cref{lem:strong-IR} below. 
The stronger condition is easier to work with recursively: the admissible interval $\bigl[\max_{\theta\in R_t}\underline{u}(\theta),\,w_{t-1}\bigr]$ depends only on the residual set $R_t$ and the previous cell's payoff $w_{t-1}$, not on which candidate coalition $(C_t,\sigma_t,w_t)$ is being considered.
\begin{lemma}\label{lem:strong-IR}
    Let $\Pi=\{(C_t,\sigma_t,w_t)\}_{t=1}^T$ be a partition of $G$ with $w_1\geq\cdots\geq w_T$. Then $\Pi$ is IR if and only if
    $ w_t\geq \max_{\theta\in R_t}\underline{u}(\theta)$ for every $t$.
\end{lemma}
\begin{proof}
    $(\Longleftarrow)$ is immediate as $C_t\subseteq R_t$. Conversely, if $\Pi$ is IR and $\theta\in R_t$, then $\theta\in C_s$ for some $s\geq t$, so $\underline{u}(\theta)\leq w_s\leq w_t$ by IR and non-increasing payoffs.
\end{proof}

Unlike the unrestricted partition algorithm, \cref{alg:pbe-partition} may halt because no admissible coalition payoff is available at some step. If it terminates, its output is a PBE partition. Since every disclosure game admits a PBE, at least one run of \cref{alg:pbe-partition} terminates and returns an associated PBE partition.

\section{Coalition-Proof PBE}\label{sec:cppbe}

Fix a disclosure game $G=(\Theta,\Msg,M,\mu^0,V)$. For a partition $\Pi=\{(C_t,\sigma_t,w_t)\}_{t=1}^T$ of $G$ and a payoff $\widetilde w\in\mathbb{R}$, let
\[
    \Theta(\widetilde w)
    :=
    \bigcup_{t\,:\,w_t<\widetilde w} C_t
\]
denote the set of types whose payoff in the partition is strictly below $\widetilde w$.

\begin{definition}\label{dfn:blocking}
    Let $\Pi$ be a partition of $G$ and $\widetilde w\in\mathbb{R}$.
    \begin{enumerate}
        \item[(i)] $(\widetilde C,\widetilde\sigma,\widetilde w)$ is a
        \emph{blocking coalition} of $\Pi$ if
        $\Theta(\widetilde w)\neq\varnothing$ and
        $(\widetilde C,\widetilde\sigma,\widetilde w)$ is a coalition of
        the restricted game $G\vert_{\Theta(\widetilde w)}$.
        \item[(ii)] $\Pi$ is \emph{coalition-proof} if it admits no
        blocking coalition.
    \end{enumerate}
    A \emph{coalition-proof PBE partition} is a PBE partition that is
    coalition-proof. A sender strategy is a \emph{coalition-proof PBE
    strategy} if it is associated with some coalition-proof PBE partition.
\end{definition}

Because the payoff correspondence $V$ may be multi-valued, the realized payoffs $w_t$ are part of the equilibrium outcome rather than pinned down by the strategy; coalition proofness is accordingly a property of the partition.

A blocking coalition is an announced deviation: the types in $\widetilde C$ declare that they will switch to strategy $\widetilde\sigma$. Under the beliefs induced by this announcement, some selection of the receiver's best responses gives every participant payoff $\widetilde w$, strictly above his equilibrium payoff.
The restriction to $G\vert_{\Theta(\widetilde w)}$ captures credibility: only types whose partition payoff is strictly below $\widetilde w$ will want to join or imitate the announcement.
Among those types, the announced evidence must be exclusive to the deviators and yield the common payoff $\widetilde w$. 

Coalition proofness allows the announced deviation to involve multiple messages at once since it may involve a mixed strategy. Furthermore, the announced deviation may include messages that are already on the equilibrium path. This causes no ambiguity about who deviates: if the partition has non-increasing payoffs, as every PBE partition does, then every on-path user of such a message participates in the announced deviation.

Because the greedy selection defined below picks, at each step, the largest available coalition payoff, we introduce notation for the payoffs that coalitions of a restricted game can attain. For any non-empty $R\subseteq\Theta$, write
\[
    \W_R
    :=
    \{w\in\mathbb{R}\sep
    (C,\sigma,w)\text{ is a coalition of }G\vert_R
    \text{ for some }C\subseteq R\text{ and some }\sigma\},
\]
which denotes the set of coalition payoffs attainable in the restricted game $G\vert_R$. This set is non-empty and compact (see \cref{lem:coalition-compact} in the appendix).

\begin{definition}\label{dfn:greedy}
    A partition $\Pi=\{(C_t,\sigma_t,w_t)\}_{t=1}^T$ of $G$ is
    \emph{greedy} if, for every $t$ (with $w_0:=\infty$),
    \[
        w_t
        =
        \max\Bigl(
        \W_t
        \cap
        \bigl[\max_{\theta\in R_t}\underline{u}(\theta),\,
        w_{t-1}\bigr]
        \Bigr),
    \]
    where $\W_t$ is the coalition payoff set of $G_t$.
\end{definition}

Greedy partitions are outputs of \cref{alg:greedy-partition} defined below. The greedy partition algorithm is essentially the PBE partition algorithm with a stronger selection rule: at each residual game $G_t$, the payoff must satisfy the PBE bounds $\max_{\theta\in R_t}\underline{u}(\theta)\leq w_t\leq w_{t-1}$, and among the payoffs in $\W_t$ satisfying them it must be the largest. If no such payoff exists, the run halts.

\begin{algorithm}[ht!]
    \caption{Greedy Partition Algorithm}\label{alg:greedy-partition}
    Let $t := 1$, $R_1 := \Theta$, and $w_0 := \infty$\;
    \While{$R_t \neq \varnothing$}{
        select a coalition $(C_t,\sigma_t,w_t)$ of $G_t$ with
        $w_t=\max\bigl(\W_t\cap\bigl[\max_{\theta\in R_t}\underline{u}(\theta),\,w_{t-1}\bigr]\bigr)$\;
        \textbf{halt} with no output if there is none\;
        let $R_{t+1} := R_t \smallsetminus C_t$ and $t := t+1$\;
        \medskip
    }
\end{algorithm}

\begin{proposition}\label{prop:cppbe-characterization}
    A partition of $G$ is a coalition-proof PBE partition if and only if it is greedy. Consequently, a sender strategy is a coalition-proof PBE strategy if and only if it is associated with some greedy partition.
\end{proposition}

\cref{prop:cppbe-characterization} says that the no-blocking condition is exactly the greedy selection rule.
Suppose that at step $t$ the residual game $G_t$ contains a coalition that can do better than the payoff actually assigned at that step without exceeding the payoff of the previous cell.
Such a payoff is attractive exactly to the types remaining in the residual set: earlier cells already receive at least as much, whereas every type in $R_t$ receives weakly less than $w_t$ and hence strictly less than the proposed payoff. 
This coalition is therefore a credible announced deviation, hence a blocking coalition.
Greediness rules out precisely these deviations by assigning, at each residual game, the largest coalition payoff compatible with the IR lower bound and the previous cell's payoff cap.

As might be expected from the literature on neologism proofness, coalition-proof PBE do not always exist, meaning that \cref{alg:greedy-partition} may halt no matter what choices are made at each step. However, as we show next, they exist under relatively weak conditions, covering many settings studied previously in the disclosure literature.

\section{Existence of Coalition-Proof PBE}\label{sec:existence}

Fix a disclosure game $G=(\Theta,\Msg,M,\mu^0,V)$. By \cref{prop:cppbe-characterization}, a coalition-proof PBE exists if and only if some run of \cref{alg:greedy-partition} terminates. This section gives sufficient conditions for existence: first on the upper envelope $\overline{v}$ and the message mapping $M$ (\cref{sec:existence-qc-mc}), then for the single-valued payoff only (\cref{sec:existence-betweenness}), then for the message mapping only (\cref{sec:existence-cheaptalk}), and finally through payoff degradation, which constrains only the lower envelope under free disposal (\cref{sec:existence-disposal}). \Cref{sec:existence-tightness} concludes the section by showing that a coalition-proof PBE always exists if $n=2$ and providing a minimal example of non-existence with three types.

The first three existence conditions have a common mechanism: they produce runs of \cref{alg:greedy-partition} along which the residual maximum $\max\W_t$ is non-increasing, so a coalition is available at each step. 
The fourth existence condition, which we call payoff degradation, ensures that we can replicate residual coalitions with ones that have the same participants but a lower payoff; a coalition is then available at each step even if the residual maximum strictly increases.

Some of our results establish the essential uniqueness of coalition-proof PBE under the following genericity condition, which requires distinct type sets to yield distinct pooling payoffs\footnote{Genericity as given in \cref{dfn:generic} is needed in \cref{sec:existence-qc-mc,sec:existence-betweenness}, in which structural restrictions such as quasiconcavity and/or quasiconvexity are imposed on the upper envelope $\overline{v}$ of the sender's payoff correspondence. Functions of either kind are not closed under addition and hence do not form a vector space, so measure-theoretic notions of genericity such as prevalence \parencite{huntsaueryorke} do not apply. We therefore use injectivity, which is topologically generic: within each such class, under the $\|\cdot\|_\infty$ metric, the envelopes with injective $C\mapsto\overline{v}(\mu^0_C)$ form an open dense set, since $\Theta$ is finite.}:
\begin{definition}\label{dfn:generic}
    Game $G$ is \emph{generic} if the map $C\mapsto\overline{v}(\mu^0_C)$ is injective on non-empty subsets $C\subseteq\Theta$.
\end{definition}

\subsection{Quasiconcavity and Message Completeness}\label{sec:existence-qc-mc}

The first set of conditions that guarantee existence consists of quasiconcavity of the upper envelope $\overline{v}$ and completeness of the message mapping.

\begin{definition}\label{dfn:qc}
    The upper envelope $\overline{v}$ is \emph{quasiconcave} (QC) if, for all $\mu,\mu'\in\Delta\Theta$ and $\lambda\in(0,1)$,
    \[
        \overline{v}(\lambda\mu+(1-\lambda)\mu')
        \geq
        \min\{\overline{v}(\mu),\overline{v}(\mu')\},
    \]
    and it is \emph{strictly quasiconcave} (QC$^*$) if the inequality is strict whenever $\mu\neq\mu'$.
\end{definition}

\begin{definition}\label{dfn:mc}
    The message mapping is \emph{complete} (M-C) if for all $m,m'\in\Msg$ there exists $m''\in\Msg$ with
    $  P(m'')=P(m)\cup P(m')$,
    where $P(m) = M^{-1}(\{m\})$ is the set of types that can send $m$.
\end{definition}

Completeness of the message mapping, as in \cite{bertomeucianciaruso}, requires that the family $\{P(m)\sep m\in\Msg\}$ of all type sets that can be expressed by the message mapping is closed under unions.
That is, if some message means ``My type is in $A$'' and another means ``My type is in $B$,'' then some message means ``My type is in $A$ or $B$.''
Roughly speaking, message completeness ensures that if a group of senders can form a coalition using multiple messages, then they can also form a coalition using a single message. 

Under these conditions, the following lemma lets us restrict our attention to ``pooling'' coalitions:

\begin{lemma}\label{lem:pooling}
     If $\overline{v}$ is quasiconcave and the message mapping is complete, then for every coalition $(C,\sigma,w)$ of $G$, there is a message $m^*$ with $P(m^*)=C$, and for any such $m^*$, the pooling triple $(C,\sigma^{m^*},\overline{v}(\mu^0_C))$ is a coalition of $G$ with
    $        \overline{v}(\mu^0_C)\geq w.
    $
    
    Consequently, the largest coalition payoff in the game is the largest pooling payoff,
    \[
        \max\W_\Theta=\max_{m\in\Msg}\overline{v}(\mu^0_{P(m)}),
    \]
    and every coalition attaining $\max\W_\Theta$ has a type set $C$ with $\overline{v}(\mu^0_C)=\max\W_\Theta$.
\end{lemma}

Note that quasiconcavity of $\overline{v}$ and completeness of the message mapping are both inherited by every restricted game. Hence \cref{lem:pooling} applies to each restricted game as well.
Under these two conditions a coalition-proof PBE exists, and under genericity it is essentially unique.

\begin{theorem}\label{thm:one}
    If $\overline{v}$ is quasiconcave and the message mapping is complete, then a coalition-proof PBE exists. 
    If, in addition,
    \begin{enumerate}
        \item[(i)] $\overline{v}$ is strictly quasiconcave or the game is generic, then no run of \cref{alg:greedy-partition} halts, and every run produces a coalition-proof PBE partition. A partition of $G$ is a coalition-proof PBE partition if and only if $w_t=\max\W_t$ and $w_t\leq w_{t-1}$ for every $t$, where $w_0:=\infty$.
        \item[(ii)] the game is generic, then the coalition-proof PBE partition is essentially unique: any two coalition-proof PBE partitions of $G$ have the same length $T$, the same cells $C_1,\dots,C_T$, and the same payoffs $w_1,\dots,w_T$.
    \end{enumerate}
\end{theorem}

To prove \cref{thm:one}, we use the following merging argument, formalized in \cref{lem:merging} in the appendix. A greedy run halts at step $t$ only if the largest coalition payoff of $G_t$ exceeds $w_{t-1}$. Suppose it does: the coalition $(C_{t-1},\sigma_{t-1},w_{t-1})$ removed at step $t-1$ attains payoff $w_{t-1}$, the largest coalition payoff of $G_{t-1}$, and some coalition of $G_t$ with type set $C'$ pays $w'>w_{t-1}$. By message completeness, the two coalitions can be merged into a single pooling coalition of $G_{t-1}$ with type set $\widetilde C\supset C_{t-1}$. By quasiconcavity, this coalition pays at least $w_{t-1}$, hence exactly $w_{t-1}$. Under strict quasiconcavity, this coalition would pay strictly more than $w_{t-1}$, a contradiction. Under genericity, $C_{t-1}$ and $\widetilde C$ would be distinct type sets with the same pooling value $w_{t-1}$, also a contradiction. Under quasiconcavity alone, the two coalitions both pay exactly $w_{t-1}$, but the one with type set $\widetilde{C}$ contains strictly more types. Consequently, a greedy run that always removes payoff-maximal coalitions with the most types does not halt.

\subsection{Betweenness}\label{sec:existence-betweenness}

Next we impose no condition on the message mapping and, in return, strengthen the payoff restriction of \cref{sec:existence-qc-mc} from quasiconcavity to betweenness: the payoff at a mixture of two beliefs lies between the payoffs at these beliefs. Betweenness is a property of a real-valued payoff, so throughout this subsection $V=\{v\}$; the two envelopes coincide, $\overline v=\underline v=v$, and $v$ is continuous by \cref{A4:V-uhc}.\footnote{\textcite{hart2017} provide a microfoundation: betweenness holds whenever the receiver's expected utility is single-peaked in her action at every belief, and for $n=2$ it reduces to monotonicity of $v$ along the simplex. There, the sender's payoff at a belief is the peak of the receiver's single-peaked expected utility. \textcite{bertomeucianciaruso} impose a related condition on a real-valued payoff, their convexity condition (CV): the value of a merged belief lies between the worst and best values of its components.}

\begin{definition}\label{dfn:B}
    The sender's payoff satisfies \emph{betweenness} (B) if $V=\{v\}$ and $v$ is both quasiconcave and quasiconvex: for all $\mu,\mu'\in\Delta\Theta$ and $\lambda\in(0,1)$,
    \[
        \min\{v(\mu),v(\mu')\}
        \leq
        v(\lambda\mu+(1-\lambda)\mu')
        \leq
        \max\{v(\mu),v(\mu')\}.
    \]
    The payoff satisfies \emph{strict betweenness} (B$^*$) if, in addition, both inequalities are strict whenever $v(\mu)\neq v(\mu')$.
\end{definition}

Like quasiconcavity, betweenness is inherited by every restricted game. Quasiconvexity replaces message completeness: even when the types that can send a set of messages cannot pool on a single one, some coalition of these types pays at least their pooling value.

\begin{lemma}\label{lem:btw-value-id}
    If the sender's payoff satisfies betweenness, then
    \begin{enumerate}
        \item[(i)] every coalition $(C,\sigma,w)$ of $G$ satisfies $w=v(\mu^0_C)$;
        \item[(ii)] for every non-empty $R\subseteq\Theta$, the largest coalition payoff is the largest pooling value over relative preimages,
        \[
            \max\W_R=\max_{\varnothing\neq X\subseteq\Msg_R}v\bigl(\mu^0_{M^{-1}_R(X)}\bigr),
        \]
        and this value is attained by a coalition of $G\vert_R$.
    \end{enumerate}
\end{lemma}

For (i), the coalition's prior $\mu^0_C$ is a convex combination of its on-path posteriors. Each of these posteriors yields $w$ by \cref{C4}, so betweenness gives $v(\mu^0_C)=w$. For (ii), exclusivity makes every coalition type set a relative preimage, so (i) gives the upper bound. To attain it, choose a maximizing message set $X^*$ and form the auxiliary game with messages $X^*$ and types $M^{-1}_R(X^*)$. This game has a PBE partition with decreasing payoffs. Its cell priors average to the auxiliary prior, so quasiconvexity and (i) imply that its first cell pays at least the maximum pooling value. The first cell is also a coalition of $G\vert_R$: any type that can send one of its messages belongs to the auxiliary type set, where the cell's exclusivity already applies. Thus this cell attains the bound.

Betweenness guarantees existence, and either strict betweenness or genericity ensures that every greedy run terminates.

\begin{theorem}\label{thm:two}
    If the sender's payoff satisfies betweenness, then a coalition-proof PBE exists. If, in addition,
    \begin{enumerate}
        \item[(i)] the sender's payoff satisfies strict betweenness or $G$ is generic, then no run of \cref{alg:greedy-partition} halts, and every run produces a coalition-proof PBE partition. A partition of $G$ is a coalition-proof PBE partition if and only if $w_t=\max\W_t$ and $w_t\leq w_{t-1}$ for every $t$, where $w_0:=\infty$.
        \item[(ii)] $G$ is generic, then the coalition-proof PBE partition is essentially unique: any two coalition-proof PBE partitions of $G$ have the same length $T$, the same cells $C_1,\dots,C_T$, and the same payoffs $w_1,\dots,w_T$.
    \end{enumerate}
\end{theorem}

As in \cref{sec:existence-qc-mc}, a greedy run can halt only if the residual maximum rises above the previous payoff. Suppose removing a payoff-maximal coalition with type set $C$ and payoff $w$ leaves a coalition with type set $C'$ and payoff $w'>w$. The merging argument shows that $C\cup C'$ is a relative preimage with pooling value $w$ (see \cref{lem:btw-merging} in the appendix). Under strict betweenness, this is impossible, since the union's value lies strictly between $w$ and $w'$. Under genericity, it contradicts injectivity, since $C$ and $C\cup C'$ are distinct type sets with the same pooling value. In either case, removing a payoff-maximal coalition cannot raise the residual maximum, so no greedy run halts.

Under plain betweenness, the equality case can occur, so we select among payoff-maximal coalitions by ranking their induced priors within each level set of $v$. This ranking ensures that a rise in the residual maximum implies that a higher-ranked payoff-maximal coalition was available at the previous step, contradicting the selection. For the construction, see \cref{lem:btw-order} in the appendix.

\subsubsection{Evidence Games and Truth-Leaning}\label{sec:comp-hkp}

An evidence game ({\textcite{hart2017}}) is a disclosure game with a particular message mapping and a payoff that satisfies betweenness.\footnote{In \textcite{hart2017}, the receiver chooses a reward $x\in\mathbb{R}$, the sender's payoff is $x$ itself, and the receiver's expected utility is single-peaked in $x$ at every belief; the unique peak $v(\mu)$ is the sender's payoff, which is single-valued, quasiconcave, and quasiconvex.} The message mapping is a preorder on types: each type can reveal itself or any type below it. Since B holds, a coalition-proof PBE exists by \cref{thm:two}. We show that coalition proofness provides a sender-side foundation for the receiver's optimal-commitment outcome in their model.

\begin{definition}\label{dfn:evidence-game}
    An \emph{evidence game} is a disclosure game $(\Theta,\Msg,M,\mu^0,V)$ such that
    \begin{itemize}
        \item the sender's payoff satisfies betweenness;
        \item the message mapping has \emph{evidence structure}, that is,
        $\Msg=\Theta$, and $M:\Theta\to2^\Theta\smallsetminus\{\varnothing\}$ satisfies
        \begin{enumerate}
            \item[(L1)] $\theta\in M(\theta)$ for every $\theta\in\Theta$ (\emph{reflexivity});
            \item[(L2)] if $\theta'\in M(\theta)$ and $\theta''\in M(\theta')$, then $\theta''\in M(\theta)$ (\emph{transitivity}).
    \end{enumerate}
    \end{itemize}
\end{definition}

\textcite{hart2017} define truth-leaning equilibria as limits of equilibria of perturbed games in which revealing the whole truth carries an infinitesimal advantage. By their Proposition~1, such equilibria exist, and each has the same outcome as an equilibrium satisfying the following two conditions, which for ease of exposition we adopt as our definition.

\begin{definition}\label{dfn:truth-leaning}
    A PBE $(\sigma,\mu,r)$ of an evidence game is \emph{truth-leaning} if
    \begin{enumerate}
        \item[(A0)] for every $\theta\in\Theta$: if $r(\theta)=\max_{m\in M(\theta)}r(m)$, then $\sigma(\theta\sep\theta)=1$;
        \item[(P0)] for every off-path $m\in\Msg$: $\mu(\cdot\sep m)=\delta_m$.
    \end{enumerate}
\end{definition}

Condition (A0) breaks ties in favor of the whole truth: if message $\theta$ gives type $\theta$ its highest reward, then that type sends $\theta$ with probability one. Condition (P0) attributes an unused message $m$ to type $m$. Note that our definition of PBE requires that off-path beliefs are feasible; this requirement is consistent with (P0) as the fully revealing strategy induces belief $\delta_m$ at message $m$, so $\delta_m \in \mathcal{F}(m)$.

We next show that every truth-leaning equilibrium of an evidence game is coalition-proof.

\begin{proposition}\label{prop:HKP}
    Let $G$ be an evidence game. Then the sender's strategy in every truth-leaning equilibrium is a coalition-proof PBE strategy; under the conditions of Theorem~1 of \textcite{hart2017}, the receiver's optimal-commitment outcome is a coalition-proof PBE outcome of $G$.
\end{proposition}

\begin{corollary}\label{cor:hkp-generic}
    In a generic evidence game, every truth-leaning equilibrium and every coalition-proof PBE induce the same cells and the same payoffs.
\end{corollary}

To prove \cref{prop:HKP}, we use the truth-telling dichotomy (\textcite{hart2017}, Proposition 2), which states that in any truth-leaning equilibrium, either type $\theta$ sends its own message with probability one, or type $\theta$ never sends its own message and $v(\delta_\theta) < u(\theta)$.
Consider a truth-leaning equilibrium $(\sigma,\mu,r)$ and suppose, toward a contradiction, that there is a blocking coalition $(\widetilde C, \widetilde \sigma, \widetilde w)$. From betweenness, that coalition pays $\widetilde w = v( \mu^0_{\widetilde C} )$ (\cref{lem:btw-value-id}).

Next, partition the participants in $\widetilde C$ into those who send their own messages with probability one in the original equilibrium (call that set $\widetilde C^{\mathrm{tr}}$) and those who do not.
For each truth-telling member $\theta'\in\widetilde C^{\mathrm{tr}}$, all types that pool on message $\theta'$ in the original equilibrium must also be in $\widetilde C$. Indeed, they can send any message available to the truth-telling member (by transitivity) and receive the same payoff as $\theta'$ in equilibrium, which is $r(\theta') = u(\theta') < \widetilde w$. Thus the on-path posterior after message $\theta'$ is supported on a subset of $\widetilde C$ and pays strictly less than $\widetilde w$.
Next, decompose $\mu^0_{\widetilde C}$ as a convex combination of the following beliefs. First, take all on-path posteriors $\mu_\sigma (\cdot \sep \theta')$ for all truth-telling members $\theta'\in\widetilde C^{\mathrm{tr}}$. These are supported on $\widetilde C$ and pay strictly less than $\widetilde w$. For each non-truth-telling member $\theta''\in\widetilde C\smallsetminus\widetilde C^{\mathrm{tr}}$, take the degenerate belief $\delta_{\theta''}$, weighted by its remaining probability mass. By the truth-telling dichotomy, each such belief pays $v(\delta_{\theta''})<u(\theta'')<\widetilde w$. Quasiconvexity of $v$ thus gives $v(\mu^0_{\widetilde C}) < \widetilde w$, a contradiction. The optimal-commitment conclusion follows because a truth-leaning equilibrium exists and yields the receiver's optimal-commitment outcome by Proposition~1 and Theorem~1 of \textcite{hart2017}. \cref{cor:hkp-generic} follows from \cref{thm:two}(ii).

These results provide an alternative foundation for selecting receiver-optimal outcomes in the model of \textcite{hart2017}. Under genericity, ruling out profitable coalitional deviations selects the same outcome as truth leaning or receiver commitment. Betweenness is central to this conclusion because it restricts the gains from pooling. Pooling can benefit some types at the expense of others, but the pooled payoff cannot exceed the payoffs of all the groups that pool. A coalition also cannot increase its common payoff by changing how its members mix over messages. In \textcite{hart2017}, this restriction follows from the receiver's expected utility being single-peaked in her action at every belief. We show that betweenness, together with the evidence structure, also ensures that the receiver's optimal-commitment outcome survives credible deviations by groups of sender types.

\subsection{Cheap Talk and Message Completeness}\label{sec:existence-cheaptalk}

This subsection gives sufficient conditions for existence under message completeness together with cheap talk without assuming quasiconcavity of the upper envelope or other conditions on $V$.

\begin{definition}\label{dfn:mct}
    The message mapping has \emph{cheap-talk copies} (M-CT) if, for every $m\in\Msg$,
    \[
        \bigl|\{m'\in\Msg\sep P(m')=P(m)\}\bigr|\geq n.
    \]
\end{definition}

Messages with the same preimage carry identical evidence, so choosing among them is cheap talk. Any disclosure game can be augmented to satisfy M-CT by adjoining, for each message, enough copies with the same preimage. The augmented game is again a disclosure game, and M-CT is inherited by every restricted game because the relative preimages of the copies coincide.

Our main tool is the quasiconcave closure of $\overline{v}$, the smallest quasiconcave function lying above it. The quasiconcave closure is the construction by which \textcite{lipnowskiravid}, who call it the \emph{quasiconcave envelope}, characterize the value of sender-optimal cheap talk; we apply it to the disclosure upper envelope $\overline{v}$.

\begin{definition}\label{dfn:qc-closure}
    The \emph{quasiconcave closure} of the upper envelope is the function $\overline{v}^{qc}:\Delta\Theta\to\mathbb{R}$ given by
    \[
        \overline{v}^{qc}(\mu):=\max\Bigl\{\min_{1\leq i\leq n}\overline{v}(p_i)\sep \lambda\in\Delta(\{1,\dots,n\}),\ p_1,\dots,p_n\in\Delta\Theta,\ \textstyle\sum_{i}\lambda_i p_i=\mu\Bigr\}.
    \]
\end{definition}
The maximum is attained because $\overline{v}$ is upper semicontinuous by \cref{A4:V-uhc} and the constraint set is compact.

With cheap-talk copies, a coalition is no longer tied to a single posterior. A coalition on a type set $C=P(m)$ can spread its members across the $n$ copies of $m$ and induce any decomposition of its prior $\mu^0_C$ into at most $n$ posteriors. Because the members have a common payoff, such a splitting pays only up to the value of $\overline{v}$ at its least favorable posterior, and the splitting that makes this least value largest attains the closure $\overline{v}^{qc}(\mu^0_C)$. Therefore, under M-CT a coalition can reach $\overline{v}^{qc}$ at any type set $C=P(m)$ (or $C=P(m)\cap R$ in $G\vert_R$), and together with message completeness this guarantees existence.

\begin{theorem}\label{thm:three}
    If the message mapping is complete and has cheap-talk copies, then a coalition-proof PBE exists.
\end{theorem}

The proof relies on a few key lemmas proved in the appendix. It applies \cref{thm:one} to a companion game whose upper envelope is already quasiconcave: the \emph{closure game} $G^{qc}$, obtained from $G$ by replacing $\overline{v}$ with $\overline{v}^{qc}$ (see \cref{lem:closure-game} in the appendix). Two facts link the two games. First, no coalition of $G$ can pay more than $\overline{v}^{qc}$ at its prior because its common payoff is bounded by $\overline{v}$ at each posterior it induces and the prior is a convex combination of those posteriors (see \cref{lem:qc-bound} in the appendix); this holds in any disclosure game. Second, M-CT lets a coalition reach $\overline{v}^{qc}$ at its prior (see \cref{lem:ct-realization} in the appendix). Hence $G$ and $G^{qc}$ have the same largest coalition payoff in every restricted game. Then, as $\overline{v}^{qc}$ is quasiconcave by construction and $G^{qc}$ inherits M-C from $G$, \cref{thm:one} applies to $G^{qc}$ and returns a coalition-proof PBE partition. By M-CT, each cell of that partition is matched in $G$ by a coalition with the same type set and payoff, so the partition carries over to $G$ step by step and is a coalition-proof PBE partition there as well. The appendix gives the construction.

A special case is pure cheap talk, where every message is available to every type and therefore carries no evidence. Coalition proofness selects the sender-optimal cheap-talk payoff $\overline{v}^{qc}(\mu^0)$ of \textcite{lipnowskiravid}: the sender obtains it in every coalition-proof PBE, though the supporting equilibrium need not be unique.

\begin{corollary}\label{cor:pure-ct}
    Suppose $P(m)=\Theta$ for every $m\in\Msg$ and $|\Msg|\geq n$. Then a coalition-proof PBE exists and is essentially unique: every coalition-proof PBE partition consists of the single cell $\Theta$ with payoff $\overline{v}^{qc}(\mu^0)$.
\end{corollary}

\subsection{Payoff Degradation}\label{sec:existence-disposal}

Our fourth existence condition, payoff degradation, requires that a coalition can lower its payoff while keeping the same participants. It places no restriction on the upper envelope $\overline{v}$, and under free disposal, none on the message mapping either. Write $v_{\min}:=\min_{\Delta\Theta}\underline{v}$ for the minimum of the lower envelope, which is attained because $\underline{v}$ is lower semicontinuous on the compact set $\Delta\Theta$ by \cref{A4:V-uhc}.

\begin{definition}\label{dfn:DEG}
    The game $G$ satisfies \emph{payoff degradation} (\DEG) if, for every non-empty $R\subseteq\Theta$, every coalition $(C,\sigma,w)$ of $G\vert_R$, and every $w'\in[v_{\min},w]$, the restricted game $G\vert_R$ admits a coalition $(C,\sigma',w')$.
\end{definition}

Two conditions guarantee payoff degradation: free disposal, and revelation aversion combined with cheap-talk copies.

\begin{proposition}\label{prop:DEG-conditions}
    In each of the following cases $G$ satisfies payoff degradation:
    \begin{enumerate}
        \item[(i)] the sender has \emph{free disposal}, meaning that $\underline{v}(\mu)=v_{\min}$ for every $\mu\in\Delta\Theta$;
        \item[(ii)] the message mapping has cheap-talk copies and the sender is \emph{revelation-averse}, meaning that $\underline{v}(\delta_\theta)=v_{\min}$ for every $\theta\in\Theta$.
    \end{enumerate}
\end{proposition}

Under free disposal, the sender's worst-case payoff is $v_{\min}$ at every belief, so a coalition can keep its strategy and lower its payoff to any target down to $v_{\min}$. Under revelation aversion, the worst-case payoff is $v_{\min}$ at every point belief, so the target payoff is available at beliefs between each posterior and the point beliefs. Cheap-talk copies let a coalition split each posterior into such beliefs.

\begin{theorem}\label{thm:DEG-exists}
    If $G$ satisfies payoff degradation, then no run of \cref{alg:greedy-partition} halts; consequently, $G$ admits a greedy partition and a coalition-proof PBE exists.
    In particular, these conclusions hold in each case of \cref{prop:DEG-conditions}.
\end{theorem}

Our previous results show existence using conditions that ensure that some run of \cref{alg:greedy-partition} terminates because removing a payoff-maximal coalition cannot raise the residual maximum. With payoff degradation, even if the residual maximum rises, the algorithm can continue by reducing the next coalition's payoff to the previous one's. The reduced payoff is still IR for the remaining types because the greedy choice in the previous step is.

\subsection{The Binary Case and Non-Existence}\label{sec:existence-tightness}

We conclude the existence section with three observations. First, \cref{thm:two}, and \cref{thm:DEG-exists} under free disposal, provide sufficient conditions for existence that constrain only the payoff correspondence $V$, while \cref{thm:three} provides sufficient conditions that constrain only the message mapping $M$. Second, with two types a coalition-proof PBE always exists, which we show in \cref{thm:binary} below. Third, our results are relatively ``tight'': single-valuedness, continuity and concavity of the sender's payoff correspondence, and either cheap-talk copies or the evidence structure on the message mapping are not sufficient for existence of a coalition-proof PBE. \Cref{ex:nonexistence} illustrates this point.

\begin{theorem}\label{thm:binary}
    If $n=2$, then every disclosure game admits a coalition-proof PBE.
\end{theorem}

To understand this result, consider a greedy run of \cref{alg:greedy-partition} with two types. For that run to halt, the first step must remove one type and leave the other type with a higher payoff. Let $\theta$ be the type removed at step 1, $w_1$ the first-step payoff, $\theta'$ the type left behind, and $m$ a message available to $\theta'$. Message $m$ must then be available to $\theta$ as well, or else $\theta'$ could obtain a higher payoff at step $1$. 
Suppose $\theta'$ sends $m$ with probability one and $\theta$ sends it with probability $\alpha$ and its own evidence otherwise. As $\alpha$ rises from $0$ to $1$, the posterior belief after $m$ moves from $\delta_{\theta'}$ (at which $\underline v(\delta_{\theta'})>w_1$ because the run halts) to the prior $\mu^0$ (at which pooling pays at most $w_1$). 
Since $V$ is upper hemicontinuous and has interval values, for some $\alpha$ every message yields payoff $w_1$, meaning that there exists a coalition of both types paying $w_1$. 
A greedy run that selects this coalition at step $1$ terminates, so a coalition-proof PBE exists.
The following three-type example is therefore a minimal example of non-existence of coalition-proof PBE.

\begin{example}\label{ex:nonexistence}
    Let $\Theta=\{1,2,3\}$ with a uniform prior, $\Msg=\{a,b\}$, $M(1)=\{a,b\}$, $M(2)=\{a\}$, $M(3)=\{b\}$, and let $V=\{v\}$ be single-valued with
    \[
        v(\mu):=\min\bigl\{\mu_1+0.9\,\mu_3,\ -0.1\,\mu_1+1.1\,\mu_2+0.9\,\mu_3\bigr\}.
    \]
\end{example}

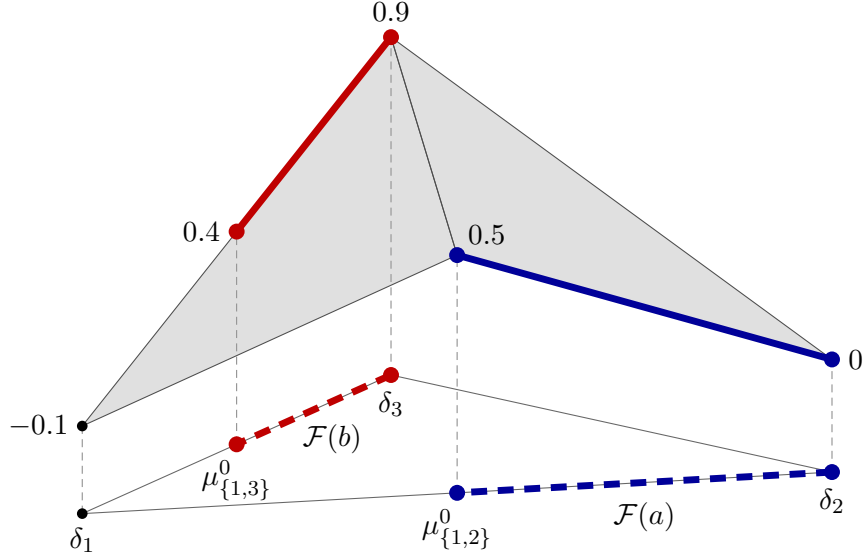
\begin{figure}[ht!]
\centering
\begin{tikzpicture}[font=\normalsize]
\def\zlevel{-0.45}
\begin{axis}[view={125}{25}, hide axis, clip=false, xtick=\empty, ytick=\empty, ztick=\empty,
  xmin=0, xmax=1, ymin=0, ymax=1, zmin=\zlevel, zmax=1, width=11.5cm, height=9.5cm]
\draw[black!60] (0,0,\zlevel) -- (1,0,\zlevel) -- (0,1,\zlevel) -- cycle;
\draw[densely dashed, black!45] (1,0,\zlevel) -- (1,0,-0.1);
\draw[densely dashed, black!45] (0,1,\zlevel) -- (0,1,0);
\draw[densely dashed, black!45] (0,0,\zlevel) -- (0,0,0.9);
\draw[densely dashed, black!45] (0.5,0.5,\zlevel) -- (0.5,0.5,0.5);
\draw[densely dashed, black!45] (0.5,0,\zlevel) -- (0.5,0,0.4);
\filldraw[fill=black!40, fill opacity=0.3, draw=black!70] (0,0,0.9) -- (1,0,-0.1) -- (0.5,0.5,0.5) -- cycle;
\filldraw[fill=black!40, fill opacity=0.3, draw=black!70] (0,0,0.9) -- (0,1,0) -- (0.5,0.5,0.5) -- cycle;
\draw[blue!60!black, line width=2.6pt, dash pattern=on 6pt off 4pt] (0,1,\zlevel) -- (0.5,0.5,\zlevel) node[pos=0.5, below, inner sep=6pt, text=black] {$\mathcal{F}(a)$};
\draw[red!75!black, line width=2.6pt, dash pattern=on 6pt off 4pt] (0,0,\zlevel) -- (0.5,0,\zlevel) node[pos=0.5, below, inner sep=6pt, yshift=2pt, xshift=7pt, text=black] {$\mathcal{F}(b)$};
\draw[blue!60!black, line width=2.6pt, line cap=round] (0,1,0) -- (0.5,0.5,0.5);
\draw[red!75!black, line width=2.6pt, line cap=round] (0,0,0.9) -- (0.5,0,0.4);
\fill[blue!60!black] (0,1,\zlevel) circle (3pt);
\fill[blue!60!black] (0.5,0.5,\zlevel) circle (3pt);
\fill[red!75!black] (0,0,\zlevel) circle (3pt);
\fill[red!75!black] (0.5,0,\zlevel) circle (3pt);
\fill[black] (1,0,\zlevel) circle (2pt);
\node[below, inner sep=5pt] at (1,0,\zlevel) {$\delta_1$};
\node[below, inner sep=6pt] at (0,1,\zlevel) {$\delta_2$};
\node[below, inner sep=6pt] at (0,0,\zlevel) {$\delta_3$};
\node[below, inner sep=6pt] at (0.5,0.5,\zlevel) {$\mu^0_{\{1,2\}}$};
\node[below, inner sep=6pt] at (0.5,0,\zlevel) {$\mu^0_{\{1,3\}}$};
\fill[blue!60!black] (0.5,0.5,0.5) circle (3pt) node[above right, inner sep=4pt, text=black] {$0.5$};
\fill[blue!60!black] (0,1,0) circle (3pt) node[right, inner sep=6pt, text=black] {$0$};
\fill[red!75!black] (0.5,0,0.4) circle (3pt) node[left, inner sep=6pt, text=black] {$0.4$};
\fill[red!75!black] (0,0,0.9) circle (3pt) node[above, inner sep=6pt, text=black] {$0.9$};
\fill[black] (1,0,-0.1) circle (2pt) node[left, inner sep=5pt] {$-0.1$};
\end{axis}
\end{tikzpicture}
\caption{The sender's payoff $v$ (gray) in \cref{ex:nonexistence}. The dashed segments are the sets of feasible beliefs $\mathcal{F}(a)$ (blue) and $\mathcal{F}(b)$ (red); the solid segments above them are the corresponding payoffs.}
\label{fig:nonexistence}
\end{figure}

Under this message mapping, types $2$ and $3$ can only send $a$ and $b$, respectively, so every sender strategy is determined by the probability $\alpha$ with which type $1$ sends $a$. As $\alpha$ ranges over $[0,1]$, the posterior belief after $a$ ranges over $\mathcal{F}(a)=\{(x,1-x,0)\sep x\in[0,\tfrac12]\}$ and the posterior belief after $b$ over $\mathcal{F}(b)=\{(x,0,1-x)\sep x\in[0,\tfrac12]\}$, shown in blue and red in \cref{fig:nonexistence}. Consequently, the payoff after message $a$ ranges from $0$ to $0.5$ and the payoff after message $b$ ranges from $0.4$ to $0.9$, the blue and red segments on the graph of $v$.

By evidence exclusivity \cref{C3}, there are three possible type sets for a coalition: $\{1,2\}$, $\{1,3\}$, and $\{1,2,3\}$. The first two pool on messages $a$ and $b$ and obtain payoffs $0.5$ and $0.4$, respectively. The last type set uses both messages, which by \cref{C4} must yield the same payoff, necessarily in $[0.4,0.5]$. However, the payoff after message $a$ is at least $0.4$ only if $\alpha\geq\tfrac23$, and the payoff after message $b$ is at most $0.5$ only if $\alpha\leq\tfrac13$. Therefore, no coalition has type set $\{1,2,3\}$.

The partition algorithm (\cref{alg:partition}) thus has two possible runs. The first one starts with removing types $\{1,2\}$, obtaining payoff $0.5$ at the first step, and leaves type $3$ with payoff $0.9>0.5$. The resulting partition is not a PBE partition because it does not have non-increasing payoffs. The second run starts with removing types $\{1,3\}$, obtaining payoff $0.4$ at the first step, and leaves type $2$ with payoff $0<0.4$; the resulting partition is the unique PBE partition. It is however not coalition-proof because it is blocked by $\{1,2\}$ pooling on $a$ with payoff $0.5>0.4$. Hence no coalition-proof PBE exists.

Non-existence persists when the game is augmented with cheap-talk copies of $a$ and $b$, or with a message that only type $1$ can send, which gives the message mapping evidence structure. In all three games, every coalition attaining the largest payoff $0.5$ has type set $\{1,2\}$, and removing it leaves type $3$ alone with payoff $0.9$; in other words, removing a payoff-maximal coalition raises the residual maximum. Non-existence is not knife-edge: the argument uses only the strict inequalities $0.5>0.4$ and $0.9>0.5$ and the gap between $\alpha\geq\tfrac23$ and $\alpha\leq\tfrac13$, so it persists under small perturbations of $v$.

\begin{proposition}\label{prop:nonexistence}
    In \cref{ex:nonexistence}, no coalition-proof PBE exists. Non-existence persists when the message mapping is augmented to have cheap-talk copies of both messages or evidence structure. Consequently, single-valuedness, continuity, and concavity of the sender's payoff, and either cheap-talk copies or evidence structure are jointly insufficient for existence.
\end{proposition}

\section{Rich Disclosure Games and the Tent}\label{sec:rich-games}

This section constructs a coalition-proof partition when the message space is maximally ``rich.'' In particular, we require not only that there is a message that allows any set of types to pool together and separate from others (as in \cite{grossman1981} and \cite{milgrom1981}, where any subset $m \subseteq \Theta$ such that $\theta \in m$ is a valid message) but also that \textsl{fractions} of each type can pool while excluding otherwise identical copies of themselves.

One example of a formal model that accommodates the richness that we study here is the stochastic-message-mapping game of \textcite{titova}. In that model, the sender observes a payoff-irrelevant label drawn from an atomless distribution in addition to his type $\theta$. The sender can then condition his message on this label, which means that, for every nonzero vector $\nu\leq\mu^0$, there is a message available to mass $\nu(\theta)$ of type $\theta$ senders for each $\theta\in\Theta$. We do not introduce the labels and messages explicitly and instead use $\nu$ to keep track of these masses.

Suppose that at some step $\nu(\theta)$ is the mass of type $\theta$ senders who remain, with $0\neq\nu\leq\mu^0$, and let $C:=\supp\nu$. A coalition is a nonzero vector $\eta\leq\nu$, where $\eta(\theta)$ is the mass of type $\theta$ senders in the coalition. It induces belief $\hat\eta:=\eta/\eta(\Theta)$, where $\eta(\Theta):=\sum_{\theta\in\Theta}\eta(\theta)$, and payoff $\overline v(\hat\eta)$. Because arbitrary fractions of the remaining types can pool, they can induce exactly the beliefs in
\[
    \Delta_C:=\{\mu\in\Delta\Theta:\supp\mu\subseteq C\}.
\]

For each non-empty $C\subseteq\Theta$, let $\mu^*_C\in\arg\max_{\mu\in\Delta_C}\overline v(\mu)$ be a belief that maximizes the sender's payoff over the beliefs supported on $C$. We write $\mu^*:=\mu^*_{\Theta}$. For a nonzero residual mass vector $\nu$ with $C=\supp\nu$, define
\[
    \lambda^*(\nu):=\min_{\theta:\,\mu^*_C(\theta)>0}\frac{\nu(\theta)}{\mu^*_C(\theta)}
\]
as the largest mass that can pool at that belief.

The game we study here is not a disclosure game in the sense of \cref{dfn:game}, but it admits a direct analogue of the greedy partition \cref{alg:greedy-partition}.
At each step, the algorithm selects a coalition that maximizes the sender's payoff across all beliefs supported on the remaining types and removes at least one type from the game completely.

First, find the largest sender payoff in the game. The first coalition reaches this payoff and fully removes some $\theta\in\Theta$ from the game. Indeed, the first coalition must obtain $\overline v(\mu^*)$ since otherwise senders in the proportions $\mu^*$ could pool and all gain. To do so, we choose it to induce the belief $\mu^*$, meaning its mass vector must be $\eta_1=\lambda\mu^*$ for some $\lambda>0$. To obtain a coalition of maximal size, we pick the highest feasible $\lambda$ such that $\lambda\mu^*(\theta)\leq\mu^0(\theta)$ for every $\theta$, that is, $\lambda=\lambda^*(\mu^0)$. Since $\lambda$ is maximal, $\lambda\mu^*(\theta_i)=\mu^0(\theta_i)$ for some $\theta_i$; hence all $\theta_i$ senders participate in this coalition and type $\theta_i$ is fully removed in step 1.

Let $\nu_1:=\mu^0-\eta_1$ be the masses left after the first coalition is removed, and let $R_2:=\supp\nu_1$. By construction, $R_2$ is a strict subset of $\Theta$. We repeat the same construction with $\mu^*_{R_2}\in\arg\max_{\mu\in\Delta_{R_2}}\overline v(\mu)$: since all $\theta_i$ senders were removed in step 1, the best the remaining types can do is induce the belief $\mu^*_{R_2}$.

We repeat until all types are assigned to a coalition. For example, if type $\theta_j$ is fully removed in step 2, then $R_3:=\supp\nu_2$ is a strict subset of $R_2$ and the next coalition induces $\mu^*_{R_3}$, and so on. The algorithm terminates in at most $n$ steps, and the resulting coalitions $\eta_1,\ldots,\eta_T$ add up to $\mu^0$: every sender belongs to exactly one coalition. We call such a sequence of coalitions a \emph{mass partition} of $\mu^0$. The mass partition generated by this procedure, which we formalize below in \cref{alg:rich-greedy}, is unique if $\overline v$ has a unique maximizer on $\Delta_{R_t}$ at each step. As in \cref{sec:cppbe}, a coalition $\widetilde\eta$ \emph{blocks} a mass partition if it consists only of senders who would strictly gain from it, that is, $\widetilde\eta\leq\sum_{t\,:\,\overline v(\hat\eta_t)<\overline v(\hat{\widetilde\eta})}\eta_t$; a mass partition is \emph{coalition-proof} if no coalition blocks it.

\begin{algorithm}[ht!]
    \caption{Greedy Partition Algorithm for Rich Disclosure Games}\label{alg:rich-greedy}
    Let $t := 1$ and $\nu_0 := \mu^0$\;
    \While{$\nu_{t-1} \neq 0$}{
        let $R_t := \supp \nu_{t-1}$ and select $\mu^*_{R_t} \in \arg\max_{\mu \in \Delta_{R_t}} \overline v(\mu)$\;
        let $\eta_t := \lambda^*(\nu_{t-1})\,\mu^*_{R_t}$\;
        let $\nu_t := \nu_{t-1} - \eta_t$ and $t := t+1$\;
        \medskip
    }
\end{algorithm}

\begin{proposition}\label{prop:cutecase}
    Let $\eta_1,\ldots,\eta_T$ be an output of \cref{alg:rich-greedy}. Then
    \begin{enumerate}
        \item[(i)] at each step, $\eta_t$ attains the highest payoff available to the remaining types and exhausts at least one type in $R_t$;
        \item[(ii)] the algorithm terminates after at most $n$ steps, the coalitions $\eta_1,\ldots,\eta_T$ form a mass partition of $\mu^0$, and their payoffs are weakly decreasing;
        \item[(iii)] if $\overline v$ has a unique maximizer on every non-empty face of $\Delta\Theta$, $(\eta_1,\ldots,\eta_T)$ is the unique mass partition satisfying \textup{(i)} at every step.
    \end{enumerate}
    This mass partition is coalition-proof.
\end{proposition}

For brevity, we do not formally define PBE and PBE partitions for rich disclosure games. Under the analogous definitions, the mass partition of \cref{prop:cutecase} is a PBE partition: its payoffs are weakly decreasing by (ii), and it is individually rational because each sender of type $\theta$ in coalition $\eta_t$ obtains $\overline v(\mu^*_{R_t})\geq\underline v(\delta_\theta)$, the payoff he can guarantee by revealing his type, as $\delta_\theta\in\Delta_{R_t}$. See \textcite{titova} for a formal treatment of PBE in rich disclosure games with finitely many receiver actions.

\subsection*{The Tent Function}

Bayesian persuasion and cheap talk admit geometric characterizations of the sender's ex-ante payoff: the concave closure of $\overline v$ in \textcite{kamenicagentzkow} and its quasiconcave closure in \textcite{lipnowskiravid}, each evaluated at the prior $\mu^0$. We now provide an analogous characterization for rich disclosure.

Assume for the remainder of this section that $\overline v$ has a unique maximizer $\mu^*_C$ on every non-empty face $\Delta_C$. Then \cref{prop:cutecase} gives a unique mass partition generated by \cref{alg:rich-greedy} for every prior $\mu\in\Delta\Theta$; when $\mu$ is on the boundary, the procedure simply starts with the types in $\supp\mu$. If the resulting coalitions are $\eta_1,\ldots,\eta_T$, each sender in coalition $\eta_t$ obtains $\overline v(\hat\eta_t)$, and the sender's ex-ante payoff is the \emph{tent} of $\overline v$:
\[
    \vtent(\mu):=\sum_{t=1}^T\eta_t(\Theta)\,\overline v(\hat\eta_t).
\]
\begin{figure}[!ht]
\centering
\begin{minipage}[t]{0.48\textwidth}
\centering
\begin{tikzpicture}[x=1.15cm,y=0.9cm,font=\normalsize]
    \path[use as bounding box] (-0.4,-1.0) rectangle (5.4,4.4);
    \draw[black!60] (0,0) -- (5,0);
    \draw[densely dashed, black!45] (0,0) -- (0,2.4);
    \draw[densely dashed, black!45] (2,0) -- (2,4);
    \draw[densely dashed, black!45] (5,0) -- (5,0.4);
    \draw[black!55, line width=1pt] plot[smooth, domain=0:2, samples=60] (\x,{0.4*\x*\x+2.4});
    \draw[black!55, line width=1pt] plot[smooth, domain=2:5, samples=60] (\x,{0.4*(\x-5)*(\x-5)+0.4});
    \draw[green!45!black, line width=2.6pt, line cap=round, line join=round] (0,2.4) -- (2,4) -- (5,0.4);
    \foreach \point in {(0,2.4),(2,4),(5,0.4)}
        \fill[green!45!black] \point circle (3pt);
    \fill[black] (0,0) circle (2pt) node[below, inner sep=6pt] {$\delta_{\theta_2}$};
    \fill[black] (2,0) circle (2pt) node[below, inner sep=6pt] {$\mu^*$};
    \fill[black] (5,0) circle (2pt) node[below, inner sep=6pt] {$\delta_{\theta_1}$};
\end{tikzpicture}
\par\smallskip
(a)
\end{minipage}\hfill
\begin{minipage}[t]{0.48\textwidth}
\centering
\begin{tikzpicture}[x=1.15cm,y=0.9cm,font=\normalsize]
    \path[use as bounding box] (-0.4,-1.0) rectangle (5.4,4.4);
    \draw[black!60] (0,0) -- (5,0);
    \draw[densely dashed, black!45] (0,0) -- (0,2.4);
    \draw[densely dashed, black!45] (2,0) -- (2,4);
    \draw[densely dashed, black!45] (5,0) -- (5,0.4);
    \draw[black!55, line width=1pt] plot[smooth, domain=0:5, samples=100] (\x,{4-0.4*(\x-2)*(\x-2)});
    \draw[green!45!black, line width=2.6pt, line cap=round, line join=round] (0,2.4) -- (2,4) -- (5,0.4);
    \foreach \point in {(0,2.4),(2,4),(5,0.4)}
        \fill[green!45!black] \point circle (3pt);
    \fill[black] (0,0) circle (2pt) node[below, inner sep=6pt] {$\delta_{\theta_2}$};
    \fill[black] (2,0) circle (2pt) node[below, inner sep=6pt] {$\mu^*$};
    \fill[black] (5,0) circle (2pt) node[below, inner sep=6pt] {$\delta_{\theta_1}$};
\end{tikzpicture}
\par\smallskip
(b)
\end{minipage}
\caption{The tent $\vtent$ (green) of two payoff functions $\overline v$ (gray) with two types. In both panels, the tent equals $\overline v$ at beliefs $\delta_{\theta_2}$, $\mu^*$, and $\delta_{\theta_1}$ and is affine between them, so the two tents are the same. Both payoff functions are quasiconcave, so the gray curves also give the cheap-talk value. The persuasion value equals the tent in panel (a) and the cheap-talk value in panel (b).}
\label{fig:tent-two-type-comparison}
\end{figure}
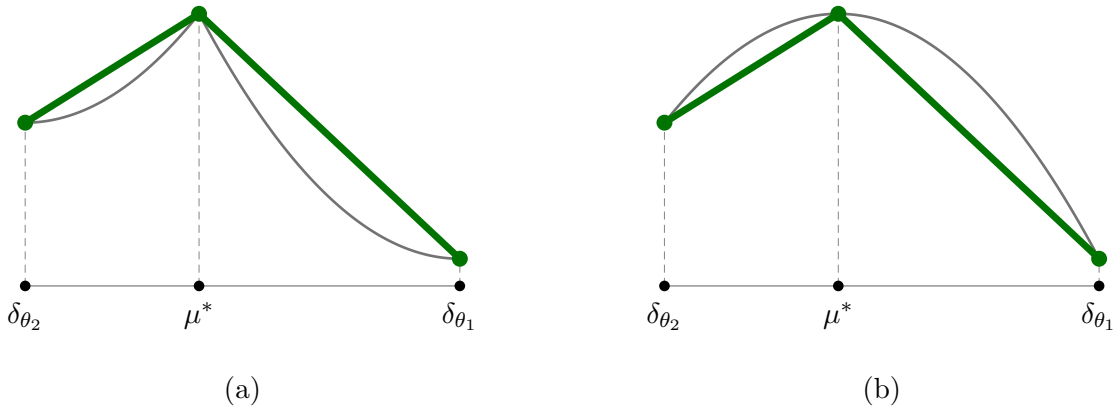
\Cref{fig:tent-two-type-comparison} illustrates the tent with two types and an interior optimal belief $\mu^*$. In a coalition-proof PBE, at least one type must obtain $\overline v(\mu^*)$, or else senders of both types, in the correct proportion, could pool and obtain it. If the prior is $\mu^*$, both types obtain this payoff. Suppose instead that there are more $\theta_1$ senders than needed to induce $\mu^*$. All $\theta_2$ senders then pool with some $\theta_1$ senders at $\mu^*$, and the leftover $\theta_1$ senders reveal their type. The tent therefore decreases linearly from $\overline v(\mu^*)$ to $\overline v(\delta_{\theta_1})$. The reverse applies when $\theta_2$ is more numerous. With two types, only the optimal belief, its payoff, and the two endpoint payoffs determine the tent.

The tent can lie above (panel (a) of \cref{fig:tent-two-type-comparison}) or below (panel (b)) the cheap-talk value but never above the persuasion value. It can lie below the cheap-talk value because there is no ``loyalty'' across sender types: senders who can obtain $\overline v(\mu^*)$ do so even if this leaves other senders worse off and lowers the ex-ante payoff.

With two types, at most one type remains after the first coalition is removed. With more types, the order in which types are exhausted depends on the prior. Suppose the types are exhausted one at a time in the order $(\theta_{i_1},\ldots,\theta_{i_n})$. The residual sets are then
\[
    R_j=\Theta\smallsetminus\{\theta_{i_1},\ldots,\theta_{i_{j-1}}\},
    \qquad j=1,\ldots,n,
\]
and the $j$th coalition induces $\mu^*_{R_j}$. The prior is the average of these beliefs weighted by the masses of the coalitions. If one coalition exhausts several types, the residual sets it skips receive weight zero. A prior may be written this way for different orderings; \cref{prop:tent} shows that all of them give the same payoff.

\begin{proposition}\label{prop:tent}
    Given the prior $\mu^0$, the sender's ex-ante payoff under the coalition-proof mass partition of \cref{prop:cutecase} is $\vtent(\mu^0)$. For every $\mu\in\Delta\Theta$, there exist an ordering of the types and weights $\omega_1,\ldots,\omega_n$ such that
    \[
        \mu=\sum_{j=1}^{n}\omega_j\mu^*_{R_j},
        \qquad \omega_j\geq0,
        \qquad \sum_{j=1}^{n}\omega_j=1.
    \]
    Every representation of this form, for any ordering of the types, gives the same value,
    \[
        \vtent(\mu)=\sum_{j=1}^{n}\omega_j\overline v(\mu^*_{R_j}).
    \]
    Consequently, $\vtent$ is the unique function such that
    \begin{enumerate}
        \item[(i)] $\vtent(\mu^*_C)=\overline v(\mu^*_C)$ for every non-empty $C\subseteq\Theta$;
        \item[(ii)] $\vtent$ is affine on $\mathrm{conv}\{\mu^*_{R_1},\ldots,\mu^*_{R_n}\}$ for every ordering of the types.
    \end{enumerate}
\end{proposition}

\begin{figure}[ht!]
\centering
\begin{tikzpicture}[scale = 1]

\normalsize

\def\zlevel{-1}

\pgfplotsset{
every axis plot/.append style = {
    surf, 
    color = green!45!black,  
    line width = 0.35mm,  
    domain = 0:1,
    samples = 25},
}

\tikzset{
      PlaneBorder/.style = {
        color = green!45!black,
        nearly transparent,
    },
}

\tikzset{
      AxisStyle/.style = {
        color = black,
    },
}

\tikzset{
      VerticalDashedLineStyle/.style = {
        color = black,
        dashed,
    },
}

\tikzset{
      PointStyle/.style = {
        color = black,
        mark size = 1.5pt,
    },
}

\tikzset{
      TentPointStyle/.style = {
        color = green!45!black,
        mark size = 2.6pt,
    },
}

\begin{axis}
[
    view={113}{30},
    xtick=\empty,
    ytick=\empty,
    ztick = \empty,
    clip = false,
    hide axis,
    width = 11cm,
    height = 9cm,
]

\def\vx{     (1,0,-0.73) }
\def\vxbase{ (1,0,\zlevel) }

\def\vy{     (0,1,-0.53) }
\def\vybase{ (0,1,\zlevel) }

\def\vz{     (0,0,-0.13) }
\def\vzbase{ (0,0,\zlevel) }

\def\vxy{     (0.45,0.55,-0.125) }
\def\vxybase{ (0.45,0.55,\zlevel) }

\def\vxz{     (0.20,0,-0.09) }
\def\vxzbase{ (0.20,0,\zlevel) }

\def\vyz{     (0,0.30,-0.04) }
\def\vyzbase{ (0,0.30,\zlevel) }

\def\vxyz{     (0.20,0.30,0) }
\def\vxyzbase{ (0.20,0.30,\zlevel) }


\draw[AxisStyle, ->] (0,0,\zlevel) -- (1.4,0,\zlevel) node[below, yshift = -0.1cm] {$\prob(\theta = \theta_1)$};

\draw[AxisStyle, ->] (0,0,\zlevel) -- (0,1.2,\zlevel) node[below, xshift = 0.5cm] {$\prob(\theta = \theta_2)$};

\draw[AxisStyle, ->] (0,0,\zlevel) -- (0,0,0.2) node[right] {$\overline v$};

\draw[AxisStyle] (1,0,\zlevel) -- (0,1,\zlevel) node[below] {};


\draw[VerticalDashedLineStyle] (1,0,\zlevel) -- \vx;
\node[below] at \vxbase {$\mu^*_{1}$};
\node[PointStyle] at \vxbase {\pgfuseplotmark{*}};
\node[TentPointStyle] at \vx {\pgfuseplotmark{*}};

\draw[VerticalDashedLineStyle] (0,1,\zlevel) -- \vy;
\node[below] at \vybase {$\mu^*_{2}$};
\node[PointStyle] at \vybase {\pgfuseplotmark{*}};
\node[TentPointStyle] at \vy {\pgfuseplotmark{*}};

\draw[VerticalDashedLineStyle] (0,0,\zlevel) -- \vz;
\node[below] at \vzbase {$\mu^*_{3}$};
\node[PointStyle] at \vzbase {\pgfuseplotmark{*}};
\node[TentPointStyle] at \vz {\pgfuseplotmark{*}};

\draw[VerticalDashedLineStyle] \vxybase -- \vxy;
\node[below] at \vxybase {$\mu^*_{12}$};
\node[PointStyle] at \vxybase {\pgfuseplotmark{*}};
\node[TentPointStyle] at \vxy {\pgfuseplotmark{*}};

\draw[VerticalDashedLineStyle] \vxzbase -- \vxz;
\node[below] at \vxzbase {$\mu^*_{13}$};
\node[PointStyle] at \vxzbase {\pgfuseplotmark{*}};
\node[TentPointStyle] at \vxz {\pgfuseplotmark{*}};

\draw[VerticalDashedLineStyle] \vyzbase -- \vyz;
\node[below] at \vyzbase {$\mu^*_{23}$};
\node[PointStyle] at \vyzbase {\pgfuseplotmark{*}};
\node[TentPointStyle] at \vyz {\pgfuseplotmark{*}};

\draw[VerticalDashedLineStyle] \vxyzbase -- \vxyz;
\node[below] at \vxyzbase {$\mu^*$};
\node[PointStyle] at \vxyzbase {\pgfuseplotmark{*}};
\node[TentPointStyle] at \vxyz {\pgfuseplotmark{*}};


\addplot3[line width=1.4pt] coordinates { \vx \vxy \vxyz \vx };

\filldraw[PlaneBorder] \vx -- \vxy -- \vxyz -- \vx;


\addplot3[line width=1.4pt] coordinates { \vx \vxz \vxyz \vx };

\filldraw[PlaneBorder] \vx -- \vxz -- \vxyz -- \vx;


\addplot3[line width=1.4pt] coordinates { \vy \vxy \vxyz \vy };

\filldraw[PlaneBorder] \vy -- \vxy -- \vxyz -- \vy;


\addplot3[line width=1.4pt] coordinates { \vy \vyz \vxyz \vy };

\filldraw[PlaneBorder] \vy -- \vyz -- \vxyz -- \vy;


\addplot3[line width=1.4pt] coordinates { \vz \vyz \vxyz \vz };

\filldraw[PlaneBorder] \vz -- \vyz -- \vxyz -- \vz;


\addplot3[line width=1.4pt] coordinates { \vz \vxz \vxyz \vz };

\filldraw[PlaneBorder] \vz -- \vxz -- \vxyz -- \vz;

 \addplot3[
      surf,
      domain = 0:1,
      samples = 60,
      fill = black!50,
      nearly transparent,
      colormap = {gray}{color=(black!60) color=(black!60)},
      unbounded coords = jump,
      z filter/.expression = {z==0 ? nan : z},
      point meta = 0,
 ] { x+y <= 1 ? (-(y-0.3)^2 - (x-0.2)^2) : nan };

\end{axis}

\end{tikzpicture}
\caption{The tent $\vtent$ (green) of the payoff function $\overline v$ (gray) with three types, where $\mu^*_{ij}:=\mu^*_{\{\theta_i,\theta_j\}}$ and $\mu^*_i:=\mu^*_{\{\theta_i\}}=\delta_{\theta_i}$. The tent equals $\overline v$ at these seven beliefs and is affine on each of the six triangles with vertices $\mu^*$, $\mu^*_{ij}$, and $\mu^*_i$.}
\label{fig:tent-three-types}
\end{figure}
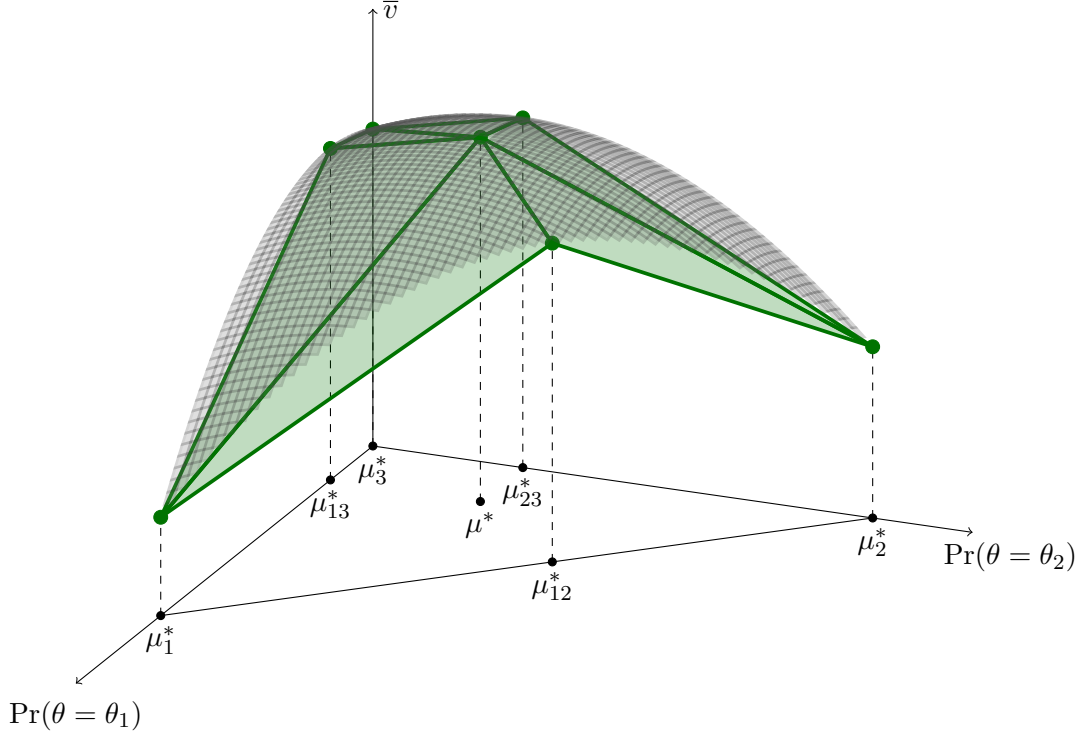

\Cref{fig:tent-three-types} illustrates how the tent is constructed with three types. There is an optimal belief $\mu^*$; three beliefs $\mu^*_{12}$, $\mu^*_{13}$, and $\mu^*_{23}$ that are optimal among beliefs supported on two types; and the three point beliefs $\mu^*_1$, $\mu^*_2$, and $\mu^*_3$. These seven beliefs ($2^n-1$ in general) partition the simplex into six triangles ($n!$ possibly degenerate simplices in general, one for each ordering of the types). Each triangle has as vertices $\mu^*$, one of $\mu^*_{12}$, $\mu^*_{13}$, and $\mu^*_{23}$, and a compatible point belief. On each triangle, the tent is affine and equals $\overline v$ at the vertices. Thus the graph of the tent is made up of six triangles joined along line segments that grow from $\mu^*$.

\section{Conclusion}\label{sec:conclusion}

In this paper, we characterized PBE in general, and properties of coalition-proof PBE in particular, for general disclosure games with type-independent sender preferences. Our general characterization shows that coalitions and partitions are the basic building blocks of all equilibria in this class of games. Simply picking coalitions with constrained-maximal payoffs at each step of the algorithm for the construction of PBE guarantees coalition proofness. Our existence results show that coalition-proof PBE exist under various conditions, some of which apply only to the message mapping, others only to the payoff correspondence, and still others to a combination of the two. Coalition-proof PBE always exist with two types but may fail to exist with three. Our results extend the analysis of cheap talk by \textcite{lipnowskiravid}. We also show that every truth-leaning equilibrium of \textcite{hart2017} is a coalition-proof PBE in the class of games they study. Finally, we characterize sender payoffs under a canonical game of maximally flexible disclosure, which can be thought of as analogous to canonical games of cheap talk and Bayesian persuasion \parencite{kamenicagentzkow}.

Our purpose was to provide a natural equilibrium selection concept as generally as possible for disclosure games. There are at least three avenues of future research. First, coalition-proof PBE may be guaranteed to exist under other sets of conditions that have not occurred to us. Second, we informally suggested in our examples (\cref{ex:intro-incumbent,ex:intro-nuisance}) that a desirable property of selected equilibria in disclosure games is that they respond sensibly to perturbations of the game (for example, the equilibrium outcome should not change when incorporating the opportunity to send messages that no sender type wants to send). It may be worth proposing a formal axiomatic approach to equilibrium selection along these lines and verifying whether coalition proofness or other selection criteria have good properties along this axis. Third, naturally, many games of disclosure that are of practical interest feature state-dependent sender payoffs. Our solution concept does not immediately generalize in this case because even senders using the same message will obtain different payoffs. However, the premise of coalitional deviations still appears sensible in principle and could be useful if there is a translation of it well adapted to this class of games.

\printbibliography

@article{Farrell1993,
	title = {Meaning and {Credibility} in {Cheap}-{Talk} {Games}},
	volume = {5},
	issn = {0899-8256},
	url = {https://www.sciencedirect.com/science/article/pii/S0899825683710298},
	doi = {10.1006/game.1993.1029},
	number = {4},
	urldate = {2023-12-11},
	journal = {Games and Economic Behavior},
	author = {Farrell, Joseph},
	month = oct,
	year = {1993},
	pages = {514--531},
}

@article{GrossmanPerry1986,
	title = {Perfect {Sequential} {Equilibrium}},
	volume = {39},
	issn = {0022-0531},
	url = {https://www.sciencedirect.com/science/article/pii/0022053186900220},
	doi = {10.1016/0022-0531(86)90022-0},
	number = {1},
	urldate = {2023-12-11},
	journal = {Journal of Economic Theory},
	author = {Grossman, Sanford J. and Perry, Motty},
	month = jun,
	year = {1986},
	pages = {97--119},
}

@article{BenPorathEtAl2019,
	title = {Mechanisms {With} {Evidence}: {Commitment} and {Robustness}},
	volume = {87},
	issn = {0012-9682},
	shorttitle = {Mechanisms {With} {Evidence}},
	url = {https://www.econometricsociety.org/doi/10.3982/ECTA14991},
	doi = {10.3982/ECTA14991},
	number = {2},
	urldate = {2023-05-26},
	journal = {Econometrica},
	author = {Ben-Porath, Elchanan and Dekel, Eddie and Lipman, Barton L.},
	year = {2019},
	pages = {529--566},
}

@article{rappoport,
  title={Evidence and Skepticism in Verifiable Disclosure Games},
  author={Rappoport, Daniel},
  journal={Theoretical Economics},
  year={2025},
  doi={10.3982/te5423}
}

@article{titova,
  title={Persuasion with verifiable information},
  author={Titova, Maria and Zhang, Kun},
  journal={Journal of Economic Theory},
  year={2025},
  volume={230},
  pages={106102},
  doi={10.1016/j.jet.2025.106102}
}

@article{grossman1981,
  title={The Informational Role of Warranties and Private Disclosure about Product Quality},
  author={Grossman, Sanford J.},
  journal={The Journal of Law and Economics},
  volume={24},
  number={3},
  pages={461--483},
  year={1981},
  publisher={The University of Chicago Press}
}

@article{milgrom1981,
  title={Good News and Bad News: Representation Theorems and Applications},
  author={Milgrom, Paul R.},
  journal={The Bell Journal of Economics},
  pages={380--391},
  year={1981},
  publisher={JSTOR}
}

@article{sher2011,
  title={Credibility and Determinism in a Game of Persuasion},
  author={Sher, Itai},
  journal={Games and Economic Behavior},
  volume={71},
  number={2},
  pages={409--419},
  year={2011},
  publisher={Elsevier}
}

@article{sher2014,
  title={Persuasion and Dynamic Communication},
  author={Sher, Itai},
  journal={Theoretical Economics},
  volume={9},
  number={1},
  pages={99--136},
  year={2014},
  publisher={Wiley Online Library}
}

@article{glazerrubinstein,
  title={On Optimal Rules of Persuasion},
  author={Glazer, Jacob and Rubinstein, Ariel},
  journal={Econometrica},
  volume={72},
  number={6},
  pages={1715--1736},
  year={2004},
  publisher={Wiley Online Library}
}

@article{kamenicagentzkow,
  title={Bayesian Persuasion},
  author={Kamenica, Emir and Gentzkow, Matthew},
  journal={American Economic Review},
  volume={101},
  number={6},
  pages={2590--2615},
  year={2011}
}

@article{lipnowskiravid,
  title={Cheap Talk with Transparent Motives},
  author={Lipnowski, Elliot and Ravid, Doron},
  journal={Econometrica},
  volume={88},
  number={4},
  pages={1631--1660},
  year={2020},
  publisher={Wiley Online Library}
}

@article{bertomeucianciaruso,
  title={Verifiable Disclosure},
  author={Bertomeu, Jeremy and Cianciaruso, Davide},
  journal={Economic Theory},
  volume={65},
  number={4},
  pages={1011--1044},
  year={2018},
  publisher={Springer}
}

@article{hart2017,
  title={Evidence Games: Truth and Commitment},
  author={Hart, Sergiu and Kremer, Ilan and Perry, Motty},
  journal={American Economic Review},
  volume={107},
  number={3},
  pages={690--713},
  year={2017}
}

@article{hagenbachkoessler,
  title={Certifiable Pre-Play Communication: Full Disclosure},
  author={Hagenbach, Jeanne and Koessler, Fr{\'e}d{\'e}ric and Perez-Richet, Eduardo},
  journal={Econometrica},
  volume={82},
  number={3},
  pages={1093--1131},
  year={2014},
  publisher={Wiley Online Library}
}

@article{huntsaueryorke,
  title={Prevalence: a Translation-Invariant “Almost Every” on Infinite-Dimensional Spaces},
  author={Hunt, Brian R. and Sauer, Tim and Yorke, James A.},
  journal={Bulletin of the American Mathematical Society},
  volume={27},
  number={2},
  pages={217--238},
  year={1992}
}

@article{matthews1991,
  title={Refining Cheap-talk Equilibria},
  author={Matthews, Steven A. and Okuno-Fujiwara, Masahiro and Postlewaite, Andrew},
  journal={Journal of Economic Theory},
  volume={55},
  number={2},
  pages={247--273},
  year={1991},
  publisher={Elsevier}
}

@article{undefeated,
  title={Belief-Based Refinements in Signalling Games},
  author={Mailath, George J. and Okuno-Fujiwara, Masahiro and Postlewaite, Andrew},
  journal={Journal of Economic Theory},
  volume={60},
  number={2},
  pages={241--276},
  year={1993},
  publisher={Elsevier}
}

@article{callanderjpe,
  title={The Power of Referential Advice},
  author={Callander, Steven and Lambert, Nicolas S. and Matouschek, Niko},
  journal={Journal of Political Economy},
  volume={129},
  number={11},
  pages={3073--3140},
  year={2021},
  publisher={The University of Chicago Press Chicago, IL}
}

@article{callander,
  title={Cheap Talk in Complex Environments},
  author={Aybas, Yunus C. and Callander, Steven},
  year={2024}
}

@article{perego,
    title={The Selective Disclosure of Evidence: an Experiment},
    author={Farina, Agata and Fr\'echette, Guillaume R. and Ispano, Alessandro and Lizzeri, Alessandro and Perego, Jacopo},
    journal={Review of Economic Studies},
    year={2026},
    doi={10.1093/restud/rdag076}
}

@article{koesslerskreta,
  title={Informed Information Design},
  author={Koessler, Fr{\'e}d{\'e}ric and Skreta, Vasiliki},
  journal={Journal of Political Economy},
  volume={131},
  number={11},
  pages={3186--3232},
  year={2023},
  publisher={The University of Chicago Press Chicago, IL}
}

@article{wenhao,
  title={A Role for Cheap Talk in Disclosure},
  author={Wu, Wenhao},
  year={2022},
  publisher={ShanghaiTech SEM Working Paper}
}

@book{border1985,
  title={Fixed Point Theorems with Applications to Economics and Game Theory},
  author={Border, Kim C.},
  year={1985},
  publisher={Cambridge University Press, Cambridge}
}

@book{rockafellar1970,
  title={Convex Analysis},
  author={Rockafellar, R. Tyrrell},
  year={1970},
  publisher={Princeton University Press, Princeton}
}

@book{aliprantisborder,
  title={Infinite Dimensional Analysis: A Hitchhiker's Guide},
  author={Aliprantis, Charalambos D. and Border, Kim C.},
  year={2006},
  edition={3},
  publisher={Springer, Berlin}
}

@unpublished{gao2026,
  author = {Gao, Ying},
  title = {Inference from Selectively Disclosed Data},
  date = {2026-07-31},
  note = {Working paper},
  eprint = {2204.07191},
  eprinttype = {arxiv}
}

\clearpage
\appendix
\numberwithin{theorem}{section}
\numberwithin{proposition}{section}
\numberwithin{lemma}{section}
\numberwithin{corollary}{section}
\numberwithin{definition}{section}
\numberwithin{remark}{section}
\numberwithin{equation}{section}

\section*{Appendix}
\subsection*{Contents}
\startcontents[appendix]
\printcontents[appendix]{}{1}{}
\medskip

\section[Proofs for Section~\ref{sec:pbe}]{Proofs for \cref{sec:pbe}}

\begin{proof}[Proof of \cref{prop:pbe-characterization}]
Suppose first that $\sigma$ is supported as a PBE by $(\mu,r)$. Let $u(\theta) = \max_{m \in M(\theta)} r(m)$ be the associated equilibrium-payoff function defined above, and list its distinct values as $w_1>\cdots>w_T$. Define
    \[
        C_t:=\{\theta\in\Theta \sep u(\theta) = w_t\},
        \qquad
        \sigma_t:=\sigma|_{C_t},
        \qquad
        R_t:=\bigcup_{s=t}^T C_s.
    \]
    We claim that $\Pi=\{(C_t,\sigma_t,w_t)\}_{t=1}^T$ is a partition. 
    First, notice that the cells are non-empty, disjoint, and cover $\Theta$, and $\sigma_t$ is a coalition strategy on $C_t$. It remains to show that $(C_t,\sigma_t,w_t)$ is a coalition of $G_t$, i.e., it satisfies conditions \cref{C3} and \cref{C4}, for each $t$.
    For \cref{C3}, note that every $m\in X_t$ satisfies $r(m)=w_t$, because it is used on path by a type in $C_t$.
    If some $\theta\in R_t$ can send a message $m\in X_t$, then $u(\theta)\geq r(m)=w_t$. 
    But every type in $R_t=\bigcup_{s\geq t}C_s$ has equilibrium payoff at most $w_t$. 
    Hence $u(\theta)=w_t$, so $\theta\in C_t$. 
    Therefore $M^{-1}_t(X_t)\subseteq C_t$.
    For \cref{C4}, fix $m \in X_t$. Any type that sends $m\in X_t$ under $\sigma$ must have payoff $w_t$, and hence belongs to $C_t$. Thus the on-path posterior $\mu_\sigma(\cdot\sep m)$ is the posterior generated by $\sigma_t$ in $G\vert_{C_t}$. Since PBE beliefs are Bayesian on path and $r(m)\in V(\mu(\cdot\sep m))$, we get $w_t=r(m)\in V(\mu_{\sigma_t}(\cdot\sep m))$. Therefore each $(C_t,\sigma_t,w_t)$ is a coalition in $G_t$, so $\Pi$ is a partition with strictly decreasing payoffs.

    It remains only to check IR. For every message $m$, $\mu(\cdot\sep m)\in \mathcal{F}(m)$ and
    $r(m)\in V(\mu(\cdot\sep m))$, so $r(m)\geq \min_{\nu\in \mathcal{F}(m)}\underline v(\nu)$. Maximizing
    over $m\in M(\theta)$ gives $u(\theta)\geq\underline u(\theta)$. Hence $w_t\geq\underline u(\theta)$ for every $\theta\in C_t$.

    Conversely, let $\Pi=\{(C_t,\sigma_t,w_t)\}_{t=1}^T$ be an IR partition with $w_1\geq\cdots\geq w_T$, and let
    $\sigma=\sigma^\Pi$. For each off-path message $m\notin X(\sigma)$, choose $\nu_m\in \mathcal{F}(m)$ attaining
    $\min_{\nu\in \mathcal{F}(m)}\underline v(\nu)$ and set $a_m:=\underline v(\nu_m)$. Then $a_m\in V(\nu_m)$ and
    $a_m\leq \underline u(\theta)$ for every $\theta$ with $m\in M(\theta)$.

    Define beliefs and payoffs by
    \[
        \mu(\cdot\sep m)=
        \begin{cases}
            \mu_\sigma(\cdot\sep m), & m\in X_t \text{ for some }t,\\
            \nu_m, & m\notin X(\sigma),
        \end{cases}
        \qquad
        r(m)=
        \begin{cases}
            w_t, & m\in X_t \text{ for some }t,\\
            a_m, & m\notin X(\sigma).
        \end{cases}
    \]
    The constructed belief system $\mu$ and payoff function $r$ therefore satisfy the first three requirements in the definition of $\sigma$ being a PBE strategy; it remains only to show that the sender is sequentially rational.
    Fix $\theta \in C_t$. 
    Every message in $X_t$, and hence every message used by $\theta$ under $\sigma$, yields payoff $w_t$. 
    Every message in any earlier evidence set $X_s$, $s<t$, is not available to $\theta$, since that would violate evidence exclusivity \cref{C3} of the step-$s$ coalition.
    Every message in any later evidence set $X_s$, $s>t$, yields payoff $w_s \leq w_t$. 
    Every off-path feasible message yields payoff at most $\underline{u}(\theta) \leq w_t$. 
    Thus, type $\theta$ does not have profitable deviations and $\sigma$ is a PBE strategy.
    \end{proof}

\section[Proofs for Section~\ref{sec:cppbe}]{Proofs for \cref{sec:cppbe}}

\begin{lemma}\label{lem:coalition-compact}
    For every non-empty $R\subseteq\Theta$, the coalition payoff set $\W_R$ is non-empty and compact.
\end{lemma}

\begin{proof}
Since $G\vert_R$ is itself a disclosure game and $\W_R$ is its set of coalition payoffs, it suffices to prove the claim for $R=\Theta$; write $\W:=\W_\Theta$. Non-emptiness is \cref{remark:coalitions-exist}.

For each non-empty $C\subseteq\Theta$, let $\mathrm{Coal}(C):=\{(\sigma,w)\sep(C,\sigma,w)\text{ is a coalition of }G\}$. We show every $\mathrm{Coal}(C)$ is compact. Hence its image under the continuous projection $(\sigma,w)\mapsto w$ is compact, and $\W$ is a finite union of compact sets, hence compact.

The coalition strategies on $C$ form the compact set $S:=\prod_{\theta\in C}\Delta M(\theta)$, and every coalition payoff lies in the interval $I:=\bigl[\min_{\Delta\Theta}\underline v,\,\max_{\Delta\Theta}\overline v\bigr]$, which is bounded as $\underline v$ is lower semicontinuous and $\overline v$ upper semicontinuous \cref{A4:V-uhc} on the compact set $\Delta\Theta$. Thus $\mathrm{Coal}(C)$ is a subset of the compact set $S\times I$. It suffices to show that it is closed.

Let $(\sigma_k,w_k)\in\mathrm{Coal}(C)$ with $(\sigma_k,w_k)\to(\sigma,w)$. The evidence sets $X(\sigma_k)$ take finitely many values, so passing to a subsequence, we may assume that $X(\sigma_k)=X$ for all $k$. We claim $(C,\sigma,w)$ is a coalition of $G$. First, $\sigma$ is again a coalition strategy on $C$, since the constraints defining $S$ are preserved in the limit. Second, $X(\sigma)\subseteq X$: if $\sigma(m\sep\theta)>0$, then $\sigma_k(m\sep\theta)>0$ for all large $k$, so $m\in X(\sigma_k)=X$. Each $(C,\sigma_k,w_k)$ satisfies exclusivity \cref{C3}, so $M^{-1}(X)=M^{-1}(X(\sigma_k))\subseteq C$; with $X(\sigma)\subseteq X$, this gives $M^{-1}(X(\sigma))\subseteq C$, which is \cref{C3} for $(C,\sigma,w)$.

Third, fix $m\in X(\sigma)$. The denominator $p_\sigma(m)=\sum_{\theta\in C}\mu^0_C(\theta)\,\sigma(m\sep\theta)$ is positive and $p_{\sigma_k}(m)\to p_\sigma(m)$, so the induced beliefs converge, $\mu_{\sigma_k}(\cdot\sep m)\to\mu_\sigma(\cdot\sep m)$. Since $m\in X(\sigma)\subseteq X=X(\sigma_k)$, condition \cref{C4} for $(C,\sigma_k,w_k)$ gives $\underline v(\mu_{\sigma_k}(\cdot\sep m))\leq w_k\leq\overline v(\mu_{\sigma_k}(\cdot\sep m))$. Letting $k\to\infty$, the semicontinuity of the envelopes from \cref{A4:V-uhc} yields $w\in V(\mu_\sigma(\cdot\sep m))$. As $m\in X(\sigma)$ was arbitrary, $(C,\sigma,w)$ is a coalition of $G$, so $(\sigma,w)\in\mathrm{Coal}(C)$. Hence $\mathrm{Coal}(C)$ is a closed subset of $S\times I$, and therefore compact. Messages in $X\smallsetminus X(\sigma)$ simply drop out of the limit coalition, since the common-payoff condition is required only on the evidence actually used.
\end{proof}

\begin{proof}[Proof of \cref{prop:cppbe-characterization}]
($\Leftarrow$) Let $\Pi$ be greedy. By \cref{dfn:greedy}, $w_t$ belongs to $\W_t\cap[\max_{\theta\in R_t}\underline{u}(\theta),w_{t-1}]$ for every $t$. Thus $w_t\leq w_{t-1}$ and $w_t\geq\max_{\theta\in R_t}\underline{u}(\theta)$ for every $t$, so $\Pi$ has non-increasing payoffs and is IR by \cref{lem:strong-IR}. Hence $\Pi$ is a PBE partition.

Suppose, toward a contradiction, that $(\widetilde C,\widetilde\sigma,\widetilde w)$ is a blocking coalition of $\Pi$. Let $\tau:=\min\{t\sep w_t<\widetilde w\}$. 
This index exists because $\Theta(\widetilde w)\neq\varnothing$. Since payoffs are non-increasing, the cells with payoff strictly below $\widetilde w$ are exactly $C_\tau,\ldots,C_T$: for $s<\tau$, $w_s\geq\widetilde w$ by minimality of $\tau$, while for $s\geq\tau$, $w_s\leq w_\tau<\widetilde w$. Hence $\Theta(\widetilde w)=R_\tau$. By the definition of blocking, $(\widetilde C,\widetilde\sigma,\widetilde w)$ is a coalition of $G\vert_{\Theta(\widetilde w)}=G_\tau$, so $\widetilde w\in\W_\tau$. Moreover, $\widetilde w>w_\tau\geq\max_{\theta\in R_\tau}\underline{u}(\theta)$, and $\widetilde w\leq w_{\tau-1}$ by minimality of $\tau$ (with $w_0=\infty$ if $\tau=1$). Thus
\[
    \widetilde w
    \in
    \W_\tau
    \cap
    \bigl[
        \max_{\theta\in R_\tau}\underline{u}(\theta),\,
        w_{\tau-1}
    \bigr],
    \qquad
    \widetilde w>w_\tau,
\]
contradicting \cref{dfn:greedy}, which makes $w_\tau$ the maximum of this set.

($\Rightarrow$) Let $\Pi$ be a coalition-proof PBE partition, and fix $t$. Since $(C_t,\sigma_t,w_t)$ is a coalition of $G_t$, we have $w_t\in\W_t$. Since $\Pi$ is a PBE partition, its payoffs are non-increasing and it is IR; by \cref{lem:strong-IR}, $w_t\geq\max_{\theta\in R_t}\underline{u}(\theta)$, and non-increasing payoffs give $w_t\leq w_{t-1}$. Hence $w_t$ belongs to the greedy constraint set at step $t$.

Suppose $w_t$ is not the maximum of that set. Then there is a payoff $w\in\W_t$ such that
\[
    w
    \in
    \W_t
    \cap
    \bigl[
        \max_{\theta\in R_t}\underline{u}(\theta),\,
        w_{t-1}
    \bigr],
    \qquad
    w>w_t.
\]
Since $w\in\W_t$, there exists a coalition $(C,\sigma,w)$ of $G_t$. For every $s<t$, non-increasing payoffs give $w_s\geq w_{t-1}\geq w$; for every $s\geq t$, they give $w_s\leq w_t<w$. Therefore $\Theta(w)=R_t\neq\varnothing$. Since $(C,\sigma,w)$ is a coalition of $G_t=G\vert_{\Theta(w)}$, it is a blocking coalition of $\Pi$, contradicting coalition proofness. Thus $w_t$ is the maximum of the greedy constraint set at every step, and $\Pi$ is greedy.

Finally, by definition, a strategy is a coalition-proof PBE strategy if and only if it is associated with some coalition-proof PBE (equivalently, greedy) partition.
\end{proof}

\section{Coalition-Optimal Equilibrium Partitions}\label{sec:coe}

This section develops an auxiliary strengthening of greedy partitions which we use to prove existence theorems.

Call a sequence $\{(C_s,\sigma_s,w_s)\}_{s=1}^{t-1}$ (possibly empty) a \emph{prefix} with residual $R_t$ if, setting $R_1:=\Theta$ and $R_{s+1}:=R_s\smallsetminus C_s$, each $(C_s,\sigma_s,w_s)$ is a coalition of $G_s$ and $R_t\neq\varnothing$. Initial segments of partitions are prefixes; conversely, a prefix is an unfinished run of \cref{alg:partition}. Call such a prefix \emph{greedy} if $w_s=\max\bigl(\W_s\cap[\max_{\theta\in R_s}\underline{u}(\theta),\,w_{s-1}]\bigr)$ at every step $s$ (with $w_0:=\infty$). Every initial segment of a greedy partition is a greedy prefix, and an unfinished run of \cref{alg:greedy-partition} is exactly a greedy prefix.
For $m\in\Msg$, let $F(m):=\{\theta\in\Theta\sep M(\theta)=\{m\}\}$ be the set of types \emph{forced} to send $m$.

\begin{lemma}\label{lem:remaining-games}
    Let $R_t$ be the residual of a prefix. Then:
    \begin{enumerate}
        \item[(i)] for every $m\in\Msg_t$, every type forced to send $m$ remains: $F(m)\subseteq R_t$;
        \item[(ii)] for every $m\in\Msg_t$, the conditional prior on the residual senders of $m$ is feasible: $\mu^0_{P(m)\cap R_t}\in\mathcal{F}(m)$;
        \item[(iii)] $\W_t$ is non-empty and compact, and
        $  \max\W_t\geq\max_{\theta\in R_t}\underline{u}(\theta)$.
    \end{enumerate}
\end{lemma}
\begin{proof}
    (i) Suppose $\theta_F\in F(m)\cap C_s$ for some $s<t$. Since $M(\theta_F)=\{m\}$ and $\varnothing\neq\supp\sigma_s(\cdot\sep\theta_F)\subseteq M(\theta_F)$, we have $m\in X_s:=X(\sigma_s)$. As $m\in\Msg_t$, some $\theta\in R_t$ has $m\in M(\theta)$; then $\theta\in R_s$ and $m\in M(\theta)\cap X_s$ give $\theta\in M^{-1}_s(X_s)\subseteq C_s$ by exclusivity at step $s$, contradicting $\theta\in R_t\subseteq R_{s+1}=R_s\smallsetminus C_s$.

    (ii) Write $P_t:=P(m)\cap R_t$, which is non-empty since $m\in\Msg_t$ and contains $F(m)$ by (i). Let $\sigma$ be a strategy under which every type in $P_t$ sends $m$ and every other type sends a message in $M(\theta)\smallsetminus\{m\}$; this is possible because no type outside $P_t$ is forced to send $m$, by (i). The senders of $m$ are then exactly $P_t$, so $\mu_\sigma(\cdot\sep m)=\mu^0_{P_t}$, and hence $\mu^0_{P_t}\in\mathcal{F}(m)$.

    (iii) $\Msg_t$ is non-empty, so $G_t$ admits a coalition by \cref{remark:coalitions-exist} and $\W_t\neq\varnothing$; it is compact by \cref{lem:coalition-compact}, so the maximum exists. Fix $\theta\in R_t$ and $m\in M(\theta)\subseteq\Msg_t$. Pooling on $m$ in $G_t$ is a coalition with type set $P_t=P(m)\cap R_t$ and payoff $\overline{v}(\mu^0_{P_t})$, so by (ii),
    \[
        \max\W_t\geq\overline{v}(\mu^0_{P_t})\geq\underline{v}(\mu^0_{P_t})\geq\min_{\mu\in\mathcal{F}(m)}\underline{v}(\mu).
    \]
    Taking the maximum over $m\in M(\theta)$ gives $\max\W_t\geq\underline{u}(\theta)$, and then the maximum over $\theta\in R_t$.
\end{proof}

\begin{definition}\label{dfn:coe}
    A partition $\Pi=\{(C_t,\sigma_t,w_t)\}_{t=1}^T$ of $G$ is a \emph{coalition-optimal equilibrium} (COE) partition if, for every $t$ (with $w_0:=\infty$),
    $ w_t=\max\W_t$ and $w_t\leq w_{t-1}$.
\end{definition}

The maximum exists by \cref{lem:remaining-games}(iii). Unlike the greedy selection, a COE partition maximizes the payoff with no upper constraint and then requires the resulting sequence to be non-increasing.

\begin{lemma}\label{lem:coe-greedy}
    \begin{enumerate}
        \item[(i)] Every COE partition of $G$ is a greedy partition.
        \item[(ii)] A greedy partition of $G$ is a COE partition iff $\max\W_t\leq w_{t-1}$ for every $t$.
    \end{enumerate}
\end{lemma}
\begin{proof}
    (i) Let $\Pi$ be a COE partition and fix $t$. The initial segment through step $t-1$ is a prefix with residual $R_t$, so by \cref{lem:remaining-games}(iii), $w_t=\max\W_t\geq\max_{\theta\in R_t}\underline{u}(\theta)$, and $w_t\leq w_{t-1}$ by definition. Hence $w_t$ lies in the greedy constraint set $\W_t\cap[\max_{\theta\in R_t}\underline{u}(\theta),\,w_{t-1}]$ and equals the maximum of $\W_t$, which contains that set; so $w_t$ is the maximum of the constraint set, and $\Pi$ is greedy. By \cref{prop:cppbe-characterization}, $\Pi$ is a coalition-proof PBE partition.

    (ii) Let $\Pi$ be greedy. If $\Pi$ is COE, then $\max\W_t=w_t\leq w_{t-1}$ for every $t$. Conversely, suppose $\max\W_t\leq w_{t-1}$ for every $t$. By \cref{lem:remaining-games}(iii), $\max\W_t\geq\max_{\theta\in R_t}\underline{u}(\theta)$, so $\max\W_t$ lies in the greedy constraint set, whose maximum is therefore $\max\W_t$ itself; by greediness $w_t=\max\W_t$, and $w_t\leq w_{t-1}$ holds since $w_t$ lies in the constraint set. So $\Pi$ is COE.
\end{proof}

\section[Proofs for Section~\ref{sec:existence-qc-mc}]{Proofs for \cref{sec:existence-qc-mc}}

\begin{proof}[Proof of \cref{lem:pooling}]
Consider coalition $(C,\sigma,w)$ of $G$. By exclusivity of evidence, $M^{-1}(X(\sigma)) = C$.
Iterating M-C over the finite set $X(\sigma)$ delivers message $m^*$ with $P(m^*) = \bigcup_{m\in X(\sigma)} P(m) = M^{-1}(X(\sigma)) = C$.
By \cref{remark:coalitions-exist}, the pooling triple $(C,\sigma^{m^*},w')$ is a coalition for every $w'\in V(\mu^0_C)$, in particular, for $w'=\overline{v}(\mu^0_C)$.

Next, in $G\vert_C$, the prior $\mu^0_C$ is a linear combination of the posteriors $\mu_\sigma(\cdot\sep m)$ induced by $\sigma$ at the messages $m\in X(\sigma)$, with weights $p_\sigma(m)>0$ summing to one. 
Iterating QC gives
\[
    \overline{v}(\mu^0_C)\geq\min_{m\in X(\sigma)}\overline{v}(\mu_\sigma(\cdot\sep m)) \geq w,
\]
where the last inequality follows from the common-payoff condition \cref{C4} for coalition $(C,\sigma,w)$.

For the identity, recall that every pooling payoff $\overline{v}(\mu^0_{P(m)})$ is itself a coalition payoff by \cref{remark:coalitions-exist}. 
Also, by the first part of this proof, every coalition $(C,\sigma,w)$ satisfies $w\leq\overline{v}(\mu^0_C)=\overline{v}(\mu^0_{P(m^*)})$, where $P(m^*)=C$.
\end{proof}

The following merging lemma is used in the existence and essential-uniqueness results below. It shows that, under QC$^*$ or genericity, removing a coalition that attains $\max\W_R$ cannot raise the largest coalition payoff on the residual types $R\smallsetminus C$.

\begin{lemma}\label{lem:merging}
    Assume QC and M-C. Let $R\subseteq\Theta$ be non-empty, let $(C,\sigma,w)$ be a coalition of $G\vert_R$ with $w=\max\W_R$, and suppose $R':=R\smallsetminus C$ is non-empty and
    $\max\W_{R'}>\max\W_R$.
    Then there is a coalition of $G\vert_R$ whose type set $\widetilde C$ satisfies $C\subsetneq\widetilde C$ and whose payoff is $w$; moreover $\overline{v}(\mu^0_C)=\overline{v}(\mu^0_{\widetilde C})=w$.
    
    If QC$^*$ holds, or if $G$ is generic, this configuration is impossible; that is,
    $\max\W_{R\smallsetminus C}\leq\max\W_R $
    whenever a coalition with type set $C$ attains $\max\W_R$ and $R\smallsetminus C\neq\varnothing$.
\end{lemma}

\begin{proof}
    Since QC and M-C are inherited by restricted games (\cref{sec:existence-qc-mc}), \cref{lem:pooling,remark:coalitions-exist} apply in $G\vert_R$ and $G\vert_{R'}$; the sets $\W_R$ and $\W_{R'}$ are non-empty and compact by \cref{lem:coalition-compact}, so both maxima exist.

    By \cref{lem:pooling} in $G\vert_R$, there is a message $m''$ with $P(m'')\cap R=C$ and $\overline{v}(\mu^0_C)\in\W_R$; since $w=\max\W_R$, this gives $\overline{v}(\mu^0_C)=w$. Let $(C',\sigma',w')$ be a coalition of $G\vert_{R'}$ with $w'=\max\W_{R'}$; by \cref{lem:pooling} in $G\vert_{R'}$, there is a message $m$ with $P(m)\cap R'=C'$ and $\overline{v}(\mu^0_{C'})=w'$.
    Apply M-C in $G\vert_R$ to $m''$ and $m$: there is a message $m'$ with
    \[
        P(m')\cap R
        =
        \bigl(P(m'')\cap R\bigr)\cup\bigl(P(m)\cap R\bigr)
        =
        C\cup\bigl(C'\cup(P(m)\cap C)\bigr)
        =
        C\cup C'
        =:
        \widetilde C,
    \]
    using $P(m)\cap R=(P(m)\cap R')\cup(P(m)\cap C)=C'\cup(P(m)\cap C)$. Since $C'$ is non-empty and disjoint from $C$, we get $C\subsetneq\widetilde C$. By \cref{remark:coalitions-exist} in $G\vert_R$, $(\widetilde C,\sigma^{m'},\widetilde w)$ is a coalition of $G\vert_R$ for every $\widetilde w\in V(\mu^0_{\widetilde C})$.
    Since $C$ and $C'$ are disjoint and non-empty,
    \[
        \mu^0_{\widetilde C}
        =
        \lambda\,\mu^0_C+(1-\lambda)\,\mu^0_{C'},
        \qquad
        \lambda:=\frac{\mu^0(C)}{\mu^0(\widetilde C)}\in(0,1).
    \]
    By QC,
    \[
        \overline{v}(\mu^0_{\widetilde C})
        \geq
        \min\{\overline{v}(\mu^0_C),\overline{v}(\mu^0_{C'})\}
        =
        \min\{w,w'\}
        =
        w,
    \]
    while $\overline{v}(\mu^0_{\widetilde C})\in\W_R$ gives $\overline{v}(\mu^0_{\widetilde C})\leq w$. Hence $\overline{v}(\mu^0_{\widetilde C})=w$ and $(\widetilde C,\sigma^{m'},w)$ is the required coalition.

    Under QC$^*$, since $\mu^0_C\neq\mu^0_{C'}$ and the QC inequality is strict,  $\overline{v}(\mu^0_{\widetilde C})>w$, a contradiction. Under genericity, $\overline{v}(\mu^0_{\widetilde C})=w=\overline{v}(\mu^0_C)$ contradicts $C\subsetneq\widetilde C$. In either case $\max\W_{R'}>\max\W_R$ is impossible.
\end{proof}

\begin{proof}[Proof of \cref{thm:one}]
For existence, assume M-C and QC. We construct a COE partition recursively. Let $R_1:=\Theta$. Given a non-empty residual set $R_t$, the maximum $\max\W_t$ exists by \cref{lem:remaining-games}(iii). Let $(C_t,\sigma_t,w_t)$ be a coalition of $G_t$ with $w_t=\max\W_t$ and $C_t$ of largest cardinality among the type sets attaining $\max\W_t$. Each step removes a non-empty cell, so the recursion stops and returns a partition $\Pi$ with $w_t=\max\W_t$ at every step.

We claim $w_t\leq w_{t-1}$ for every $t\geq2$, so that $\Pi$ is a COE partition. Suppose not, so $\max\W_t=w_t>w_{t-1}=\max\W_{t-1}$ at some step. Apply \cref{lem:merging} with $R=R_{t-1}$, the coalition $(C_{t-1},\sigma_{t-1},w_{t-1})$, and $R\smallsetminus C_{t-1}=R_t$: since $w_{t-1}=\max\W_{t-1}$ and $\max\W_t>\max\W_{t-1}$, it gives a coalition of $G_{t-1}$ with payoff $w_{t-1}$ and type set $\widetilde C\supsetneq C_{t-1}$. This contradicts the choice of $C_{t-1}$ as a type set of largest cardinality attaining $\max\W_{t-1}$. Hence $\Pi$ is a COE partition, so by \cref{lem:coe-greedy}(i) it is a coalition-proof PBE partition and $\sigma^\Pi$ a coalition-proof PBE strategy; in particular a coalition-proof PBE exists, proving existence.

(i) Assume M-C and QC$^*$, or M-C, QC, and genericity. We show by induction on $t$ that every greedy choice satisfies $w_t=\max\W_t$. Suppose $w_s=\max\W_s$ for every step $s<t$ of a greedy prefix, and let $R_t\neq\varnothing$. We first bound $\max\W_t$ from above by $w_{t-1}$. If $t\geq2$, apply the QC$^*$ case (respectively, the genericity case) of \cref{lem:merging} with $R=R_{t-1}$ and the step-$(t-1)$ coalition. By the inductive hypothesis this coalition attains $\max\W_{t-1}=w_{t-1}$, so the lemma gives $\max\W_t\leq w_{t-1}$. If $t=1$, the bound $\max\W_t\leq w_0=\infty$ is trivial. For the lower bound, \cref{lem:remaining-games}(iii) gives $\max\W_t\geq\max_{\theta\in R_t}\underline{u}(\theta)$. The two bounds place $\max\W_t$ in the greedy constraint set $\W_t\cap[\max_{\theta\in R_t}\underline{u}(\theta),\,w_{t-1}]$. This set is non-empty, and its maximum is $\max\W_t$. So every greedy choice at step $t$ takes $w_t=\max\W_t$, which completes the induction.

Since the greedy constraint set is non-empty at every step, no run of \cref{alg:greedy-partition} halts. Each run removes a non-empty cell, so it terminates in a greedy partition. By \cref{prop:cppbe-characterization}, that greedy partition is a coalition-proof PBE partition.

For the characterization, every coalition-proof PBE partition is greedy by \cref{prop:cppbe-characterization}, so the induction gives $w_t=\max\W_t$, while $w_t\leq w_{t-1}$ is part of the greedy constraint. Conversely, a partition satisfying these two conditions is coalition-optimal; by \cref{lem:coe-greedy}(i) it is greedy, and by \cref{prop:cppbe-characterization} it is a coalition-proof PBE partition.

(ii) Assume M-C, QC, and genericity. By (i), every greedy partition again has $w_t=\max\W_t$ at every step. Let $\Pi$ and $\Pi'$ be greedy partitions; we show by induction on $t$ that they have the same residual sets and cells. Suppose $R_t=R'_t$, which holds for $t=1$. If both are empty, the two partitions both have length $t-1$. Otherwise both continue, with $w_t=w'_t=\max\W_t$. By the last claim of \cref{lem:pooling} applied in $G_t$, the cells satisfy $\overline{v}(\mu^0_{C_t})=\max\W_t=\overline{v}(\mu^0_{C'_t})$, so $C_t=C'_t$ because $G$ is generic, and then $R_{t+1}=R'_{t+1}$. Hence $\Pi$ and $\Pi'$ have the same length, cells, and payoffs. By \cref{prop:cppbe-characterization} the coalition-proof PBE partitions are exactly the greedy partitions, so any two coalition-proof PBE partitions of $G$ have the same length, cells, and payoffs.
\end{proof}

\section[Proofs for Section~\ref{sec:existence-betweenness}]{Proofs for \cref{sec:existence-betweenness}}

\begin{proof}[Proof of \cref{lem:btw-value-id}]
(i) In $G\vert_C$, the prior is the convex combination
\[
    \mu^0_C=\sum_{m\in X(\sigma)}p_\sigma(m)\,\mu_\sigma(\cdot\sep m),
\]
with positive weights on the on-path posteriors. By \cref{C4} and single-valuedness, each posterior in this combination has value $w$. Applying B to that convex combination yields $w\leq v(\mu^0_C)\leq w$.

(ii) For every non-empty $R\subseteq\Theta$, write
\[
    v^*(R):=\max_{\varnothing\neq X\subseteq\Msg_R}v\bigl(\mu^0_{M^{-1}_R(X)}\bigr).
\]
The maximum is over finitely many sets, and each relative preimage is non-empty because every message in $\Msg_R$ is available to some type in $R$.

Every coalition $(C,\sigma,w)$ of $G\vert_R$ has type set $C=M^{-1}_R(X(\sigma))$ by exclusivity and its converse. Hence \cref{lem:btw-value-id}(i), applied in the restricted game, gives $w=v(\mu^0_C)\leq v^*(R)$, so $\max\W_R\leq v^*(R)$.

Let $X^*\subseteq\Msg_R$ attain $v^*(R)$ and set $C^*:=M^{-1}_R(X^*)$. Consider the auxiliary game
\[
    G':=(C^*,X^*,M',\mu^0_{C^*},V_{C^*}),
    \qquad
    M'(\theta):=M(\theta)\cap X^* .
\]
Note that $G'$ is a disclosure game: by construction every type in $C^*$ can send some message in $X^*$, and every message in $X^*$ is available to some type of $C^*$; the remaining assumptions are inherited from $G$. However, $G'$ is not simply a restricted game of $G$,
because its message space is set to $X^*$ rather than $\Msg_{C^*}$, so some messages available to types in $C^*$ are not available in $G'$.
Nevertheless, the definitions and equilibrium results apply to it.

Every disclosure game has a PBE partition with decreasing payoffs (\cref{lem:pbe-existence,prop:pbe-characterization}); applied to $G'$ this gives a partition $\{(\widetilde C_\tau,\widetilde\sigma_\tau,\widetilde w_\tau)\}_{\tau=1}^k$ with $\widetilde w_1\geq\cdots\geq\widetilde w_k$. Each cell is a coalition in a residual game of $G'$, so \cref{lem:btw-value-id} gives $\widetilde w_\tau=v(\mu^0_{\widetilde C_\tau})$. Since the cells partition $C^*$,
\[
    \mu^0_{C^*}
    =
    \sum_\tau\frac{\mu^0(\widetilde C_\tau)}{\mu^0(C^*)}\,\mu^0_{\widetilde C_\tau}
\]
is a convex combination with positive weights. From quasiconvexity of $v$, we get
\[
    v^*(R)=v(\mu^0_{C^*})\leq\max_\tau\widetilde w_\tau=\widetilde w_1.
\]

Next notice that $(\widetilde{C}_1,\widetilde\sigma_1,\widetilde w_1)$ is a coalition of $G\vert_R$. Conditions \cref{C1}, \cref{C2}, and \cref{C4} carry over because $\widetilde\sigma_1$ uses only messages in $M'(\theta)\subseteq M(\theta)$ and induces the same posteriors when viewed in $G\vert_R$. For exclusivity in $G\vert_R$, take $\theta\in R$ with $M(\theta)\cap X(\widetilde\sigma_1)\neq\varnothing$. Since $X(\widetilde\sigma_1)\subseteq X^*$, this type lies in $C^*$; then $M'(\theta)\cap X(\widetilde\sigma_1)\neq\varnothing$, so exclusivity of the first cell in $G'$ places $\theta$ in $\widetilde C_1$. Thus $\widetilde w_1\in\W_R$. Combining the inequalities,
\[
    \widetilde w_1\leq\max\W_R\leq v^*(R)\leq\widetilde w_1,
\]
so $(\widetilde{C}_1,\widetilde\sigma_1,\widetilde w_1)$ is a coalition of $G\vert_R$ paying $v^*(R)$.
\end{proof}

\begin{lemma}\label{lem:btw-merging}
    Assume B. Let $R\subseteq\Theta$ be non-empty, let $(C,\sigma,w)$ be a coalition of $G\vert_R$ with $w=\max\W_R$, and suppose $R':=R\smallsetminus C$ is non-empty and $G\vert_{R'}$ admits a coalition $(C',\sigma',w')$ with $w'>w$. Then
    $M^{-1}_R\bigl(X(\sigma)\cup X(\sigma')\bigr)=C\cup C'$
        and $v(\mu^0_{C\cup C'})=w$.
        
    If in addition B$^*$ holds, no such configuration exists, so
    $ \max\W_{R\smallsetminus C}\leq\max\W_R $  whenever a coalition with type set $C$ attains $\max\W_R$ and $R\smallsetminus C\neq\varnothing$.
\end{lemma}

\begin{proof}
Since preimages distribute over unions, $M^{-1}_R(X(\sigma)\cup X(\sigma'))
    =
    M^{-1}_R(X(\sigma))\cup M^{-1}_R(X(\sigma')) $.
The first term is $C$ by exclusivity and its converse. For the second term, we know $C' \subseteq M^{-1}_R(X(\sigma')) \subseteq C \cup C'$.
The union is therefore exactly $C\cup C'$.

By \cref{lem:btw-value-id}, $v(\mu^0_C)=w$ and $v(\mu^0_{C'})=w'>w$. The sets $C$ and $C'$ are disjoint and non-empty, so
\[
    \mu^0_{C\cup C'}
    =
    \lambda\mu^0_C+(1-\lambda)\mu^0_{C'},
    \qquad
    \lambda:=\frac{\mu^0(C)}{\mu^0(C\cup C')}\in(0,1).
\]
By B, $v(\mu^0_{C\cup C'})\geq \min\{w,w'\}=w$.
On the other hand, $C\cup C'$ is a relative preimage in $G\vert_R$, so \cref{lem:btw-value-id}(ii) gives
$ v(\mu^0_{C\cup C'})\leq v^*(R)=\max\W_R=w$.
Thus $v(\mu^0_{C\cup C'})=w$.

If B$^*$ holds, the inequality becomes strict: $v(\mu^0_{C\cup C'})>w$, contradicting the equality just proved. Therefore no coalition of $G\vert_{R'}$ can pay more than $\max\W_R$.
\end{proof}

Existence under plain betweenness in \cref{thm:two} rests on the following ordering of each level set of $v$.

\begin{lemma}\label{lem:btw-order}
    Assume B and fix $y\in\mathbb{R}$. Write $L:=\{\mu\in\Delta\Theta\sep v(\mu)=y\}$, $L^{++}:=\{v>y\}$, and $L^{--}:=\{v<y\}$. There is a complete and transitive relation $\succeq$ on $L$ such that:
    \begin{enumerate}
        \item[(i)] If $x\in L$, $u\in L^{++}$, $\alpha\in(0,1)$, and $z:=\alpha x+(1-\alpha)u\in L$, then $z\succ x$.
        \item[(ii)] If $\bar\mu\in L$ and $\bar\mu=\sum_\tau\alpha_\tau\mu_\tau$ is a convex combination with positive weights of beliefs $\mu_\tau\in L\cup L^{--}$, then $\mu_\tau\in L$ and $\mu_\tau\succeq\bar\mu$ for some $\tau$.
    \end{enumerate}
\end{lemma}

\begin{proof}
We use the following proper-separation fact. If two non-empty convex sets $A_1,A_2\subseteq\mathbb{R}^\Theta$ are disjoint, then their relative interiors are disjoint, so there is an affine functional $h$ with $h\leq0$ on $A_1$, $h\geq0$ on $A_2$, and $h$ not identically zero on $A_1\cup A_2$ \parencite[Theorem~11.3]{rockafellar1970}.

The order is constructed recursively on convex subsets of the simplex. Fix a non-empty compact convex set $F\subseteq\Delta\Theta$ and write
\[
    S_F:=L\cap F,\qquad
    U_F:=L^{++}\cap F,\qquad
    B_F:=L^{--}\cap F.
\]
Under B, these three sets are convex; so are $L^{+}:=\{v\geq y\}$ and $L^{-}:=\{v\leq y\}$. We prove, by induction on the affine dimension $d$ of $F$, that there is a complete and transitive relation $\succeq_F$ on $S_F$ satisfying (i) and (ii), with $(S_F,U_F,B_F)$ in place of $(L,L^{++},L^{--})$. The lemma is the case $F=\Delta\Theta$. If $S_F=\varnothing$, there is nothing to order, so assume $S_F\neq\varnothing$.

\emph{Base case $d=0$.} Then $F$ is a singleton. Declare its points indifferent. Condition (i) is vacuous, because it would require the singleton to lie in both $S_F$ and $U_F$, and (ii) holds because every belief in the decomposition is the singleton.

\emph{Inductive step.} Let $d\geq1$ and assume the claim for all non-empty compact convex sets of affine dimension below $d$. In Cases 2--4, we construct the order in two layers. First, we assign each $x\in S_F$ a rank $\rho(x)$. We choose the rank so that, for every $r\in\rho(S_F)$, all beliefs in $S_F$ with rank $r$ lie in a compact convex set $F_r\subseteq F$ of affine dimension below $d$. The inductive hypothesis gives an order $\succeq_{F_r}$ on each such slice. We then order beliefs lexicographically:
\[
    x\succeq_F x'
    \quad:\iff\quad
    \rho(x)>\rho(x')
    \ \text{ or }\
    \bigl(\rho(x)=\rho(x')\ \text{ and }\ x\succeq_{F_{\rho(x)}}x'\bigr).
\]
Thus, the rank compares beliefs in different slices, while the inductively constructed order compares beliefs within the same slice. The resulting relation is complete and transitive.

\emph{Case 1: $U_F=B_F=\varnothing$.} Declare indifference; (i) is vacuous and (ii) holds with any $\mu_\tau$, all of which lie in $S_F$.

\emph{Case 2: $U_F\neq\varnothing$, $B_F=\varnothing$.} Separate the weak lower set $L^{-}\cap F$ from the strict upper set $U_F$. Thus there is an affine $h$ with $h\leq0$ on $L^{-}\cap F$, $h\geq0$ on $U_F$, and $h$ not identically zero on their union. Put $\rho:=h$ on $S_F$ and $F_r:=\{x\in F\sep h(x)=r\}$. The functional $h$ is non-constant on $\aff(F)$: if it were constant on $F$, its value would be both non-positive on $S_F$ and non-negative on $U_F$, hence zero, contradicting proper separation. Each slice is therefore lower-dimensional.

For (i), write $z=\alpha x+(1-\alpha)u$ with $x,z\in S_F$ and $u\in U_F$. Then
\[
    h(z)-h(x)=(1-\alpha)\bigl(h(u)-h(x)\bigr)\geq0.
\]
If the inequality is strict, then $z\succ_F x$ by rank. If it is an equality, then $x,z,u\in F_0$, and the inductive version of (i) inside $F_0$ gives $z\succ_F x$.

For (ii), all beliefs in the decomposition lie in $S_F$, because $B_F=\varnothing$. If some component has $h(\mu_\tau)>h(\bar\mu)$, it is ranked above $\bar\mu$. Otherwise every component has $h(\mu_\tau)\leq h(\bar\mu)$, and the affine identity
\[
    \sum_\tau\alpha_\tau h(\mu_\tau)=h(\bar\mu)
\]
forces all components to have the same value $r:=h(\bar\mu)$. They all lie in $F_r$, so the inductive version of (ii) in $F_r$ applies.

\emph{Case 3: $U_F=\varnothing$, $B_F\neq\varnothing$.} Separate the strict lower set $B_F$ from the weak upper set $L^{+}\cap F$. Let $h$ be affine with $h\leq0$ on $B_F$, $h\geq0$ on $L^{+}\cap F$, and not identically zero on their union. Define $\rho$ and the slices as in Case 2. The same argument shows that $h$ is non-constant on $\aff(F)$. Condition (i) is vacuous because $U_F=\varnothing$.

For (ii), at least one component of the decomposition lies in $S_F$. If all components lay in $B_F$, the quasiconvex half of B would give
\[
    v(\bar\mu)\leq\max_\tau v(\mu_\tau)<y,
\]
contradicting $\bar\mu\in S_F$. Let $r:=h(\bar\mu)$. If some component in $S_F$ has $h(\mu_\tau)>r$, it works. Otherwise every term in
\[
    \sum_\tau\alpha_\tau\bigl(h(\mu_\tau)-r\bigr)=0
\]
is non-positive: this is true for $\mu_\tau\in S_F$ by assumption and for $\mu_\tau\in B_F$ because $h(\mu_\tau)\leq0\leq r$. Hence all terms vanish, every component lies in $F_r$, and the inductive version of (ii) in $F_r$ applies.

\emph{Case 4: $U_F\neq\varnothing$ and $B_F\neq\varnothing$.} When both strict sides are non-empty, the two properties require different separators: one controls mixtures with beliefs in $U_F$ for property~(i), while the other controls decompositions involving beliefs in $B_F$ for property~(ii). Take affine separators $a$ and $b$, each proper in the sense above, such that
\[
    a\leq0\text{ on }L^{-}\cap F,\quad a\geq0\text{ on }U_F,
    \qquad
    b\leq0\text{ on }B_F,\quad b\geq0\text{ on }L^{+}\cap F.
\]
On $S_F$ we have $a\leq0\leq b$. For $\lambda\in[0,1]$, put
\[
    h_\lambda:=\lambda a+(1-\lambda)b.
\]
Then $h_\lambda\geq0$ on $U_F$ and $h_\lambda\leq0$ on $B_F$. Call $x\in S_F$ \emph{degenerate} if $a(x)=b(x)=0$; otherwise $b(x)-a(x)>0$. Define
\[
    \rho(x):=\frac{b(x)}{b(x)-a(x)}\in[0,1]\ \text{ for non-degenerate }x,
    \qquad
    \rho(x):=1\ \text{ for degenerate }x.
\]
A direct calculation gives, for non-degenerate $x\in S_F$,
\[
    h_\lambda(x)=\bigl(\rho(x)-\lambda\bigr)\bigl(b(x)-a(x)\bigr).
\]
Thus the sign of $h_\lambda(x)$ is the sign of $\rho(x)-\lambda$, while every degenerate point has $h_\lambda(x)=0$ for all $\lambda$. Let $ F_\lambda:=\{x\in F\sep h_\lambda(x)=0\}$.
Every point of $S_F$ with rank $\lambda$ lies in $F_\lambda$, with degenerate points placed in $F_1$.

We now check that every relevant slice is lower-dimensional. This is the only place where semicontinuity is used. Because $v$ is upper semicontinuous, $L^{--}$ is relatively open in $\Delta\Theta$; since $B_F$ is non-empty, $\aff(B_F)=\aff(F)$. Fix $\lambda$ with $F_\lambda\cap S_F\neq\varnothing$. If $h_\lambda$ were constant on $F$, then evaluating it at a point of $F_\lambda$ would make the constant zero. For $\lambda\in(0,1)$, the equality $\lambda a+(1-\lambda)b=0$ on $B_F$, together with $a\leq0$ and $b\leq0$ there, would force $a=b=0$ on $B_F$. Since $\aff(B_F)=\aff(F)$, $b$ would vanish on all of $\aff(F)$, contradicting proper separation for $b$. For $\lambda=1$, the same argument uses proper separation for $a$; for $\lambda=0$, it uses proper separation for $b$. Hence $h_\lambda$ is non-constant on $\aff(F)$, and $F_\lambda$ has affine dimension below $d$.

For (i), put $\lambda:=\rho(x)$. Then $h_\lambda(x)=0$ and
\[
    h_\lambda(z)=(1-\alpha)h_\lambda(u)\geq0.
\]
If $h_\lambda(z)>0$, then $z$ has rank strictly above $\lambda=\rho(x)$, so $z\succ_F x$. If $h_\lambda(z)=0$, then $x,z,u\in F_\lambda$. Any strict rank increase already proves the claim; otherwise the tie is resolved inside $F_\lambda$, where the inductive version of (i) gives $z\succ_F x$.

For (ii), first note as in Case 3 that some component lies in $S_F$. Put $\bar\lambda:=\rho(\bar\mu)$, so $h_{\bar\lambda}(\bar\mu)=0$. If some component in $S_F$ has rank above $\bar\lambda$, it works. Otherwise every component in $S_F$ has non-positive $h_{\bar\lambda}$ value, and every component in $B_F$ has non-positive $h_{\bar\lambda}$ value as well. Since
\[
    \sum_\tau\alpha_\tau h_{\bar\lambda}(\mu_\tau)=h_{\bar\lambda}(\bar\mu)=0,
\]
all these values are zero, so every component lies in $F_{\bar\lambda}$. The inductive version of (ii) in $F_{\bar\lambda}$ gives a component $\mu_{\tau^*}\in S_{F_{\bar\lambda}}$ with $\mu_{\tau^*}\succeq_{F_{\bar\lambda}}\bar\mu$. This component is weakly above $\bar\mu$ in the order on $F$, either by a higher rank or by the tie-breaker in $F_{\bar\lambda}$. This completes the induction.
\end{proof}

\begin{proof}[Proof of \cref{thm:two}]
For existence, we construct the desired partition recursively. Let $R_1:=\Theta$. Given a non-empty residual set $R_t$, set $w_t:=\max\W_t$; this maximum exists and is attained by \cref{lem:btw-value-id}(ii). Among the coalitions of $G_t$ attaining $w_t$, choose one with a maximal prior under the relation from \cref{lem:btw-order} for the level set $\{v=w_t\}$. Such a coalition exists because only finitely many type sets, and hence only finitely many priors, can arise. Call the selected coalition $(C_t,\sigma_t,w_t)$ and let $R_{t+1}:=R_t\smallsetminus C_t$. Since each step removes a non-empty cell, the recursion returns a partition $\Pi$ with $w_t=\max\W_t$ at every step.

It remains to show that $\{ w_t \}$ is non-increasing. Suppose, toward a contradiction, that $w_{s+1}>w_s$ for some step $s$. Write
\[
    y:=w_s,\qquad
    L:=\{v=y\},\qquad
    L^{++}:=\{v>y\},\qquad
    L^{--}:=\{v<y\}.
\]
By \cref{lem:btw-value-id}, $\mu^0_{C_s}\in L$ and $\mu^0_{C_{s+1}}\in L^{++}$. Since the step-$s$ coalition attains $\max\W_s$, \cref{lem:btw-merging} applied in $G_s$ gives
\[
    D:=C_s\cup C_{s+1}=M^{-1}_{R_s}\bigl(X_s\cup X_{s+1}\bigr)
    \qquad\text{and}\qquad
    v(\mu^0_D)=y,
\]
and $\mu^0_D=\alpha\mu^0_{C_s}+(1-\alpha)\mu^0_{C_{s+1}}$ with $\alpha:=\mu^0(C_s)/\mu^0(D)\in(0,1)$.

Now form the following auxiliary game:
\[
    G'
    :=
    \bigl(D,\,
    X_s\cup X_{s+1},\,
    \theta\mapsto M(\theta)\cap(X_s\cup X_{s+1}),\,
    \mu^0_D,\,
    V_D\bigr).
\]
Note that $G'$ is a disclosure game by the same argument as in the proof of \cref{lem:btw-value-id}(ii): every type in $D$ can send some message in $X_s\cup X_{s+1}$ because $D$ is the relative preimage, and every message in $X_s\cup X_{s+1}$ is on path in one of the two coalitions and hence is available to some type in $D$.

By \cref{lem:pbe-existence} and the strict-payoff form of \cref{prop:pbe-characterization}, $G'$ admits a PBE partition
$  \{(\widetilde C_\tau,\widetilde\sigma_\tau,\widetilde w_\tau)\}_{\tau=1}^k$
with $    \widetilde w_1>\cdots>\widetilde w_k$.
\Cref{lem:btw-value-id} gives $\widetilde w_\tau=v(\mu^0_{\widetilde C_\tau})$ for every $\tau$, and the cells partition $D$, so 
   $ \mu^0_D=\sum_\tau\beta_\tau\mu^0_{\widetilde C_\tau}
    $ with all $\beta_\tau>0$.
As in the proof of \cref{lem:btw-value-id}(ii), the first cell $(\widetilde C_1,\widetilde\sigma_1,\widetilde w_1)$ is a coalition of the original residual game $G_s$. Hence $\widetilde w_1\leq\max\W_s=y$. On the other hand, the quasiconvex half of B applied to the last convex combination gives
$    y=v(\mu^0_D)\leq\max_\tau\widetilde w_\tau=\widetilde w_1$.
Therefore $\widetilde w_1=y$, and strict decrease gives $\widetilde w_\tau<y$ for all $\tau\geq2$. Equivalently, $\mu^0_{\widetilde C_1}\in L$ and $\mu^0_{\widetilde C_\tau}\in L^{--}$ for $\tau\geq2$.

Apply \cref{lem:btw-order} at level $y$. Property (i), with $x=\mu^0_{C_s}$, $u=\mu^0_{C_{s+1}}$, and $z=\mu^0_D$, gives
$  \mu^0_D\succ\mu^0_{C_s}$.
Property (ii), applied to the decomposition $\mu^0_D=\sum_\tau\beta_\tau\mu^0_{\widetilde C_\tau}$, gives a cell on the level set whose prior is weakly above $\mu^0_D$. Since the auxiliary partition has strictly decreasing payoffs, the only cell on level $y$ is the first one. Thus $  \mu^0_{\widetilde C_1}\succeq\mu^0_D\succ\mu^0_{C_s}$.
But $(\widetilde C_1,\widetilde\sigma_1,y)$ is a coalition of $G_s$ attaining $\max\W_s$, so this contradicts the choice of $(C_s,\sigma_s,w_s)$ as maximal in the level-set order among payoff-maximal coalitions at step $s$.

Hence $w_t\leq w_{t-1}$ at every step. The constructed partition is a COE partition (\cref{dfn:coe}), so by \cref{lem:coe-greedy}(i) it is greedy. By \cref{prop:cppbe-characterization}, it is a coalition-proof PBE partition. Therefore a coalition-proof PBE exists.

\emph{Part (i).} Assume B$^*$ holds or $G$ is generic. We first prove the bound on residual maxima that drives the argument. Let a coalition of $G\vert_R$ with type set $C$ attain $\max\W_R$, and let $R':=R\smallsetminus C$ be non-empty. We claim that
$  \max\W_{R'}\leq\max\W_R $.

Under B$^*$ this is exactly the strict-betweenness conclusion of \cref{lem:btw-merging}. Under genericity, suppose instead that $G\vert_{R'}$ admits a coalition $(C',\sigma',w')$ with $w'>w:=\max\W_R$. By \cref{lem:btw-merging}, the union $C\cup C'$ is a relative preimage in $G\vert_R$ and $v(\mu^0_{C\cup C'})=w$.
But \cref{lem:btw-value-id} gives $v(\mu^0_C)=w$, contradicting genericity. This proves the bound.

We show by induction that every greedy choice takes the unconstrained maximum in the current residual game. Suppose this holds before step $t$, and let $R_t\neq\varnothing$. If $t\geq2$, the step-$(t-1)$ coalition attains $\max\W_{t-1}$ by the induction hypothesis, so the bound above gives
$ \max\W_t\leq w_{t-1}$.
 As for the lower bound, \cref{lem:remaining-games}(iii) gives
$    \max\W_t\geq\max_{\theta\in R_t}\underline u(\theta)$.
Thus $\max\W_t \in \bigl[\max_{\theta\in R_t}\underline u(\theta),\,w_{t-1}\bigr]$,
whence the greedy choice at step $t$ has $w_t=\max\W_t$, completing the induction.

The constraint set is non-empty at every step, so no run of \cref{alg:greedy-partition} halts. Each run removes a non-empty cell and therefore terminates in a greedy partition. By \cref{prop:cppbe-characterization}, every such partition is a coalition-proof PBE partition.

For the characterization, if a partition satisfies $w_t=\max\W_t\leq w_{t-1}$ at every step, then it is a coalition-optimal partition in the sense of \cref{dfn:coe}. By \cref{lem:coe-greedy}(i) it is greedy, and by \cref{prop:cppbe-characterization} it is a coalition-proof PBE partition. Conversely, any coalition-proof PBE partition is greedy. The induction above gives $w_t=\max\W_t$ at every step, and the inequality $w_t\leq w_{t-1}$ is part of the greedy constraint.

\emph{Part (ii).} Let $\Pi$ and $\Pi'$ be coalition-proof PBE partitions of a generic game. They are both greedy by \cref{prop:cppbe-characterization}. We prove by induction on $t$ that they have the same residual sets, payoffs, and cells. At $t=1$ the residual sets and previous payoffs agree. Suppose they agree at step $t$. If the common residual set is empty, both partitions have length $t-1$. Otherwise both continue, and greediness gives the same payoff:
\[
    w_t=w'_t
    =
    \max\Bigl(\W_t\cap\bigl[\max_{\theta\in R_t}\underline u(\theta),\,w_{t-1}\bigr]\Bigr).
\]
By \cref{lem:btw-value-id},
$    v(\mu^0_{C_t})=w_t=w'_t=v(\mu^0_{C'_t})$,
so genericity gives $C_t=C'_t$. Hence $R_{t+1}=R'_{t+1}$, and the induction continues.
\end{proof}

\begin{proof}[Proof of \cref{prop:HKP}]
Fix a truth-leaning equilibrium $(\sigma,\mu,r)$ and let $\Pi$ be its associated PBE partition from \cref{prop:pbe-characterization}.

If message $\theta'$ is on path, type $\theta'$ sends it with probability one. To see this, choose a sender type $\theta$ with $\sigma(\theta'\sep\theta)>0$. Sequential rationality gives $u(\theta)=r(\theta')$. Since $\theta$ can send message $\theta'$, transitivity gives $M(\theta')\subseteq M(\theta)$. Together with reflexivity, this yields
\[
    r(\theta')\leq u(\theta')\leq u(\theta)=r(\theta').
\]
Hence $r(\theta')=u(\theta')$, and (A0) gives $\sigma(\theta'\sep\theta')=1$. If type $\theta$ does not send its own message with probability one, that message is therefore off path, and (P0) and (A0) give $v(\delta_\theta)=r(\theta)<u(\theta)$.

Suppose $(\widetilde C,\widetilde\sigma,\widetilde w)$ blocks $\Pi$. By \cref{dfn:blocking} and exclusivity,
\[
    \widetilde C=M^{-1}_{\Theta(\widetilde w)}(X(\widetilde\sigma)),
    \qquad u(\theta)<\widetilde w\quad\text{for every }\theta\in\widetilde C,
\]
and \cref{lem:btw-value-id} gives $v(\mu^0_{\widetilde C})=\widetilde w$. Let
$ \widetilde C^{\mathrm{tr}}
    :=\{\theta'\in\widetilde C:\sigma(\theta'\sep\theta')=1\}$
be the members who tell the whole truth in the original equilibrium.

Fix a truth-telling member $\theta'\in\widetilde C^{\mathrm{tr}}$. Every actual type $\theta$ that sends message $\theta'$ in equilibrium must also belong to $\widetilde C$. Indeed, $u(\theta)=r(\theta')=u(\theta')<\widetilde w$, so $\theta\in\Theta(\widetilde w)$. Type $\theta'$ sends some message $m\in M(\theta')\cap X(\widetilde\sigma)$ in the proposed deviation; transitivity makes this message available to type $\theta$ as well, and exclusivity therefore includes $\theta$ in $\widetilde C$. Thus the entire equilibrium pool at message $\theta'$ consists of coalition members. Its posterior $\mu_\sigma(\cdot\sep\theta')$ is supported on a subset of $\widetilde C$ and has value $r(\theta')=u(\theta')<\widetilde w$.

Decompose $\mu^0_{\widetilde C}$ into these equilibrium posteriors and the remaining mass at each non-truth-telling member. Let
\[
    \begin{aligned}
        \alpha_{\theta'}&:=\frac{p_\sigma(\theta')}{\mu^0(\widetilde C)}
        &&(\theta'\in\widetilde C^{\mathrm{tr}}),\\
        \beta_\theta&:=\mu^0_{\widetilde C}(\theta)
            \left(1-\sum_{\theta'\in\widetilde C^{\mathrm{tr}}}\sigma(\theta'\sep\theta)\right)
        &&(\theta\in\widetilde C\smallsetminus\widetilde C^{\mathrm{tr}}).
    \end{aligned}
\]
These weights are nonnegative. Since each truth-telling member sends its own message with probability one, Bayes' rule gives the convex decomposition
\[
    \mu^0_{\widetilde C}
    =\sum_{\theta'\in\widetilde C^{\mathrm{tr}}}\alpha_{\theta'}\mu_\sigma(\cdot\sep\theta')
     +\sum_{\theta\in\widetilde C\smallsetminus\widetilde C^{\mathrm{tr}}}\beta_\theta\delta_\theta.
\]
Every posterior in the first sum has value below $\widetilde w$, and every singleton in the second satisfies $v(\delta_\theta)<u(\theta)<\widetilde w$. Quasiconvexity gives $v(\mu^0_{\widetilde C})<\widetilde w$, a contradiction. Thus $\sigma$ is a coalition-proof PBE strategy. The optimal-commitment conclusion follows from Proposition~1 and Theorem~1 of \textcite{hart2017}. 
\end{proof}

\begin{proof}[Proof of \cref{cor:hkp-generic}]
Since $G$ is generic, \cref{thm:two}(ii) applies: all coalition-proof PBE partitions of $G$ have the same cells and payoffs. By \cref{prop:HKP}, all truth-leaning equilibria are coalition-proof PBE strategies, and hence share these unique cells and payoffs.
\end{proof}

\section[Proofs for Section~\ref{sec:existence-cheaptalk}]{Proofs for \cref{sec:existence-cheaptalk}}

For the existence proofs in Section \ref{sec:existence-cheaptalk}, we build an auxiliary game from the quasiconcave closure $\overline{v}^{qc}$.

\begin{lemma}\label{lem:closure-game}
    \begin{enumerate}
        \item[(i)] $\overline{v}^{qc}$ is upper semicontinuous and quasiconcave, $\{\overline{v}^{qc}\geq y\}=\mathrm{conv}\{\overline{v}\geq y\}$ for every $y$, and $\overline{v}^{qc}\geq\overline{v}$;
        \item[(ii)] the \emph{qc-closure} of $G$, $G^{qc}:=(\Theta,\Msg,M,\mu^0,V^{qc})$ with $V^{qc}(\mu):=[\min_{\Delta\Theta}\underline{v},\ \overline{v}^{qc}(\mu)]$, is a disclosure game with upper envelope $\overline{v}^{qc}$;
    \end{enumerate}
\end{lemma}

\begin{proof}
(i) $\overline{v}^{qc}$ is the quasiconcave envelope of $\overline{v}$ \parencite{lipnowskiravid}: by their Theorem~2 it is upper semicontinuous and quasiconcave, and by their securability characterization (Corollary~1) $\{\overline{v}^{qc}\geq y\}$ is the closed convex hull of $\{\overline{v}\geq y\}$. The two remaining points are immediate: $\overline{v}^{qc}\geq\overline{v}$ because $p_i=\mu$ is admissible in \cref{dfn:qc-closure}, and that closed convex hull equals $\mathrm{conv}\{\overline{v}\geq y\}$ because $\{\overline{v}\geq y\}$ is compact ($\overline{v}$ upper semicontinuous by \cref{A4:V-uhc}).

(ii) The interval $V^{qc}(\mu)=[\min_{\Delta\Theta}\underline{v},\ \overline{v}^{qc}(\mu)]$ is non-empty and compact, as $\min_{\Delta\Theta}\underline{v}\leq\underline{v}(\mu)\leq\overline{v}(\mu)\leq\overline{v}^{qc}(\mu)$. It is upper hemicontinuous in $\mu$: given an open $U\supseteq V^{qc}(\mu)$, compactness gives $\varepsilon>0$ with $[\min_{\Delta\Theta}\underline{v},\ \overline{v}^{qc}(\mu)+\varepsilon)\subseteq U$, and upper semicontinuity of $\overline{v}^{qc}$ keeps $V^{qc}$ inside $U$ near $\mu$. Hence $G^{qc}$ satisfies \cref{A4:V-uhc}; it inherits \cref{A1:finite-sets}--\cref{A3:message-coverage} from $G$, its upper envelope is $\overline{v}^{qc}$, and its lower envelope is the constant $\min_{\Delta\Theta}\underline{v}$.
%
\end{proof}

\begin{lemma}\label{lem:qc-bound}
    For every non-empty $R\subseteq\Theta$ and every coalition $(C,\sigma,w)$ of $G\vert_R$,
    $        w\leq\overline{v}^{qc}(\mu^0_C)$.
\end{lemma}

\begin{proof}
    For each $m\in X(\sigma)$, condition \cref{C4} gives $w\in V(\mu_\sigma(\cdot\sep m))$, and hence $\overline{v}(\mu_\sigma(\cdot\sep m))\geq w$. The prior $\mu^0_C=\sum_{m\in X(\sigma)}p_\sigma(m)\,\mu_\sigma(\cdot\sep m)$ is a convex combination of these posteriors, so $\mu^0_C\in\mathrm{conv}\{\overline{v}\geq w\}=\{\overline{v}^{qc}\geq w\}$ by \cref{lem:closure-game}(i). Hence $\overline{v}^{qc}(\mu^0_C)\geq w$.
\end{proof}

\begin{lemma}\label{lem:ct-realization}
    Assume M-CT. For every non-empty $R\subseteq\Theta$ and every message $\bar m\in\Msg_R$, the restricted game $G\vert_R$ has a coalition $(C,\sigma,w)$ with type set $C=P(\bar m)\cap R$ and payoff $w=\overline{v}^{qc}(\mu^0_C)$.
\end{lemma}

\begin{proof}
    By \cref{lem:closure-game}(i), $\mu^0_C=\sum_{i=1}^k\lambda_i p_i$ is a convex combination of beliefs $p_i$ with $\overline{v}(p_i)\geq w$, weights $\lambda_i>0$, and $\supp p_i\subseteq C$; by Carath\'eodory's theorem, since $\Delta\Theta$ has dimension $n-1$, we may take $k\leq n$. We construct the coalition in two steps. First, we replace each $p_i$ by a belief $\widetilde p_i$ on the segment from $p_i$ to $\mu^0_C$ with $w\in V(\widetilde p_i)$, keeping $\mu^0_C$ a convex combination of the $\widetilde p_i$. Second, using the cheap-talk copies of $\bar m$, we construct a coalition strategy on $C$ that induces these beliefs.

    \emph{An intermediate belief on each segment.} Fix $i$ and let $\gamma(t)=(1-t)p_i+t\mu^0_C$ for $t\in[0,1]$. Then $\overline{v}(\gamma(0))=\overline{v}(p_i)\geq w$ and $\overline{v}(\gamma(1))=\overline{v}(\mu^0_C)\leq\overline{v}^{qc}(\mu^0_C)=w$. We claim that $w\in V(\gamma(t))$ for some $t$. If not, then every $t$ has $w\notin V(\gamma(t))=[\underline{v}(\gamma(t)),\overline{v}(\gamma(t))]$, so $\overline{v}(\gamma(t))<w$ or $\underline{v}(\gamma(t))>w$. By upper semicontinuity of $\overline{v}$ and lower semicontinuity of $\underline{v}$ (\cref{A4:V-uhc}), these two conditions define open subsets of $[0,1]$; they are disjoint and cover $[0,1]$. The first excludes $0$ and the second excludes $1$, so $0$ belongs to the second and $1$ to the first, and both are non-empty. This splits the connected interval $[0,1]$ into two disjoint non-empty open sets, which is impossible. Hence $w\in V(\gamma(t_i))$ for some $t_i$. Set $\widetilde p_i=\gamma(t_i)$; since $p_i$ and $\mu^0_C$ are supported on $C$, so is $\widetilde p_i$.

    If $t_i=1$ for some $i$, then $\widetilde p_i=\mu^0_C$, so $w\in V(\mu^0_C)$ and \cref{remark:coalitions-exist}, applied to $\bar m$ in $G\vert_R$, already gives a coalition $(C,\sigma^{\bar m},w)$; so assume $t_i<1$ for every $i$. Put $\beta_i:=\lambda_i/(1-t_i)>0$. Then
    \[
        \sum_i\beta_i\widetilde p_i=\sum_i\lambda_i p_i+\Bigl(\sum_i\tfrac{\lambda_i t_i}{1-t_i}\Bigr)\mu^0_C=\Bigl(\sum_i\beta_i\Bigr)\mu^0_C,
    \]
    so, after normalizing the $\beta_i$ to sum to one, $\mu^0_C=\sum_i\beta_i\widetilde p_i$ is a convex combination of the $\widetilde p_i$ with positive weights.

    \emph{Splitting across cheap-talk copies.} By M-CT there are $k$ distinct copies $m_1,\dots,m_k$ of $\bar m$ sharing its preimage, $P(m_i)=P(\bar m)$, so $P(m_i)\cap R=C$ and every type in $C$ can send every copy. The plan is to spread the coalition across the copies so that copy $m_i$ carries exactly the belief $\widetilde p_i$. To do this, let each $\theta\in C$ send copy $m_i$ with probability
    \[
        \sigma(m_i\sep\theta):=\frac{\beta_i\,\widetilde p_i(\theta)}{\mu^0_C(\theta)}\qquad(1\leq i\leq k),
    \]
    which is well-defined since $\mu^0_C(\theta)>0$ by \cref{A2:full-support}. These probabilities sum to one over $i$, because $\sum_i\beta_i\widetilde p_i=\mu^0_C$, so $\sigma$ is a coalition strategy on $C$. Copy $m_i$ is then on path with probability $p_\sigma(m_i)=\sum_{\theta\in C}\mu^0_C(\theta)\,\sigma(m_i\sep\theta)=\beta_i$, so the evidence is exactly $\{m_1,\dots,m_k\}$, exclusive to $C$ (that is, $M^{-1}(X(\sigma))\cap R=C$), and Bayes' rule returns the intended belief, $\mu_\sigma(\cdot\sep m_i)=\widetilde p_i$. Since $w\in V(\widetilde p_i)$ at every copy, every message in the coalition's evidence pays $w$, so $(C,\sigma,w)$ is a coalition of $G\vert_R$.
\end{proof}

\begin{proof}[Proof of \cref{thm:three}]
By \cref{lem:closure-game}, the closure game $G^{qc}$ satisfies M-C and QC, so the recursive construction in the existence part of the proof of \cref{thm:one} yields a COE partition of $G^{qc}$. Write $R_t$ for its residual sets, $G^{qc}_t:=G^{qc}\vert_{R_t}$, and $\W^{qc}_t$ for the coalition payoff set of $G^{qc}_t$, so $w_t=\max\W^{qc}_t$ at every step. By \cref{lem:pooling} applied in $G^{qc}_t$, we may take each cell to be a pooling cell $(C_t,\sigma^{m_t},w_t)$ with $C_t=P(m_t)\cap R_t$ and $w_t=\overline{v}^{qc}(\mu^0_{C_t})$.

For each $t$, \cref{lem:ct-realization} applied to $R_t$ and $m_t$ yields a coalition $(C_t,\widetilde\sigma_t,w_t)$ of $G_t$, with type set $C_t=P(m_t)\cap R_t$ and payoff $w_t=\overline{v}^{qc}(\mu^0_{C_t})$. These cells and residual sets coincide with those of the COE partition of $G^{qc}$, so $\widetilde\Pi:=\{(C_t,\widetilde\sigma_t,w_t)\}_{t=1}^T$ is a partition of $G$, and its payoffs are non-increasing because those of the COE partition of $G^{qc}$ are.

It remains to show $w_t=\max\W_t$ in $G$ at every step; with the non-increasing payoffs established above, $\widetilde\Pi$ is then a COE partition of $G$. The payoff $w_t$ is attained in $G_t$ by $(C_t,\widetilde\sigma_t,w_t)$, so $w_t\in\W_t$. For the upper bound, let $(C,\sigma,w)$ be any coalition of $G_t$. By \cref{lem:qc-bound}, $w\leq\overline{v}^{qc}(\mu^0_C)$, so it suffices to show $\overline{v}^{qc}(\mu^0_C)\leq w_t$. Iterating M-C over the finite evidence $X(\sigma)$, as in the proof of \cref{lem:pooling}, gives a message $m^*$ with $P(m^*)=\bigcup_{m\in X(\sigma)}P(m)$, so $P(m^*)\cap R_t=M^{-1}_t(X(\sigma))=C$. Pooling on $m^*$ in $G^{qc}_t$ is a coalition with payoff $\overline{v}^{qc}(\mu^0_C)$, so $\overline{v}^{qc}(\mu^0_C)\in\W^{qc}_t$ and hence $\overline{v}^{qc}(\mu^0_C)\leq\max\W^{qc}_t=w_t$. Therefore $w\leq w_t$, giving $w_t=\max\W_t$. By \cref{lem:coe-greedy}(i), $\widetilde\Pi$ is a coalition-proof PBE partition, so a coalition-proof PBE exists.
\end{proof}

\begin{proof}[Proof of \cref{cor:pure-ct}]
The message space trivially satisfies M-C and M-CT. Every non-empty $X\subseteq\Msg$ has $M^{-1}(X)=\Theta$, so by exclusivity \cref{C3} and its converse the type set of any coalition is $\Theta$. By \cref{lem:qc-bound}, every coalition payoff is at most $\overline{v}^{qc}(\mu^0)$, and by \cref{lem:ct-realization} (with $R=\Theta$ and any $\bar m$) the payoff $\overline{v}^{qc}(\mu^0)$ is attained, so $\max\W_\Theta=\overline{v}^{qc}(\mu^0)$. By \cref{prop:cppbe-characterization}, a coalition-proof PBE partition is greedy. All this implies a single-cell partition $\{\Theta,\sigma, \overline{v}^{qc}(\mu^0)\}$, which is essentially unique.
\end{proof}

\section[Proofs for Section~\ref{sec:existence-disposal}]{Proofs for \cref{sec:existence-disposal}}


\begin{proof}[Proof of \cref{prop:DEG-conditions}]
Fix a non-empty $R\subseteq\Theta$, a coalition $(C,\sigma,w)$ of $G\vert_R$, and $w'\in[v_{\min},w]$; in each case we produce a coalition $(C,\sigma',w')$ of $G\vert_R$, so that $G$ satisfies \DEG. Write $q_m:=\mu_\sigma(\cdot\sep m)$ for the posterior at each $m\in X(\sigma)$; by \cref{C4}, $\overline{v}(q_m)\geq w\geq w'$, and $\supp q_m\subseteq P(m)\cap C$.

\emph{Case (i).} Free disposal gives $V(q_m)=[v_{\min},\overline{v}(q_m)]$ for each $m\in X(\sigma)$, and $v_{\min}\leq w'\leq w\leq\overline{v}(q_m)$, so $w'\in V(q_m)$. Hence $(C,\sigma,w')$, with the same strategy $\sigma$, is again a coalition of $G\vert_R$: conditions \cref{C1}--\cref{C3} are unchanged and \cref{C4} is the membership just shown. No assumption on the message mapping is used.

\emph{Case (ii).} We first show that revelation aversion implies the following property: for every $q\in\Delta\Theta$ and every $w'\in[v_{\min},\overline{v}(q)]$,
\[
    q\in\mathrm{conv}\bigl\{\mu\in\Delta\Theta\sep \supp\mu\subseteq\supp q,\ w'\in V(\mu)\bigr\}.
\] 
To see why, fix $q\in\Delta\Theta$ and $w'\in[v_{\min},\overline{v}(q)]$. If $w'\in V(q)$, the claim is immediate. Otherwise, for each $\theta_j\in\supp q$, the intermediate-value argument in the proof of \cref{lem:ct-realization}, together with $\underline{v}(\delta_{\theta_j})=v_{\min}\leq w'\leq\overline{v}(q)$, gives existence of $t_j \in (0,1]$ such that $q_j:=(1-t_j)q+t_j\delta_{\theta_j}$ satisfies $w'\in V(q_j)$.
Each $q_j$ is supported in $\supp q$, and
\[
    q=\sum_j\alpha_j q_j \, \, \text{ for } \, \, \,
        \alpha_j:=\frac{q(\theta_j)/t_j}{\sum_k q(\theta_k)/t_k}.
\]

It remains to construct $\sigma'$ that induces a lottery over these beliefs in place of the  $q$ induced by each original message. Group the evidence by preimage: let $P_\rho$ be the distinct preimages among $\{P(m)\sep m\in X(\sigma)\}$, and let $\mu_\rho:=\lambda_\rho^{-1}\sum_{m\,:\,P(m)=P_\rho} p_\sigma(m)\,q_m$ be the coalition's posterior conditional on sending a message with preimage $P_\rho$, where $\lambda_\rho:=\sum_{m\,:\,P(m)=P_\rho} p_\sigma(m)>0$ and $\sum_\rho\lambda_\rho=1$, so that $\mu^0_C=\sum_\rho\lambda_\rho\mu_\rho=\sum_{m\in X(\sigma)}p_\sigma(m)\,q_m$. Each $\mu_\rho$ is a convex combination of at most $n$ beliefs supported on $P_\rho\cap C$ and paying $w'$. Since cheap-talk copies supply at least $n$ messages with preimage $P_\rho$, spreading the coalition across these copies as in the proof of \cref{lem:ct-realization} yields a coalition strategy $\sigma'$ on $C$ whose on-path posteriors are exactly these beliefs; in particular every message in $X(\sigma')$ pays $w'$. Since $X(\sigma')$ realizes the same preimages as $X(\sigma)$, $M^{-1}_R(X(\sigma'))=M^{-1}_R(X(\sigma))\subseteq C$ by exclusivity \cref{C3} for $(C,\sigma,w)$. Hence $(C,\sigma',w')$ is a coalition of $G\vert_R$.
\end{proof}

\begin{proof}[Proof of \cref{thm:DEG-exists}]
Let $\{(C_s,\sigma_s,w_s)\}_{s<t}$ be a greedy prefix with non-empty residual $R_t$; we show the greedy constraint set $\W_t\cap[\max_{\theta\in R_t}\underline{u}(\theta),\,w_{t-1}]$ is non-empty (with $w_0:=\infty$). Write $\overline{w}_t:=\max\W_t$, which exists and satisfies $\overline{w}_t\geq\max_{\theta\in R_t}\underline{u}(\theta)$ by \cref{lem:remaining-games}(iii).

If $\overline{w}_t\leq w_{t-1}$, then $\overline{w}_t$ itself lies in the constraint set. Otherwise $\overline{w}_t>w_{t-1}$, so $t\geq2$. Let $(C,\sigma,\overline{w}_t)$ be a coalition of $G_t$ attaining $\overline{w}_t$. Every coalition payoff lies in the interval $[v_{\min},\,\max_{\Delta\Theta}\overline{v}]$ (see the proof of \cref{lem:coalition-compact}), so $w_{t-1}\in\W_{t-1}$ gives $w_{t-1}\in[v_{\min},\overline{w}_t]$. Because $G$ satisfies \DEG, applied in $G_t$ to this coalition with target $w_{t-1}$, the residual game $G_t$ admits a coalition $(C,\sigma',w_{t-1})$, so $w_{t-1}\in\W_t$. Moreover $w_{t-1}\geq\max_{\theta\in R_{t-1}}\underline{u}(\theta)\geq\max_{\theta\in R_t}\underline{u}(\theta)$ by greediness at step $t-1$ and $R_t\subseteq R_{t-1}$, so $w_{t-1}$ lies in the constraint set.

In either case the constraint set is non-empty; being the intersection of the compact $\W_t$ with a closed set, it has a maximum, which a greedy choice at step $t$ selects. Hence no run of \cref{alg:greedy-partition} halts. Each run removes a non-empty cell, so it terminates in a greedy partition, which by \cref{prop:cppbe-characterization} is a coalition-proof PBE partition.
The final claim follows from \cref{prop:DEG-conditions}.
\end{proof}

\section[Proofs for Section~\ref{sec:existence-tightness}]{Proofs for \cref{sec:existence-tightness}}

\begin{proof}[Proof of \cref{thm:binary}]
Covered in the text.
\end{proof}

\begin{proof}[Proof of \cref{prop:nonexistence}]

Non-existence of coalition-proof PBE in \cref{ex:nonexistence} is covered in the text.

\emph{Cheap-talk copies.} Augment $\Msg$ with three copies of each message. Coalition type sets still lie in $\{\{1,2\},\{1,3\},\Theta\}$, as preimages $P(m)$ are unchanged. Note that $v$ is quasiconcave, so for any coalition $(C,\sigma,w)$,  $w\leq v(\mu^0_C)$ (\cref{lem:pooling}). Then a coalition with $C=\{1,2\}$ yields payoff $w \leq v\left(\frac{1}{2},\frac{1}{2},0\right)=0.5$, while one with $C=\{1,3\}$ yields payoff $w \leq v\left(\frac{1}{2},0,\frac{1}{2}\right)=0.4$. A coalition with $C=\Theta$ must include at least one copy of $a$ with induced belief $(\mu_1,\mu_2,0)$ with $\mu_1 \leq \mu_2$, hence payoff $w \leq 0.5$; $w=0.5$ is only attainable if every $a$-copy has $\mu_1=\mu_2$, and hence $b$-copies are used only by type 3, yielding $w=0.9$, a contradiction. So every greedy run again selects $C_1=\{1,2\}$ at $w_1=\tfrac12$; the residual coalition payoffs are $V(\delta_3)=\{0.9\}$, and every run halts.


\emph{Evidence structure.} Augment $\Msg$ with a message $c$ and set $M(1)=\{a,b,c\}$, keeping $M(2)=\{a\}$ and $M(3)=\{b\}$. Relabeling $a$, $b$, $c$ as $2$, $3$, $1$ gives $\Msg=\Theta$, $M(1)=\Theta$, $M(2)=\{2\}$, and $M(3)=\{3\}$, which satisfies (L1) and (L2). Let $(C,\sigma,w)$ be a coalition of $G$. If $c\notin X(\sigma)$, then $(C,\sigma,w)$ is a coalition of the original game, so $w\in\{0.4,\tfrac12\}$, and $w=\tfrac12$ only if $C=\{1,2\}$. If $c\in X(\sigma)$, the posterior belief after $c$ is $\delta_1$, so $w=v(\delta_1)=-0.1$ by \cref{C4}. 
Then $\max\W_\Theta$ remains equal to $\tfrac12$, attained only by coalitions with type set $\{1,2\}$. 
After such a coalition is removed, the residual game on $\{3\}$ has coalition payoff set $V(\delta_3)=\{0.9\}$, so every run halts.
\end{proof}

\section[Proofs for Section~\ref{sec:rich-games}]{Proofs for \cref{sec:rich-games}}

\begin{lemma}\label{lem:attainability}
    Let $\nu\colon\Theta\to\mathbb{R}_{\ge0}$ be nonzero and let $\mu\in\Delta_{\supp\nu}$. Then there is a nonzero $\eta\le\nu$ with $\hat\eta=\mu$ and $\supp(\nu-\eta)\subsetneq\supp\nu$.
\end{lemma}
\begin{proof}
    Set $\lambda:=\min_{\theta:\,\mu(\theta)>0}\nu(\theta)/\mu(\theta)$ and $\eta:=\lambda\mu$. Since $\mu\ge0$ and $\lambda>0$, we have $\eta\ge0$. Then $\eta\le\nu$, $\eta\neq0$, and $\hat\eta=\mu$. Because $\eta\ge0$, we have $\supp(\nu-\eta)\subseteq\supp\nu$. If $\theta'$ attains the minimum, then $\eta(\theta')=\nu(\theta')$, so $(\nu-\eta)(\theta')=0$ and $\theta'$ leaves the support after $\eta$ is removed; hence the inclusion is strict.
\end{proof}

\begin{proof}[Proof of \cref{prop:cutecase}]
By \cref{lem:attainability}, the greedy increment induces $\mu^*_{\supp\nu}$, attains the maximum payoff over $\Delta_{\supp\nu}$, and removes at least one type from the support, proving (i). Iterating from $\nu_0=\mu^0$ therefore exhausts the prior in at most $n$ steps. The selected payoffs are weakly decreasing because the feasible faces shrink with the support, proving (ii). For (iii), if maximizers are unique, any partition satisfying (i) must induce $\mu^*_{\supp\nu_{t-1}}$ at step $t$. Since the step also retires a type, feasibility forces its scale to be $\lambda^*(\nu_{t-1})$, so the increment is the greedy increment. Induction from $\nu_0=\mu^0$ gives uniqueness. For coalition proofness, suppose a coalition $\widetilde\eta\leq\mu^0$ blocks, and let $t$ be the smallest index with $\overline v(\hat\eta_t)<\overline v(\hat{\widetilde\eta})$. By (ii), payoffs are weakly decreasing, so the blocking condition gives $\widetilde\eta\leq\sum_{s\geq t}\eta_s=\nu_{t-1}$ and hence $\hat{\widetilde\eta}\in\Delta_{R_t}$. Therefore $\overline v(\hat{\widetilde\eta})\leq\overline v(\mu^*_{R_t})=\overline v(\hat\eta_t)$, a contradiction. 
\end{proof}

\medskip
The remaining proofs maintain the uniqueness hypothesis of \cref{prop:cutecase}(iii), as in the tent subsection of \cref{sec:rich-games}: $\overline{v}$ has a unique maximizer $\mu^*_C$ over each face $\Delta_C$.

\begin{remark}\label{lem:peak-consistency}
    Let $\varnothing\neq D\subseteq C\subseteq\Theta$. If $\supp\mu^*_C\subseteq D$, then $\mu^*_D=\mu^*_C$.
\end{remark}

\begin{proof}[Proof of \cref{prop:tent}]
 Under the coalition-proof mass partition of \cref{prop:cutecase} with prior $\mu^0$, the senders in coalition $\eta_t$ have mass $\eta_t(\Theta)$ and each obtains $\overline v(\hat\eta_t)$, so the sender's ex-ante payoff is $\sum_t\eta_t(\Theta)\,\overline v(\hat\eta_t)=\vtent(\mu^0)$.
For a nonzero sub-measure $\nu$, the greedy construction of \cref{prop:cutecase} run from $\nu$ produces increments $\eta_t$ with $\hat\eta_t=\mu^*_{\supp\nu_{t-1}}$; call $\sum_t\eta_t(\Theta)\,\overline{v}(\hat\eta_t)$ the \emph{greedy value} of $\nu$. Existence of the representation in \cref{prop:tent} follows by running this construction from $\mu$. The increments sum to $\mu$, their masses sum to one, and their induced beliefs are the constrained-optimal beliefs of the successive residual supports. If a step exhausts several types, add the omitted intermediate sets with zero weights; the resulting sets form a deletion order $R_1,\ldots,R_n$. If $\supp\mu\neq\Theta$, begin the deletion order by removing the types outside $\supp\mu$, with zero weights. We prove the following claim: for every family of non-empty sets $E_0\supseteq E_1\supseteq\cdots\supseteq E_K$ and weights $\alpha_0,\dots,\alpha_K>0$ with $\nu=\sum_k\alpha_k\,\mu^*_{E_k}$, the value $\sum_k\alpha_k\,\overline{v}(\mu^*_{E_k})$ equals the greedy value of $\nu$. For the value claim, drop the zero-weight terms from any representation based on a deletion order. The remaining terms form a family of the kind just described.

First normalize the family. Let $S:=\supp\nu$. For every $k$, $\alpha_k\mu^*_{E_k}\leq\nu$ gives $\supp\mu^*_{E_k}\subseteq E_k\cap S$, so $\mu^*_{E_k\cap S}=\mu^*_{E_k}$ by \cref{lem:peak-consistency}; replacing each $E_k$ by $E_k\cap S$ changes neither the summands nor the values, so assume $E_k\subseteq S$ for every $k$. Then $S=\supp\nu\subseteq\bigcup_k\supp\mu^*_{E_k}\subseteq E_0$ forces $E_0=S$. Merging consecutive equal sets, whose constrained-optimal beliefs coincide, assume $E_0\supsetneq E_1\supsetneq\cdots\supsetneq E_K$.

We induct on $K$. If $K=0$, then $\nu=\alpha_0\mu^*_S$, so $\lambda^*(\nu)=\alpha_0$: the run retires all of $\nu$ in one step, with greedy value $\alpha_0\,\overline{v}(\mu^*_S)$. If $K\geq1$, the first greedy increment is $\lambda^*(\nu)\,\mu^*_S$ with $\lambda^*(\nu)=\min_{\theta\in\supp\mu^*_S}\nu(\theta)/\mu^*_S(\theta)$. On the one hand, $\nu-\alpha_0\mu^*_S=\sum_{k\geq1}\alpha_k\mu^*_{E_k}\geq0$, so $\lambda^*(\nu)\geq\alpha_0$. On the other hand, $\supp\mu^*_S\not\subseteq E_1$: otherwise every summand of $\nu$ would be supported in $E_1$, giving $S=\supp\nu\subseteq E_1\subsetneq S$. Any $\theta^*\in\supp\mu^*_S\smallsetminus E_1$ has $\mu^*_{E_k}(\theta^*)=0$ for $k\geq1$, so $\nu(\theta^*)=\alpha_0\,\mu^*_S(\theta^*)$ and $\lambda^*(\nu)\leq\alpha_0$. Hence the first increment is exactly $\alpha_0\mu^*_S$, contributing $\alpha_0\,\overline{v}(\mu^*_S)$, and the run continues from $\nu_1=\sum_{k\geq1}\alpha_k\mu^*_{E_k}$, whose greedy value is $\sum_{k\geq1}\alpha_k\,\overline{v}(\mu^*_{E_k})$ by the induction hypothesis. 

For (i), choose a deletion order containing $C$ and assign weight $1$ to $C$; this representation has value $\overline{v}(\mu^*_C)$, and the path-independence claim just proved gives the same value for every representation of $\mu^*_C$. For (ii), beliefs represented using the same deletion order share the same list $(\mu^*_{R_j})$, so their tent values, which are well defined by the path-independence claim, are the corresponding sums $\sum_j(\cdot)_j\,\overline{v}(\mu^*_{R_j})$, which are linear in the weight vector.

For uniqueness, let $f$ satisfy (i) and (ii), fix $\mu$, and choose a representation $\mu=\sum_j\omega_j\mu^*_{R_j}$ based on a deletion order. Induct on the number of positive weights. If there is only one, then $\mu=\mu^*_{R_j}$ for some $j$, and (i) gives $f(\mu)=\overline{v}(\mu^*_{R_j})=\vtent(\mu)$. If there are at least two, choose an index $J$ with $\zeta:=\omega_J\in(0,1)$ and set $\mu':=\sum_{j\neq J}\frac{\omega_j}{1-\zeta}\mu^*_{R_j}$. Then $\mu'$ has a representation based on the same deletion order with one fewer positive weight and $\mu=\zeta\mu^*_{R_J}+(1-\zeta)\mu'$. By (ii) and (i), $f(\mu)=\zeta\,\overline{v}(\mu^*_{R_J})+(1-\zeta)f(\mu')$, and the same calculation gives $\vtent(\mu)=\zeta\,\overline{v}(\mu^*_{R_J})+(1-\zeta)\vtent(\mu')$. The induction hypothesis gives $f(\mu')=\vtent(\mu')$, so $f(\mu)=\vtent(\mu)$. 
\end{proof}

\clearpage
\section{Technical Lemmas (for online publication)}\label{sec:technical}

The following lemma records the standard microfoundation under which a disclosure game's payoff correspondence $V$ satisfies \cref{A4:V-uhc}: it is induced by a receiver who best-responds to her belief with continuous utilities over a compact action space.
\begin{lemma}\label{lem:best-response}
    Let $A$ be a non-empty compact Hausdorff space, let $u_S\colon A\to\mathbb{R}$ be continuous, and, for each $\theta\in\Theta$, let $u_R(\cdot,\theta)\colon A\to\mathbb{R}$ be continuous. For $\mu\in\Delta\Theta$ and $a\in A$ let $U_R(a,\mu):=\sum_{\theta\in\Theta}\mu(\theta)\,u_R(a,\theta)$ and $a^*(\mu):=\arg\max_{a\in A}U_R(a,\mu)$, and set $V(\mu):=\mathrm{conv}\big(u_S(a^*(\mu))\big)$, with envelopes $\overline v(\mu):=\max_{a\in a^*(\mu)}u_S(a)$ and $\underline v(\mu):=\min_{a\in a^*(\mu)}u_S(a)$, so that $V(\mu)=[\underline v(\mu),\overline v(\mu)]$. Then $V$ is non-empty, compact-valued, convex-valued, and upper hemicontinuous, so it satisfies \cref{A4:V-uhc}; moreover $\overline v$ is upper semicontinuous and $\underline v$ is lower semicontinuous.
\end{lemma}
\begin{proof}
    The map $(a,\mu)\mapsto U_R(a,\mu)$ is continuous and $A$ is compact, so by Berge's maximum theorem \parencite[Theorem~17.31]{aliprantisborder} the best-response correspondence $a^*$ is non-empty, compact-valued, and upper hemicontinuous, and the value functions $\overline v$ and $\underline v$ are upper and lower semicontinuous, respectively. Each $V(\mu)=[\underline v(\mu),\overline v(\mu)]$ is therefore a non-empty compact interval, and the upper hemicontinuity of $V$ follows from the semicontinuity of its two endpoints.
\end{proof}

\begin{lemma}\label{lem:restriction}
    For every disclosure game $G$ and non-empty $R\subseteq\Theta$, the restricted game $G\vert_R$ is itself a disclosure game, and $G\vert_\Theta=G$. Moreover, restriction is transitive: for non-empty $R'\subseteq R\subseteq\Theta$,
    \[
        (G\vert_R)\vert_{R'}=G\vert_{R'}.
    \]
\end{lemma}
\begin{proof}
    The sets $R$ and $\Msg_R$ are non-empty and finite, each $M_R(\theta)=M(\theta)$ is non-empty, $\Msg_R=\bigcup_{\theta\in R}M(\theta)$ gives \cref{A3:message-coverage}, and $\mu^0_R$ has full support on $R$. The zero-extension $\iota_{R,\Theta}$ is continuous, so $V_R=V\circ\iota_{R,\Theta}$ is upper hemicontinuous and takes the same non-empty compact interval values as $V$, giving \cref{A4:V-uhc}. Thus $G\vert_R$ is a disclosure game. For $R=\Theta$ we have $\mu^0(\Theta)=1$ and $\iota_{\Theta,\Theta}=\mathrm{id}$, so every component of $G\vert_\Theta$ equals that of $G$.

    For transitivity, both sides have type space $R'$, message space $\bigcup_{\theta\in R'}M(\theta)$, and message mapping $M$. The priors agree: for $\theta\in R'$,
    \[
        (\mu^0_R)_{R'}(\theta)=\frac{\mu^0_R(\theta)}{\mu^0_R(R')}=\frac{\mu^0(\theta)/\mu^0(R)}{\mu^0(R')/\mu^0(R)}=\mu^0_{R'}(\theta).
    \]
    The payoff correspondences agree because $\iota_{R,\Theta}\circ\iota_{R',R}=\iota_{R',\Theta}$.
\end{proof}

Fix a disclosure game $G=(\Theta,\Msg,M,\mu^0,V)$ and a message $m\in\Msg$. Recall that $P(m)=M^{-1}(\{m\})$ is the set of types that can send $m$, and $F(m)$ is the set of types \emph{forced} to send $m$; note that $F(m)\subseteq P(m)$.

\begin{lemma}\label{lem:feasible-polytope}
    The set of feasible beliefs $\mathcal{F}(m)$ is the set of all $\mu\in\Delta\Theta$ such that
    \begin{enumerate}
        \item[(i)] $\mu(\theta)=0$ for every $\theta\notin P(m)$;
        \item[(ii)] $\mu(\theta)\,\mu^0(\theta')=\mu(\theta')\,\mu^0(\theta)$ for all $\theta,\theta'\in F(m)$;
        \item[(iii)] $\mu(\theta)\,\mu^0(\theta')\leq\mu^0(\theta)\,\mu(\theta')$ for all $\theta\in P(m)\smallsetminus F(m)$ and $\theta'\in F(m)$.
    \end{enumerate}
    In particular, $\mathcal{F}(m)$ is a non-empty compact convex polytope.
\end{lemma}

\begin{proof}
    Write $P:=P(m)$ and $F:=F(m)$; the set $P$ is non-empty by \cref{A3:message-coverage}, and write $\mu^0(F):=\sum_{\theta\in F}\mu^0(\theta)$. Every sender strategy $\sigma$ satisfies $\sigma(m\sep\theta)=0$ for $\theta\notin P$, since $\supp\sigma(\cdot\sep\theta)\subseteq M(\theta)$, and $\sigma(m\sep\theta)=1$ for $\theta\in F$, since $M(\theta)=\{m\}$.

    \emph{Parametrization.} Let $\sigma$ be a strategy with $m\in X(\sigma)$, and put $\alpha_\theta:=\sigma(m\sep\theta)\in[0,1]$ for $\theta\in P\smallsetminus F$. The denominator of Bayes' rule is $D:=p_\sigma(m)=\mu^0(F)+\sum_{\theta\in P\smallsetminus F}\mu^0(\theta)\,\alpha_\theta$, and the induced belief $\mu:=\mu_\sigma(\cdot\sep m)$ is
    \[
        \mu(\theta)=
        \begin{cases}
            \mu^0(\theta)/D, & \theta\in F,\\
            \mu^0(\theta)\,\alpha_\theta/D, & \theta\in P\smallsetminus F,\\
            0, & \theta\notin P.
        \end{cases}
    \]

    \emph{Feasible beliefs satisfy \textup{(i)--(iii)}.} Let $\mu\in\mathcal{F}(m)$, with parameter $\alpha$ and denominator $D$ as above. Condition (i) is immediate. For $\theta,\theta'\in F$, both $\mu(\theta)\,\mu^0(\theta')$ and $\mu(\theta')\,\mu^0(\theta)$ equal $\mu^0(\theta)\,\mu^0(\theta')/D$, so (ii) holds. For $\theta\in P\smallsetminus F$ and $\theta'\in F$,
    \[
        \mu(\theta)\,\mu^0(\theta')=\frac{\mu^0(\theta)\,\alpha_\theta\,\mu^0(\theta')}{D}\leq\frac{\mu^0(\theta)\,\mu^0(\theta')}{D}=\mu^0(\theta)\,\mu(\theta'),
    \]
    since $\alpha_\theta\leq1$, so (iii) holds.

    \emph{Beliefs satisfying \textup{(i)--(iii)} are feasible.} Let $\mu\in\Delta\Theta$ satisfy (i)--(iii). Each $\theta\in P\smallsetminus F$ has $M(\theta)\smallsetminus\{m\}\neq\varnothing$, so it can send $m$ with any chosen probability and place the remaining mass on $M(\theta)\smallsetminus\{m\}$; each $\theta\in F$ sends $m$; each $\theta\notin P$ sends a message in $M(\theta)$, none of which is $m$.

    If $F=\varnothing$, then (i) gives $\supp\mu\subseteq P$. Set $c:=\max_{\theta\in\supp\mu}\mu(\theta)/\mu^0(\theta)>0$ and let each $\theta\in P$ send $m$ with probability $\alpha_\theta:=\mu(\theta)/(c\,\mu^0(\theta))\in[0,1]$. Then $p_\sigma(m)=\sum_{\theta\in P}\mu^0(\theta)\,\alpha_\theta=1/c>0$ and $\mu_\sigma(\theta\sep m)=\mu^0(\theta)\,\alpha_\theta\,c=\mu(\theta)$.

    If $F\neq\varnothing$, then $\mu(\theta')>0$ for every $\theta'\in F$: otherwise (ii) makes $\mu$ vanish on all of $F$, and then (iii) makes $\mu(\theta)\leq0$ for every $\theta\in P\smallsetminus F$, so with (i) the belief $\mu$ would have total mass zero. By (ii) the ratio $\lambda:=\mu(\theta')/\mu^0(\theta')$ is the same positive number for every $\theta'\in F$. Put $\alpha_\theta:=\mu(\theta)/(\lambda\,\mu^0(\theta))$ for $\theta\in P\smallsetminus F$; then $\alpha_\theta\geq0$, and $\alpha_\theta\leq1$ by (iii). Since $\mu$ has total mass one,
    \[
        1=\sum_{\theta\in F}\lambda\,\mu^0(\theta)+\sum_{\theta\in P\smallsetminus F}\lambda\,\mu^0(\theta)\,\alpha_\theta=\lambda\Bigl(\mu^0(F)+\sum_{\theta\in P\smallsetminus F}\mu^0(\theta)\,\alpha_\theta\Bigr),
    \]
    so the denominator $D$ equals $1/\lambda$, and the strategy that sends $m$ with probabilities $\alpha$ induces $\mu$ by the parametrization.

    \emph{Conclusion.} Taking the strategy under which every type in $P$ sends $m$ gives $\mu^0_{P(m)}\in\mathcal{F}(m)$, so $\mathcal{F}(m)$ is non-empty. Conditions (i)--(iii) are finitely many linear equalities and weak inequalities, so they cut the simplex $\Delta\Theta$ down to a bounded polyhedron. Hence $\mathcal{F}(m)$ is a non-empty compact convex polytope.
\end{proof}

\begin{lemma}\label{lem:pbe-existence}
    Every disclosure game admits a PBE.
\end{lemma}
\begin{proof}
We exhibit a triple $(\sigma,\mu,r)$ satisfying the conditions of \cref{dfn:PBE} as a fixed point of a correspondence on the non-empty compact convex set
\[
    K := \Sigma \times \prod_{m\in\Msg}\mathcal{F}(m) \times \prod_{m\in\Msg} I,
    \qquad
    \Sigma := \prod_{\theta\in\Theta}\Delta M(\theta),
    \quad
    I := \Bigl[\min_{\Delta\Theta}\underline{v},\,\max_{\Delta\Theta}\overline{v}\Bigr].
\]
Here $\Sigma$ is the set of sender strategies, a product of simplices; each $\mathcal{F}(m)$ is a non-empty compact convex polytope (\cref{lem:feasible-polytope}); and $I$ is a compact interval, finite because $\underline{v}$ is lower semicontinuous and $\overline{v}$ upper semicontinuous on the compact simplex $\Delta\Theta$ by \cref{A4:V-uhc}, with $V(\mu)=[\underline{v}(\mu),\overline{v}(\mu)]\subseteq I$ for every $\mu\in\Delta\Theta$. Since $\Theta$ and $\Msg$ are finite, $K$ is a subset of a finite-dimensional Euclidean space.

Define $\Phi:K\rightrightarrows K$ by
\[
    \Phi(\sigma,\mu,r) := S(r)\times\prod_{m\in\Msg} B_m(\sigma)\times\prod_{m\in\Msg} V(\mu(\cdot\sep m)),
\]
where
\[
    S(r) := \prod_{\theta\in\Theta}\arg\max_{\tau\in\Delta M(\theta)}\sum_{m\in M(\theta)}\tau(m)\,r(m),
    \qquad
    B_m(\sigma) := \begin{cases}
        \{\mu_\sigma(\cdot\sep m)\}, & m\in X(\sigma),\\
        \mathcal{F}(m), & m\notin X(\sigma).
    \end{cases}
\]
Because a distribution maximizes a linear objective if and only if its support lies in the set of maximizing coordinates,
\begin{equation}\label{eqn:argmax-face}
    \arg\max_{\tau\in\Delta M(\theta)}\sum_{m\in M(\theta)}\tau(m)\,r(m)
    = \bigl\{\tau\in\Delta M(\theta)\sep \supp\tau\subseteq\textstyle\arg\max_{m\in M(\theta)} r(m)\bigr\}.
\end{equation}

\emph{$\Phi$ has non-empty, compact, and convex values.} Fix $(\sigma,\mu,r)\in K$. For each $\theta$, the component $\arg\max_{\tau\in\Delta M(\theta)}\sum_{m}\tau(m)\,r(m)$ of $S(r)$ is non-empty, because the continuous linear map $\tau\mapsto\sum_{m}\tau(m)\,r(m)$ attains its maximum on the non-empty compact simplex $\Delta M(\theta)$; by \eqref{eqn:argmax-face} it is a face of $\Delta M(\theta)$, hence compact and convex. Each $B_m(\sigma)$ is a subset of $\mathcal{F}(m)$: on path it is the singleton $\{\mu_\sigma(\cdot\sep m)\}$, with $\mu_\sigma(\cdot\sep m)\in\mathcal{F}(m)$ by \cref{dfn:set-feasible-beliefs}; off path it is $\mathcal{F}(m)$, which is non-empty, compact, and convex. Finally, $V(\mu(\cdot\sep m))=[\underline{v}(\mu(\cdot\sep m)),\overline{v}(\mu(\cdot\sep m))]$ is a non-empty compact interval contained in $I$ by \cref{A4:V-uhc}. As products of these sets, the values of $\Phi$ are non-empty, compact, and convex.

\emph{$\Phi$ has closed graph.} Let $(\sigma_k,\mu_k,r_k)\to(\sigma,\mu,r)$ in $K$, and let $(\sigma'_k,\mu'_k,r'_k)\in\Phi(\sigma_k,\mu_k,r_k)$ converge to $(\sigma',\mu',r')$. We verify $(\sigma',\mu',r')\in\Phi(\sigma,\mu,r)$ coordinate by coordinate.
\begin{itemize}
    \item \emph{Sender.} Since $\sigma'_k\in S(r_k)$, for every $\theta$ and every $\tau\in\Delta M(\theta)$,
    \[
        \sum_{m\in M(\theta)}\sigma'_k(m\sep\theta)\,r_k(m)\;\geq\;\sum_{m\in M(\theta)}\tau(m)\,r_k(m).
    \]
    Each side is a finite sum of products of convergent sequences, so the inequality is preserved in the limit: $\sum_{m}\sigma'(m\sep\theta)\,r(m)\geq\sum_{m}\tau(m)\,r(m)$ for every $\tau$. Hence $\sigma'\in S(r)$.
    \item \emph{Beliefs.} Fix $m$. If $m\notin X(\sigma)$, then $B_m(\sigma)=\mathcal{F}(m)$, and $\mu'(\cdot\sep m)\in\mathcal{F}(m)$ because $\mu'_k(\cdot\sep m)\in B_m(\sigma_k)\subseteq\mathcal{F}(m)$ for every $k$ and $\mathcal{F}(m)$ is closed. If $m\in X(\sigma)$, then $p_\sigma(m)>0$. The on-path probability $p_\sigma(m)=\sum_{\theta}\mu^0(\theta)\,\sigma(m\sep\theta)$ is continuous in $\sigma$, so $p_{\sigma_k}(m)\to p_\sigma(m)>0$ and hence $m\in X(\sigma_k)$ for all large $k$. For such $k$, $B_m(\sigma_k)=\{\mu_{\sigma_k}(\cdot\sep m)\}$, so $\mu'_k(\cdot\sep m)=\mu_{\sigma_k}(\cdot\sep m)$; and since the induced belief $\mu_\sigma(\theta\sep m)=\mu^0(\theta)\,\sigma(m\sep\theta)/p_\sigma(m)$ is continuous in $\sigma$ wherever $p_\sigma(m)>0$, we get $\mu'_k(\cdot\sep m)\to\mu_\sigma(\cdot\sep m)$. By uniqueness of limits, $\mu'(\cdot\sep m)=\mu_\sigma(\cdot\sep m)\in B_m(\sigma)$.
    \item \emph{Payoffs.} For each $m$, $r'_k(m)\in V(\mu_k(\cdot\sep m))$ with $\mu_k(\cdot\sep m)\to\mu(\cdot\sep m)$ and $r'_k(m)\to r'(m)$. By \cref{A4:V-uhc}, $V$ is upper hemicontinuous with closed values on the compact simplex $\Delta\Theta$, so it has closed graph; therefore $r'(m)\in V(\mu(\cdot\sep m))$.
\end{itemize}

By Kakutani's fixed-point theorem \parencite[Corollary~15.3, p.~72]{border1985}, $\Phi$ has a fixed point $(\sigma,\mu,r)$. It satisfies the four conditions of \cref{dfn:PBE}: $\mu(\cdot\sep m)\in B_m(\sigma)\subseteq\mathcal{F}(m)$ gives feasibility; for $m\in X(\sigma)$, $B_m(\sigma)=\{\mu_\sigma(\cdot\sep m)\}$ gives the on-path Bayesian condition; $r(m)\in V(\mu(\cdot\sep m))$ gives payoff compatibility; and $\sigma\in S(r)$ together with \eqref{eqn:argmax-face} gives sequential rationality. Hence $(\sigma,\mu,r)$ is a PBE.
\end{proof}

\end{document}